\documentclass[11pt,letterpaper]{article}

\usepackage{mathptmx}

\usepackage{latexsym,amssymb}
\usepackage{amsmath}

\usepackage{amsthm}
\usepackage[in]{fullpage}

\usepackage[english]{babel}

\usepackage{graphicx}
\usepackage{url}

\newcommand{\naturals}{\mathbb{N}}

\newcommand{\nullval}{\bot}
\newcommand{\linkstatesgeneric}[2]{\stackrel{#2}{#1}}
\newcommand{\linkstates}[1]{\linkstatesgeneric{\sim}{#1}}

\newcommand{\isets}[1]{\mathsf{iSets}(#1)}

\newcommand{\states}[1]{\mathsf{States}(#1)}
\newcommand{\setofboxes}[1]{\mathbf{Inv}(#1)}
\newcommand{\scvalof}[2]{\mathsf{scval}(#1,#2)}
\newcommand{\regularpath}[0]{{\sf C}-regular }

\newcommand{\laddersubset}[3]{#1_{#2}(#3)}
\newcommand{\laddersubsetX}[2]{\laddersubset{\Lambda}{#1}{#2}}
\newcommand{\laddersubsetY}[2]{\laddersubset{\Omega}{#1}{#2}}

\newcommand{\writesnapshot}{{\sf update}}
\newcommand{\readsnapshot}{{\sf scan}}
\newcommand{\execobject}{exec}

\newcommand{\eventfont}[1]{{\sf #1}}
\newcommand{\event}[2]{\eventfont{#1}(#2)}
\newcommand{\writeop}[1]{\event{W}{#1}}
\newcommand{\readop}[1]{\event{R}{#1}}

\newcommand{\scop}[1]{\event{S}{#1}}

\newcommand{\xschedule}[1]{\xi(#1)}
\newcommand{\xschedulesup}[2]{\xi^{#1}(#2)}

\newcommand{\xschedulewro}[1]{\eta(#1)}

\newcommand{\inthashmapname}{{\bf tup }}
\newcommand{\inthashmap}[2]{\inthashmapname(#1, #2)}

\theoremstyle{plain}
\newtheorem{theorem}{Theorem}[section]
\newtheorem{lemma}[theorem]{Lemma}
\newtheorem{corollary}[theorem]{Corollary}

\theoremstyle{definition}
\newtheorem*{definition*}{Definition}
\newtheorem*{definitions*}{Definitions}
\newtheorem*{example*}{Example}

\theoremstyle{remark}

\newcounter{linecounter}
\newcommand{\linenumbering}{(\arabic{linecounter})}
\renewcommand{\line}[1]{\refstepcounter{linecounter}
\label{#1}
\linenumbering}
\newcommand{\resetline}{\setcounter{linecounter}{0}}

\title{The solvability of consensus in iterated models extended with safe-consensus
\thanks{A Preliminary version of these results appeared in SIROCCO 2014.}}

\author{Rodolfo Conde\\ Instituto Técnológico Autónomo de México\\Rio Hondo \#1 Col. Progreso Tizapán, 
CDMX 01080, México\\ \tt{rodolfo.conde@itam.mx}
\and Sergio Rajsbaum\\Instituto de Matemáticas, Universidad Nacional Autónoma de México\\
Ciudad Universitaria, CDMX 04510, México\\ \tt{rajsbaum@im.unam.mx}}

\date{July 2022}

\begin{document}


\maketitle

\begin{abstract}
The safe-consensus task was introduced by Afek, Gafni and Lieber (DISC' 09) as a weakening of the 
classic consensus.
When there is concurrency, 
the consensus output can be arbitrary, not even the input of any process. They showed that 
safe-consensus is equivalent to consensus, in a wait-free system.
We study the solvability of consensus in three shared memory iterated models extended with 
the power of 
safe-consensus black boxes. 
In the first iterated model, for the $i$-th iteration, 
the processes write to memory, then they snapshot it and 
finally they invoke safe-consensus boxes. We prove that in this model, consensus cannot be 
implemented. In a second iterated model,
processes first 
invoke safe-consensus, then they write to memory and finally they snapshot it.
We show that this model is equivalent to the previous model and thus 
consensus cannot be implemented.
In the last iterated model, processes 
write to the memory, invoke safe-consensus boxes and finally they snapshot
the 
memory. We show that in this model, any wait-free implementation of consensus requires 
$\binom{n}{2}$ safe-consensus black-boxes and this bound is tight.\\

\end{abstract}

\newpage

\section{Introduction}

The ability to agree on a common decision 
is key to distributed computing. The most
widely studied agreement abstraction is  \emph{consensus}.
In the {consensus} task each process proposes a value, and all correct
processes have to decide the same value.
In addition,    \emph{validity} requires   that the decided value 
is a proposed value.

Herlihy's seminal paper~\cite{Herlihy91} examined the power of different synchronization primitives 
for \emph{wait-free computation}, e.g., 
when computation completes in a finite number of steps by a process, 
regardless of how fast or slow other processes run, and even if some of 
them halt permanently. He  showed that consensus is a universal primitive,
in the sense that a solution to consensus (with read/write registers)
can be used to implement any synchronization primitive in a wait-free manner. Also, consensus 
cannot 
be wait-free implemented 
from read/write re\-gis\-ters alone~\cite{fischerImpossCon,impossConSHM}; indeed, all modern 
shared-memory multiprocessors provide some form of universal primitive.


Afek, Gafni and Lieber~\cite{yehudagafni} introduced \emph{safe-consensus},
which  seemed to be a  synchronization primitive much weaker than consensus.
The validity requirement becomes: 
if the first process 
to invoke the task returns before any other process invokes it, 
then it outputs its input; otherwise, when there is concurrency, 
the consensus output can be arbitrary, not even the input of any process. 
In any case, all processes must agree on the same output value.
Trivially, consensus implements safe-consensus.
Surprisingly, they  proved that the converse is also true, by presenting a wait-free implementation 
of consensus using safe-consensus black-boxes and read/write registers. 
Why is it then, that safe-consensus seems
a much weaker synchronization primitive?

\noindent
{\bf Our Results.}
We show that while consensus and safe-consensus are wait-free equivalent,
 any wait-free implementation of consensus  for $n$ processes
 in an iterated model (with the appropriate order of snapshot 
 and safe-consensus operations) 
requires $\binom{n}{2}$ safe-consensus black-boxes, and this bound is tight. 

Our main result is the lower bound.
It uses connectivity arguments  based on  subgraphs of \emph{Johnson graphs},
and an intricate  combinatorial and bivalency argument, that yields a detailed bound on how many  
safe-consensus  objects of each type 
(i.e., which processes and how many processes invoke a safe-consensus object) 
are used by the implementation protocol. 
For the upper bound, we present a simple protocol, based on the new 
 \emph{$g$-2coalitions-consensus} task, which may be of independent interest\footnote{These results 
appeared for the first time in the Proceedings of the 21st International Colloquium on Structural 
Information and Communication 
Complexity \cite{powersafeconsensusconderajsbaum}.}.

We develop our results in an iterated model of computation~\cite{rajsbaum2010iterated}, where the 
processes repeatedly: write their information to a (fresh) shared array, invoke (fresh) 
safe-consensus boxes and  snapshot the contents of the shared array.

Also, we study the solvability of consensus in two alternate iterated models extended with 
safe-consensus. In the first model, the processes write to memory, then they snapshot it and 
finally they invoke safe-consensus boxes. We prove that in this model, consensus cannot be 
implemented from safe-consensus. In the second model, processes first 
invoke safe-consensus, then they write to shared memory and finally they snapshot the contents of 
the memory. We show that this model is equivalent to the previous model and thus consensus cannot 
be solved in this model.

\noindent
{\bf Related Work.}
Distributed computing theory has been concerned from early
on with understanding the relative power of synchronization primitives.
The wait-free context
is the basis to study other failure models e.g.~\cite{BGdistalg},
and there is a characterization of the wait-free, read/write solvable tasks~\cite{HS99}.
For instance, the weakening of consensus,  \emph{set agreement}, where
$n$ processes may agree on at most $n-1$ different input values,
is still not wait-free solvable~\cite{BG93-2,HS99,SZ00} with read/write registers only.
The \emph{renaming} task where $n$ processes have to agree on at most $2n-1$
names has also been studied in detail 
e.g.~\cite{AttiyaBDPR90,Castaneda:2012:ETA:2247370.2247382,Castaneda:2008:NCT:1400751.1400791,Castaneda:2012:NCT:2108242.2108245,CastanedaRRrev2011}.

Iterated models 
e.g.~\cite{rajsbaum2010iterated,borowsky,herlihyAdversaries,herlihyAdversaries-j,impossFDs,RRT08} 
facilitate impossibility results, and (although more restrictive) facilitate 
the analysis of protocols~\cite{gafniRrecursion}.
We follow in this paper the approach of~\cite{GRH06} that used an iterated
model to prove the separation result that set agreement can implement {renaming}
but not vice-versa, and expect our result can be extended to a general 
model using simulations, as was done in~\cite{gafnirajsOPODIS10} for that separation result. 
For an overview of the use of to\-po\-lo\-gy to 
study computability, including the use of iterated models
and simulations see~\cite{HerlihyKR2013}.

Afek, Gafni and Lieber~\cite{yehudagafni} 
presented a wait-free protocol that
implements consensus using $\binom{n}{2}$ safe-consensus black-boxes (and read/write registers). 
Since our implementation uses the iterated form of shared-memory, it
 is easier to prove correct. 
Safe-consensus was used in \cite{yehudagafni} to show that the $g$-tight-group-renaming 
task\footnote{In a
tight group renaming task with group size $g$, 
$n$ processes with ids from a 
large
domain $\{ 1, 2,\ldots,N \}$ are partitioned into 
$m$ groups with ids from a large domain
$\{1, 2,\ldots, M \}$, with at most $g$ processes 
per group. A tight group renaming task
renames groups from the domain $1\ldots M$ to 
$1\ldots l$ for $l \ll M$ , where all processors
with the same initial group ID are renamed to the 
same new group ID, and no
two different initial group ids are renamed to the 
same new group ID.} 
\cite{yehudagamzu} is as powerful as $g$-consensus.

The idea of the classical consensus impossibility result~\cite{fischerImpossCon,impossConSHM}
is  (roughly speaking)  that the executions of
a protocol in such a system can be represented by a graph which is always
connected. The connectivity invariance has been proved in many
papers using the critical state
argument introduced in~\cite{fischerImpossCon}, or sometimes using
a layered analysis as
in~\cite{mosesrajsbaum}. Connectivity can be used also to prove 
 time lower bounds e.g.~\cite{mosesrajsbaum,AttiyaDLS94,Dwork1990156}.
We extend here the layered analysis
to prove the lower bound result on the number of safe-consensus objects needed to implement 
consensus. Also, our results show that when the basic shared memory iterated model is used with 
objects stronger than read/write memory, care has be taken in the way they are added to the model, 
as the resulting power of the model to solve tasks can vary.

 In a previous work~\cite{getco10topologicaltheory} we had already studied an
 iterated model extended with the power of safe-consensus. 
However, that model had the restriction that in each 
iteration, all processes invoke the same safe-consensus object. 
We showed that set agreement can be implemented, but not  consensus. 
The impossibility proof uses much simpler connectivity arguments
than those of this paper. 

\noindent
{\bf Organization.}
The paper is organized as follows. Section~\ref{secdcBasicDef} describes the basic concepts and 
previous results of the models of computation. 
This section can be skipped by readers familiar with standard distributed computing notions. 
Section \ref{secAllModels} defines the three iterated models of computation that we investigate and 
also we present the main results for each model.
Section \ref{secImpossConsensus} is dedicated to our 
first iterated model with safe-consensus. In 
this model, the processes first write to memory, then they 
snapshot the contents of the memory and 
after that, they invoke safe-consensus objects. We prove 
that in this model, consensus cannot be 
implemented. In Section \ref{secEquivWROOWR}, we develop 
our second iterated model extended with 
safe-consensus. In this model, the processes first invoke 
safe-consensus objects, then they write 
to the memory and finally, they snapshot the shared memory. We prove that this model is equivalent 
to the previous model (for task solvability), thus the consensus task cannot be 
implemented in this model.
 Section \ref{secResultsWOR} is devoted to develop 
all the results obtained for our last iterated model, in it the processes write to memory, invoke 
the safe-consensus objects and then they snapshot the shared memory. For this model, our results 
are the following: 
\begin{itemize}
 \item We construct a protocol that solves $n$-process 
consensus using $\binom{n}{2}$ safe-consensus boxes 
(Section~\ref{secTheConsensusProtocol}). We give our 
consensus 
protocol using the new 2coalitions-consensus task, which is described in full detail in Section 
\ref{secCCTask}.
\item We 
describe and prove 
the main result for this iterated model, which is 
also the main result of this paper, it is a lower 
bound on the number of safe-consensus objects needed to 
solve consensus in this iterated model 
(Section \ref{asymtotictighbounds}). 
\end{itemize}
In Section \ref{secConclusions} we give our final  conclusions and some open 
problems. 

\section{Basic definitions}\label{secdcBasicDef}

In this section, we introduce the model of computation and present many basic concepts used in this 
paper. We follow the usual definitions and extend some concepts from \cite{attiyarajsbaum,AW}.

\subsection{Distributed systems}\label{secSystems}

Our formal model is an extension of the standard iterated version \cite{borowsky} of the 
usual read/write shared memory model e.g.~\cite{AW}. 
 A \emph{process} is a deterministic state machine, which has a (possible infinite) set of \emph{local
states}, including a subset called the \emph{initial states} and a subset called the \emph{output
states}.

A \emph{shared object} $\mathsf{O}$ has a domain $D$ of 
input values, and a domain $D^\prime$ of output values. $\mathsf{O}$ provides a unique operation, 
$\mathsf{O}.\execobject(d)$, that receives an input value $d\in D$ and returns an 
output value $d^\prime\in D^\prime$. 

A one-shot snapshot object $A$ is a shared memory 
array $A\left[ 1,\ldots,n\right]$ with one 
entry per process. That array is initialized to 
$\left[\nullval,\ldots,\nullval\right]$, where 
$\nullval$ is a default value that cannot be 
written by a process. The snapshot object $A$ 
provides 
two atomic operations that can be used by a 
process at most once:
\begin{itemize}
\item $A.\writesnapshot(v)$: when called by 
process $p_j$, it writes the value $v$ to the 
register 
$A\left[j\right]$.
\item $A.\readsnapshot ()$: returns a copy of the whole shared memory array $A$.
\end{itemize}

It is also customary to make no assumptions about the size of the registers of the shared memory, 
and therefore we may assume that each process $p_i$ can write its entire local state in a single 
register. Notice that the snapshot operation can be implemented in read/write shared
 memory, according to \cite{AADGMS93,borowskyIASI}.

A \emph{system} consists of the following data:
\begin{itemize}
\item A set of $n \geqslant 1$ processes $\Pi = \{ p_1, \ldots, p_n \}$;
\item a shared memory $SM\left[ i \right]$ ($i \geqslant 0$) structured as an infinite sequence 
of one-snapshot objects;
\item an infinite sequence $\mathsf{S} \left[ j \right]$ ($j \geqslant 0$) of shared objects.
\end{itemize}

A \emph{global state} of a system is a vector $P$ of the form
\[ P=\langle s_1,\ldots, s_n; SM \rangle, \]
where $s_i$ is the local state of process $p_i\in \Pi$ and $SM$ is the shared memory of the system. 
An {\it initial state} is a state in which every local
state is an initial local state and all registers in the shared memory are set to $\nullval$.  A {\it
decision state} is a state in which all local states are output states. When referring to a global 
state $P$, we usually omit the word global and simply refer 
to $P$ as a state.

\subsection{Events and round schedules}

An \emph{event} in the system is performed by a single process 
$p_i\in \Pi$, which applies only one of the following actions: a write (update) operation, denoted by \eventfont{W} or a read (scan) operation, denoted by 
\eventfont{R}, on the shared memory or an invocation to a shared object (\eventfont{S}). 
Any of these operations may be preceded/followed by some local computation, formally a change of 
the process to its next local state. 
We will need 
to consider events performed \emph{concurrently} by the processes. If \eventfont{E} is any event 
and $p_{i_1},\ldots, p_{i_k}\in \Pi$ are processes, then we denote the fact that $p_{i_1},\ldots, 
p_{i_k}$ 
execute concurrently the event \eventfont{E} by $\event{E}{X}$, where $X=\{i_1,\ldots, i_k\}$.

We fix once and for all some notation. Let 
$\overline{n}=\{1,\ldots,n\}$, when convenient, 
we 
will denote $\event{E}{X}$ by $\event{E}{i_1,
\ldots, i_k}$ and if $i\in \overline{n}$ is a 
process 
id, then $\event{E}{\overline{n}-\{ i\}}$ is 
written simply as $\event{E}{\overline{n}-i}$.

A \emph{round schedule} $\pi$ is a finite 
sequence of events of the form 
\[ \pi \colon \event{E_1}{X_1},\ldots,
\event{E_r}{X_r},\]
that encodes the way in which the processes with ids in the set $\bigcup_{j=1}^r X_j$ perform the 
events ${\sf E}_1,\ldots,{\sf E}_r$. For example, the round schedule given by
\[ \writeop{1, 3}, \readop{1,3}, \writeop{2}, \readop{2},\scop{1,2,3} \]
means that processes $p_1,p_3$ perform 
the write and read events concurrently; after that, $p_2$
executes solo its read and write events and finally all three processes invoke the shared objects concurrently. 
Similarly, the round schedule $\writeop{1,2,3},\readop{1,2,3},$ $\scop{1,2,3}$ 
says that $p_1,p_2$ and $p_3$ execute concurrently the write and read events in the shared memory 
and then they invoke the shared objects concurrently.

\subsection{Protocols and executions}

The state machine of each process $p_i\in \Pi$ is called a \emph{local protocol} $\mathcal{A}_i$,
 that determines the steps taken by $p_i$. We assume that all local protocols are identical; i.e.
 Processes have the same state machine. A {\it protocol} is a collection $\mathcal{A}$ 
 of local protocols $\mathcal{A}_1,\ldots,\mathcal{A}_n$.  

For the sake of simplicity, we will give protocols specifications using pseudocode and we establish the
following conventions: A lowercase variable denotes a local variable, with a subindex that indicates
to which process it belongs; the shared memory (which is visible to all processes) is denoted with
uppercase letters. Intuitively, the local state $s_i$ of process $p_i$ is composed of the contents 
of all the local variables of $p_i$. Also, we identify 
two special components of each process' states: an input and an output. It is assumed that initial 
states differ only in the value of the input component; moreover, the input component never changes. 
The protocol cannot overwrite the output, it is initially $\nullval$; once a non-$\nullval$ value is 
written to the output component of the state, it never changes; when this occurs, we say that the 
process \emph{decides}. The output states are those with non-$\nullval$ output values.

Let $\mathcal{A}$ be a protocol. An \emph{execution} of $\mathcal{A}$ is a finite or infinite
alternating sequence of states and round schedules 
\[ S_0,\pi_1,\ldots,S_k,\pi_{k+1},\ldots,\] 
where $S_0$ is an initial state and for each $k\geqslant 1$, $S_k$ is the resulting state of
applying the sequence of events performed by the processes in the way described by the round
schedule $\pi_{k}$. 
An \emph{r-round partial execution} of $\mathcal{A}$ is a
finite execution of $\mathcal{A}$ of the form $S_0,\pi_1,\ldots,S_{r-1},\pi_r,S_r$.

If $P$ is a state, $P$ is said to be \emph{reachable in $\mathcal{A}$} if there exists an $r$-round
partial execution of $\mathcal{A}$ $(r\geqslant 0)$ that ends in the state $P$ and when there is no
confusion about which protocol we refer to, we just say that $P$ is reachable. 

Given the protocol $\mathcal{A}$ and two states $S,R$, we say that $R$ is a
{\it successor} of $S$ in $\mathcal{A}$, if there exists an execution
$\alpha$ of $\mathcal{A}$ such that
\[ \alpha=S_0,\pi_1,\ldots,S_r=S,\pi_{r+1},\ldots,\pi_{r+k},S_{r+k}=R,\ldots, \]
i.e., starting from $S$, we can run the protocol $\mathcal{A}$ $k$ rounds (for some $k\geqslant 0$)
such that the system enters state $R$. 
If $\pi$ is any round schedule and $S$ is a state, the successor of $S$ in $\mathcal{A}$ obtained by
running the protocol (starting in the state $S$) one round with the round schedule $\pi$ is
denoted by $S\cdot \pi$.

\subsection{Decision tasks}

In distributed computing, a \emph{decision task} is a problem that must be solved in a distributed 
system. Each process starts with a private input value, communicates with the others, and halts with 
a private output value. Formally, a \emph{decision task} $\Delta$ is a relation that has a domain 
$\mathcal{I}$ 
of input values and a domain
$\mathcal{O}$ of output values; $\Delta$ specifies for each assignment of the inputs to processes on
which outputs processes can decide. A \emph{bounded} decision task is a task whose number of input 
values is finite.

We also refer to decision task simply as tasks. Examples of tasks includes \emph{consensus} 
\cite{fischerImpossCon}, \emph{renaming} \cite{yehudagamzu,borowskyIASI} and the \emph{set 
agreement} task \cite{setagreementJournal}.

A protocol $\mathcal{A}$ \emph{solves} a decision task $\Delta$ if any finite execution $\alpha$
of $\mathcal{A}$ can be extended to an execution $\alpha^\prime$ in which all processes decide on
values which are allowable (according to $\Delta$) for the inputs in $\alpha$. 
Because the outputs cannot be overwritten, if a process has decided on a value in $\alpha$, it must
have the same output in $\alpha^\prime$. This means that outputs already written by the processes
can be completed to outputs for all processes that are permissible for the inputs in $\alpha$.

 A protocol $\mathcal{A}$ is {\it wait-free} if in any execution of $\mathcal{A}$, a process either
has a finite number of events or decides. This implies that if a process has an infinite
number of events, it must decide after a finite number of events.
Roughly speaking, $\mathcal{A}$ is wait-free if any process that continues to run will halt with an 
output value in a fixed number of steps, regardless of delays or failures by other processes. 
However, in our formal model, we do not require the processes to halt; they solve the decision task 
and decide by writing to the output component; processes can continue to par\-ti\-ci\-pa\-te. We 
typically consider the behaviour of a process until it decides, and therefore, the above distinction 
does not matter. 

The study of wait-free shared memory protocols has been fundamental in distributed 
computing, some of the most powerful results have been constructed on top of wait-free protocols 
\cite{BG93-2,HS99,rajsbaum2010iterated,SZ00}. Also, other variants of distributed systems can 
be reduced to the wait-free case \cite{BGdistalg,BG93-2,gafniextbgsim}.

\subsubsection*{Definition of consensus and safe-consensus tasks.}\label{subsecconsensustasks}

The tasks of interest in this paper are the \emph{consensus} and \emph{safe-consensus} 
\cite{yehudagafni} tasks.

\paragraph{Consensus}
Every process starts with some initial input value taken from 
a set $I$ 
and must 
output a value 
such 
that:
\begin{itemize}
\item Termination: Each process must eventually output some value.
\item Agreement: All processes output the same value.
\item Validity: If some process outputs $v$, then $v$ is the initial input of some process.
\end{itemize}

\paragraph{Safe-consensus} Every process starts with some initial 
input value taken from a set $I$ and must 
output a value 
such that 
Termination and Agreement are satisfied, and:
\begin{itemize}
\item Safe-Validity:
 If a process $p_i$ starts executing the task and outputs before any other
process  starts executing the task, then its decision is its own proposed input value. 
Otherwise, if two or more processes  access the safe-consensus task concurrently,
then any decision value 
is valid.
\end{itemize}

The safe-consensus task was proposed first in \cite{yehudagafni} as a result of weakening the validity condition of consensus.

\section{Iterated models extended with safe-consensus and results}\label{secAllModels}

Intuitively, a model of distributed computing describes a set of protocols that share some common 
properties and restrictions in the way the processes can access the shared objects and these 
conditions affect the way in which the protocols can be specified. In this paper, we are 
interested in protocols which can be written in a simple and structured way,
 such that the behaviour of the system in the $i$th-iteration, can be described by using the behaviour
 of the $(i - 1)$th-iteration, in an inductive way.
 
 In this section, we introduce an extension of the basic iterated model of \cite{borowsky}, adding 
the power of safe-consensus shared objects. We also present all the results of this paper.
 
\subsection{The iterated model with shared objects}
 
 In the iterated model extended with shared objects, the processes can use two kinds of communication media. 
 The first is the shared memory $SM$ structured as an infinite array of snapshot objects; 
 the second medium is the infinite array $\mathsf{S}$ of shared 
objects ($SM$ and $\mathsf{S}$ are described in Section \ref{secSystems}). 
The processes communicate between them through the snapshot objects and the shared objects of 
$\mathsf{S}$,
 in an asynchronous and round-based pattern. In all the iterated models that we investigate, we 
make the assumption that the shared memory is composed of \emph{one-shot snapshot objects} and the array 
$\mathsf{S}$ contains \emph{one-shot shared objets}. Specifically, we assume that:

\begin{itemize}
 \item The operations $\writesnapshot$ and $\readsnapshot$ of the snapshot objects in $SM$ can 
be used by a process at most once.
\item The $\execobject$ operation of each shared object in $\mathsf{S}$ can be used at most 
once by each process that invokes it.
\end{itemize}

When we want to add the power of shared objects to the standard iterated model \cite{borowsky}, 
we must consider two important questions. The first question is: In which order should we place the three
basic operations (write, read and invoke a shared object)? We have three
possibilities: 
\begin{itemize}
 \item Write, read from the shared memory and invoke a shared object;
\item invoke a shared object, write and read the shared memory;
\item write to memory, invoke a shared object and read the contents of the shared memory.
\end{itemize}
The second question, which is very closely related to the previous one is: Does the order of the
main operations affect the computational power of the new model (for task solvability)? In this
paper, we address these two questions and show the differences of the models of distributed
computing that we obtain when the shared objects invoked by the processes are 
safe-consensus objects.

 A \emph{safe-consensus object} is a shared object that can be invoked by any number of processes. 
The
object receives an input value from each process that invokes it, and 
returns to all the processes an output value that satisfies the Agreement and Validity condition of
the safe-consensus task.
In other words, a safe-consensus object is like a ``black box'' that the
processes can use to solve instances of the safe-consensus task. The method of using distributed
tasks as black boxes inside protocols is a standard way to study the relative
computational power of distributed tasks (i.e. if one task is weaker than another, see 
\cite{yehudagafni,GRH06,Castaneda:2012:RWS:2247370.2247383}). 
Notice that safe-consensus shared objects combined with shared memory registers, can implement
consensus \cite{yehudagafni}.

From now on, we work exclusively in iterated models, where the shared objects invoked by the 
processes are safe-consensus objects.

\subsection{The WRO iterated model}\label{secWRONF}

We now define the first iterated model that we investigate; in it, processes 
write to shared memory, then 
they snapshot it and finally they invoke the 
safe-consensus objects. We say that a
protocol is in the \emph{WRO (Write, Read and 
invoke Object) iterated model} if it can be 
written in the form given 
in Figure \ref{figWROprot}.

\begin{figure}[htb]
\centering{ 
\fbox{
\begin{minipage}[t]{150mm}
\footnotesize
\small
\renewcommand{\baselinestretch}{2.5}
\resetline
\begin{tabbing}
aaaaa\=aaaaa\=aaaaaa\=aaaa\kill 
\line{FI1} \> \textbf{init} $r\leftarrow 0$;
              $sm \leftarrow input$;
$dec \leftarrow \nullval$;
$val \leftarrow \nullval$;
\\~\\
\line{FI2} \> \textbf{loop forever} \\

\line{FI3} \> \>   $r\leftarrow r+1$;\\

\line{FI4} \> \> $SM\left[r\right].\writesnapshot(sm, val)$; \\

\line{FI5} \> \> $sm$ $\leftarrow SM\left[r\right].\readsnapshot()$;\\

\line{FI6} \> \> $val$ $\leftarrow \mathsf{S}\left[ h( \langle r, id, sm, val \rangle) 
\right].\execobject(v)$;\\

\line{FI7} \> \> {\bf if} $(dec=\nullval)$  {\bf then} \\

\line{FI8} \> \> \>  $dec \leftarrow \delta(sm, val)$; \\

\line{FI9} \> \> {\bf end if}  \\

\line{FI10} \> \textbf{end loop}
\end{tabbing}
\normalsize
\end{minipage}
}
\caption{The WRO iterated model}
\label{figWROprot}
}
\end{figure}

An explanation of the pseudocode in Figure \ref{figWROprot} follows.  All the variables $r,sm,$ $val,input$ and $dec$ are local to process 
$p_i$ and only when we analyse a protocol, we add a subindex $i$ to a variable to specify it is 
local to $p_i$. The symbol ``$id$'' contains the id if the executing process. Initially, $r$ is zero and 
$sm$ is assigned the contents of the read-only variable 
$input$, which contains the input value for process $p_i$; all other variables are initialized to 
$\nullval$. In each round, $p_i$ increments by one the loop counter $r$, accesses the current 
shared memory array $SM\left[r \right]$, writing all the information it has stored in $sm$ and 
$val$ (full information) 
and taking a snapshot of the shared memory array and after these operations, 
$p_i$ decides which shared object it is going to invoke by 
executing a deterministic function $h$ that returns an index $l$, 
which $p_i$ uses to invoke the shared object 
$\mathsf{S}\left[l\right]$ with 
some value $v$.
Finally, $p_i$ checks if $dec$ is equal to 
$\nullval$, if so, 
it executes a deterministic function $\delta$ to  determine if it may \emph{decide} 
 a valid output value or $\nullval$. Notice that in each round of a protocol, each process invokes at most one safe-consensus object of the 
array $\mathsf{S}$.

In Section \ref{secImpossConsensus}, we argue that the consensus task cannot be implemented in 
the WRO iterated model using safe-consensus objects (Theorem \ref{thmnowroprot}), this is the main 
result for this iterated model.

\subsection{The OWR iterated model}\label{secOWR}

The second iterated model that we study is the 
\emph{OWR (invoke Object, Write and Read) iterated 
model}. In this model, 
processes first invoke safe-consensus shared 
objects and then they write and snapshot the 
shared memory. We say that 
a protocol $\mathcal{A}$ is in the 
OWR iterated 
model if $\mathcal{A}$ can be written in 
the form specified in Figure \ref{figOWRprot}.

\begin{figure}[htb]
\centering{ 
\fbox{
\begin{minipage}[t]{150mm}
\footnotesize
\small
\renewcommand{\baselinestretch}{2.5}
\resetline
\begin{tabbing}
aaaaa\=aaaaa\=aaaaaa\=aaaa\kill 
\line{JI1} \> \textbf{init} $r\leftarrow 0$;
              $sm \leftarrow input$;
$dec \leftarrow \nullval$;
$val \leftarrow \nullval$;
\\~\\
\line{JI2} \> \textbf{loop forever} \\

\line{JI3} \> \>   $r\leftarrow r+1$;\\

\line{JI4} \> \> $val$ $\leftarrow 
             \mathsf{S}\left[ h( \langle r, id, sm, val \rangle) \right].\execobject(v)$;\\

\line{JI5} \> \> 
     $SM\left[ r \right].\writesnapshot(sm, val)$; \\

\line{JI6} \> \> $sm$ $\leftarrow SM \left[ r \right].\readsnapshot()$;\\

\line{JI7} \> \> {\bf if} $(dec=\nullval)$  {\bf then} \\

\line{JI8} \> \> \>  $dec \leftarrow \delta(sm, val)$; \\

\line{JI9} \> \> {\bf end if}  \\

\line{JI10} \> \textbf{end loop}
\end{tabbing}
\normalsize
\end{minipage}
}
\caption{The OWR iterated model}
\label{figOWRprot}
}
\end{figure}

This pseudocode is explained in a similar way to that used for the code in Figure \ref{figWROprot},
the only thing that changes is the place where we put the invocations to the safe-consensus
shared objects, before the execution of the
snapshot operations.

In Section \ref{secEquivWROOWR}, we argue that for task solvability, there is no real difference 
between the WRO and the OWR iterated models. Any protocol in the WRO iterated model can be simulated 
by a protocol in the OWR iterated model and the converse is also true, this is stated formally in 
Theorem \ref{thmWRO2OWR}. Combining this result with Theorem \ref{thmnowroprot}, we can conclude 
that it is impossible to solve consensus in the OWR iterated model (Corollary 
\ref{cornoowrconsensusprot}).

\subsection{The WOR iterated model}

The last iterated model that we introduce is 
constructed by placing the safe-consensus objects 
between the update and snapshot operations. 
A protocol
$\mathcal{A}$ is in the \emph{WOR (Write, invoke Object and Read)} iterated model 
if it can be written as specified in Figure \ref{figWORprot}.

\begin{figure}[htb]
\centering{ 
\fbox{
\begin{minipage}[t]{150mm}
\footnotesize
\small
\renewcommand{\baselinestretch}{2.5}
\resetline
\begin{tabbing}
aaaaa\=aaaa\=aaaa\=aaaaaaa\=aaaa\kill 
\line{GI1} \> \textbf{init} $r\leftarrow 0$;  
              $sm \leftarrow input$;
$dec \leftarrow \nullval$;
$val \leftarrow \nullval$;
\\~\\
\line{GI2} \> \textbf{loop forever} \\

\line{GI3} \> \>   $r\leftarrow r+1$;\\

\line{GI4} \> \> 
     $SM\left[ r \right].\writesnapshot(sm, val)$; \\


\line{GI5} \> \>  $val$ $\leftarrow 
             \mathsf{S}\left[ h( \langle r, id, sm, val \rangle) \right].exec(v)$;\\


\line{GI6} \> \> $sm$ $\leftarrow SM\left[ r \right].\readsnapshot()$;\\
$ $ \\
\line{GI7} \> \> {\bf if} $(dec=\nullval)$  {\bf then} \\

\line{GI8} \> \> \>  $dec \leftarrow \delta(sm, val)$; \\

\line{GI9} \> \> {\bf end if}  \\

\line{GI10} \> \textbf{end loop}
\end{tabbing}
\normalsize
\end{minipage}
}
\caption{The WOR iterated model
}
\label{figWORprot}
}
\end{figure}

It turns out that the WOR iterated model is quite different from the two previous iterated models. 
This is true because of the following facts:

\begin{itemize}
 \item The consensus task for $n$ processes can be solved in the WOR iterated model using 
only $\binom{n}{2}$ safe-consensus black boxes (Theorem \ref{thmWORconsensusbuenas}).
\item Any protocol in the WOR iterated model which implements consensus using safe-consensus objects 
must use at least $\binom{n}{2}$ safe-consensus objects.
\end{itemize}

The second fact, which is a consequence of Theorem \ref{lemminimumKboxesFULLWOR}, is the main 
result of this paper. It describes a matching lower bound on the number of safe-consensus objects 
needed to solve consensus by any protocol in the WOR iterated model which 
implements consensus. In Section \ref{secResultsWOR}, we give the detailed description of the WOR 
iterated protocol which implements consensus using safe-consensus objects, we prove its correctness 
and finally, we give the proof of the lower bound on the number of safe-consensus objects needed to 
solve consensus in 
the WOR iterated model. Our lower bound proof is based in the fact that, if for a protocol $\mathcal{A}$
there exists $m_0 \in \{2, \ldots, n\}$ such that there are not enough groups (of size $m_0$) of processes 
that can invoke safe-consensus shared objects, then $\mathcal{A}$ will fail to solve consensus. 
Specifically, $\mathcal{A}$ cannot solve consensus if the total number of groups of size $m_0$ is no more 
than $n - m_0$. See Theorem \ref{lemminimumKboxesFULLWOR} in Section \ref{asymtotictighbounds} 
for the full details.

\subsection{Shared objects represented as combinatorial sets}\label{secSharedObjsCombinatorialSets}

We now introduce some 
combinatorial definitions which will help us represent shared objects and the specific way in which 
the processes can invoke these shared objects. These definitions are useful in Sections 
\ref{secImpossConsensus} and \ref{secResultsWOR}.

For any $n\geqslant 1$ and $m \in \overline{n}$, let 
$V_{n,m} = \{ c \subseteq \overline{n} \mid \lvert c \rvert = m \}$. Given a protocol 
$\mathcal{A}$, 
we define for each $m\leqslant n$ the set $\Gamma_\mathcal{A}(n,m)\subseteq 
V_{n,m}$ as 
follows: $b = \{ i_1,\ldots, i_m \} \in \Gamma_\mathcal{A}(n,m)$ if and only if
in some execution of $\mathcal{A}$, only the processes $p_{i_1}, \ldots, p_{i_m}$ 
invoke a safe-consensus object of the array $\mathsf{S}$ (see Figures 
\ref{figWROprot}, \ref{figOWRprot} and \ref{figWORprot}). Roughly speaking,
each $c\in \Gamma_\mathcal{A}(n,m)$
represents a set 
of processes which together can invoke
safe-consensus shared objects
in
$\mathcal{A}$. 

For example, if $m = 3$ and $c=\{ i,j,k \}\in \Gamma_\mathcal{A}(n,3)$, then in at 
least one round
of $\mathcal{A}$, processes 
$p_i,p_j$ and  $p_k$ invoke a safe-consensus object and if in another round or perhaps 
another execution of $\mathcal{A}$, these processes invoke another safe-consensus object in the 
same 
way, then these two invocations are represented by the same set $c\in \Gamma_\mathcal{A}(n,3)$, 
that 
is, shared objects invoked by the same processes are considered as the same element of 
$\Gamma_\mathcal{A}(n,3)$ (repetitions do not count). On the other hand, if $d=\{ i,j,l \} 
\notin \Gamma_\mathcal{A}(n,3)$, then there does not exist an execution of $\mathcal{A}$ in 
which only the three processes $p_i,p_j$ and $p_l$ invoke a safe-consensus shared object. 

A set $b\in \Gamma_\mathcal{A}(n,m)$ is called a 
\emph{$m$-box} or simply a \emph{box}. An element $d \in \Gamma_\mathcal{A}(n,1)$ is called a 
\emph{trivial box}, it represents a safe-consensus object invoked only by one process, we 
consider such objects as useless, because they do not give any additional information to the 
process. We model a process that does not invoke a safe-consensus object as a process that invokes 
a safe-consensus object and no other process invokes that object, i.e., this safe-consensus object 
is represented by a trivial box. A process $p_i$ \emph{participates} in the box $b$ if $i\in b$. 
Let the set $\Gamma_\mathcal{A}(n)$ and the quantities  $\nu_\mathcal{A}(n,m)$ and $\nu_\mathcal{A}(n)$ be defined as 
follows:
\begin{itemize}
 \item [] $\Gamma_\mathcal{A}(n)=\bigcup_{m=2}^n \Gamma_\mathcal{A}(n,m)$;
 \item [] $\nu_\mathcal{A}(n,m)=\lvert \Gamma_\mathcal{A}(n,m) \rvert$;
 \item [] $\nu_\mathcal{A}(n)=\sum_{m=2}^n \nu_\mathcal{A}(n,m)$.
\end{itemize}

 From now on, for all our protocols, we consider global states only at the end of some iteration. Suppose 
that $P$ 
 is a 
reachable state in the protocol $\mathcal{A}$. The set of shared objects $o_1,\ldots,o_q$ 
invoked by the processes to enter state $P$ is represented by a set of boxes 
$\setofboxes{P}=\{ b_1,\ldots, b_q \}$ which
 is called the \emph{global invocation specification} of $P$. We assume 
without loss of generality that in all rounds, each process invokes some shared object, 
that is, the set $\setofboxes{P}$ satisfies
\[\bigcup_{b \in \setofboxes{P}} b = \overline{n},\] 
(a process that does not invoke a safe-consensus object
 can be seen as a process that invokes a safe-consensus object and no other process invokes that 
object). Notice that since each process invokes only one safe-consensus object in each round, 
$\setofboxes{P}$ is a \emph{partition} of $\overline{n}$.

If $b = \{ l_1, \ldots, l_s \} \in \setofboxes{P}$ is a box representing a 
safe-consensus shared object invoked by the processes $p_{l_1},\ldots,p_{l_s}$, we define the 
\emph{safe-consensus value of $b$ in $P$}, denoted by $\scvalof{b}{P}$ as the unique output value of 
the 
safe-consensus shared object represented by $b$.

\subsection{Additional definitions on global states}

We now introduce the notions of connectivity and paths between global states. These are well known 
concepts \cite{fischerImpossCon,HS99} and have become a fundamental tool to study 
distributed systems.

\subsubsection*{Paths of global states}

Two states $S,P$ are said to be \emph{adjacent} if there exists a non-empty subset $X\subseteq
\overline{n}$ such that all processes with ids in $X$ have the same local state in both $S$ and 
$P$. That is, for each $i\in X$, $p_i$ cannot \emph{distinguish} between $S$ and $P$. We denote 
this by $S\linkstates{X} P$.
States $S$ and $P$ are {\it connected}, if we can find a sequence of states (called a \emph{path})
\[\mathfrak{p} \colon S=P_1\linkstates{} \cdots \linkstates{} P_r=P,\] 
such that for all $j$ with $1\leqslant j \leqslant r-1$, $P_j$ and $P_{j+1}$ are adjacent.

Connectivity of global states are a key concept for many beautiful results in distributed systems, 
namely, impossibility proofs. The indistinguishability of states between processes is the building 
block to construct topological structures based on the executions of a given protocol and is 
fundamental in many papers \cite{BG93-2,getco10topologicaltheory,HS99,impossConSHM,SZ00}. In 
addition to the classic definitions of connectivity  and paths, we also introduce the 
following concepts. Let  $\mathfrak{q}\colon Q_1 \linkstates{} 
\cdots \linkstates{} Q_l$ be a path of connected states, define the \emph{set of states} 
$\states{\mathfrak{q}}$; the \emph{set of indistinguishability sets} $\isets{\mathfrak{q}}$;
and the \emph{degree of indistinguishability} $\deg \mathfrak{q}$, of $\mathfrak{q}$ as follows:
\begin{itemize}
 \item [] $\states{\mathfrak{q}}=\{ Q_1, \ldots, Q_l \}$;
\item  [] $\isets{\mathfrak{q}}=\{ X\subseteq \overline{n} \mid (\exists Q_i,Q_j \in
\states{\mathfrak{q}})(Q_i \linkstates{X} Q_j) \}$;
\item [] $\deg \mathfrak{q} = \min \{ \lvert X \rvert \mid X \in \isets{\mathfrak{q}} \}$.
\end{itemize}

The degree of indistinguishability of the path $\mathfrak{q}$ guarantees that we can find for any 
pair of states $Q_i,Q_j \in \states{\mathfrak{q}}$ with $Q_i \linkstates{} Q_j$, a set of 
processes $P\subset \Pi$ that cannot distinguish between $Q_i$ and $Q_j$ such that 
$|P| \geqslant \deg \mathfrak{q}$. The degree of indistinguishability of a path is usually of 
non-importance in the 
standard and well known bivalency proofs of the impossibility of consensus in various systems 
\cite{fischerImpossCon,impossConSHM}, but in the impossibility proof of 
Section \ref{secImpossConsensus} and also the lower bound proof of Section 
\ref{secResultsWOR}, measuring the degree of indistinguishability will be a recurring action in all the 
proofs.

A path $\mathfrak{p}$ of connected states of $\mathcal{A}$ is said to be \emph{\regularpath} if and 
only if $\setofboxes{S} = \setofboxes{Q}$ for all $S,Q\in \states{\mathfrak{p}}$, that is, 
$\mathfrak{p}$ is \regularpath when all the states in the set $\states{\mathfrak{p}}$ have the same 
global invocation specification. 

\begin{lemma}\label{lemAllNormalFormsSameboxes}
Let $\mathcal{A}$ be an iterated protocol for $n$ processes, $A \subseteq \overline{n}$ a
non-empty set and $S,Q$ two reachable states of $\mathcal{A}$ in round $r$, such that 
for all $j\in A$, $p_j$ 
cannot distinguish between $S$ and $Q$.
Then all processes with ids in $A$ participate in the same boxes in $S$ and $Q$. 
\end{lemma}

\begin{proof}[Sketch]
  Intuitively, for each $j \in A$, $p_j$ has the same content in all its local variables, thus 
  it inputs the same values for the function $h$ that outputs the index of the shared object (box) $p_j$ 
  is going to invoke in $S$ and $Q$ respectively. Therefore $h$ outputs the same index in both states,
  gather that, $p_j$ invokes the same shared object in $S$ and $Q$.
\end{proof}

\subsection{Consensus protocols}

We 
need some extra definitions regarding 
consensus protocols: 
If $v$ is a valid input value of 
consensus for processes and $S$ is a state,  we say that $S$ is \emph{$v$-valent} if in 
any possible 
execution starting from $S$, there exists a process that outputs $v$. $S$ is \emph{univalent} if in 
every execution starting from $S$, processes always output the same value. If $S$ is not 
univalent, then $S$ is \emph{bivalent}\footnote{This definition of valency is based on \cite{fischerImpossCon}.}.

\begin{lemma}\label{leministates}
Any two initial states of a protocol for consensus are connected.
\end{lemma}

\begin{proof}
 Let $S,P$ be two initial states. If $S$ and $P$
differ only in the initial value $input_i$ of a single process $p_i$, then $S$ and $P$ are adjacent
(Only $p_i$ can distinguish between the two states, the rest of the processes have the same initial
values). In the case $S$ and $P$ differ in more that one initial value, they can be connected by a
sequence of initial states $S=S_1 \linkstates{} \cdots \linkstates{} S_q=P$ such that $S_j,S_{j+1}$
differ only in the initial value of some process (we obtain $S_{j+1}$ from $S_j$ by changing the
input value of some process $p_i$, the result is a valid input of the consensus problem), hence they
are adjacent. In summary, $S$ and $P$ are connected.
\end{proof}

We need one last result about consensus protocols, we omit its easy proof.

\begin{lemma}\label{lemvvstates}
Suppose that $\mathcal{A}$ is a protocol that solves the consensus task and that $I,J$ are 
connected initial states of $\mathcal{A}$, such that for all rounds $r \geqslant 1$, $I_r, J_r$ are 
connected successor states of $I$ and $J$ respectively. Also, assume that $I$ is a $v$-valent 
state. Then $J$ is $v$-valent.
\end{lemma}

\begin{proof}
  Suppose that $I$ is $v$-valent and $J$ is $v^\prime$-valent ($v\neq v^\prime$). 
  By hypothesis, for all rounds $r\geqslant 1$, $I_r$ and $J_r$ are connected successor states 
  of $I$ and $J$ respectively. Thus, for any $r \geqslant 0$, we can find $r$-round partial executions
  \[ I=I_0,\pi_1,I_1,\ldots,\pi_r,I_r\quad\text{and}\quad J=J_0,\pi^\prime_1,J_1,\ldots,\pi^\prime_r,J_r, \]
  such that $I_r$ and $J_r$ are connected states for all $r\geqslant 0$. Given that $\mathcal{A}$ solves consensus, 
  there must exists $k\geqslant 1$ such that $I_u$ and $J_u$ are decision states for all $u \geqslant k$. 
  Let $\mathfrak{p}_u$ be a path connecting $I_u$ and $J_u$. For the rest of the proof, assume that 
  $\states{\mathfrak{p}_u} =\{ I_u, J_u\}$ (i.e., $I_u \linkstates{} J_u$), as the general case will follow easily from this case and induction 
  on $|\states{\mathfrak{p}_u}|$.
  
  Since 
  $I_u$ is a successor state of $I$, which is a $v$-valent state, all processes decide $v$ in $I_u$; similarly, as $J_u$
  is a successor state of $J$ (a $v^\prime$-valent state), processes decide $v^\prime$ in $J_u$. Let $p$ be a process that 
  cannot distinguish between $I_u$ and $J_u$. Then the local state of $p$ in $I_u$ and $J_u$ must be the same and this includes 
  its output component. But this is a contradiction, because in $I_u$, $p$ decides $v$ while in $J_u$, $p$ decides 
  $v^\prime$. We conclude that $J$ must be $v$-valent, such as $I$.
\end{proof}

\section{Impossibility of consensus in the WRO and OWR iterated models}\label{secImpossConsensus}

In this section, we prove that the consensus task cannot be implemented in the WRO iterated model with 
safe-consensus objects and we extend this result to the OWR iterated model, by showing that the  WRO and 
OWR iterated  models are equivalent models in their computational power to solve tasks. Given
that there exists a protocol in the standard 
model that solves consensus with registers 
and safe-consensus objects \cite{yehudagafni},
it is a consequence of the results presented 
in this section that the WRO and OWR iterated 
models are not equivalent to the standard model 
extended with safe-consensus objects.

\subsection{The impossibility of consensus in the WRO iterated model}\label{secImpossConsensusWRO}

We give a negative result concerning protocols in the WRO iterated model.
Although there exist wait-free shared memory protocols that can
 solve consensus using safe-consensus objects \cite{yehudagafni}, in this paper we show that
there is no protocol in the WRO iterated model that solves consensus using safe-consensus objects.
We first introduce some extra definitions and results that we will be using and  after that, we 
prove that  there is no protocol in the  WRO  iterated model which can solve the consensus task. 
For simplicity, we will refer to a protocol 
in the WRO iterated model simply as a WRO protocol. 

\begin{lemma}\label{lemWROsamesafeconfiguration}
Let $\mathcal{A}$ be a WRO protocol for $n$ processes, $A_i=\overline{n} - i$ for
some $i\in \overline{n}$. Assume there exist $S,Q$ two reachable states of $\mathcal{A}$ in round $r$, such that 
for all $j\in A_i$, $p_j$ 
does not distinguish between 
$S$ and $Q$. Then
$\setofboxes{S}=\setofboxes{Q}$.
\end{lemma}

\begin{proof}
 Let $\mathcal{A}$ be a WRO protocol,
$A_i=\overline{n} - i$ and $S,Q$ such that every process $p_j$ with $j\in A_i$ 
cannot distinguish between 
$S$ and $Q$. By Lemma \ref{lemAllNormalFormsSameboxes}, if $b\in
\setofboxes{S}$ is a box such that $i\notin b$, then $b\in \setofboxes{Q}$.
For the box $c\in \setofboxes{S}$ that contains the id $i$, we argue by cases:
\begin{itemize}
 \item [] {\it Case I.} $\lvert c \rvert = 1$. All processes with ids in $A_i$ participate in the
same boxes in both states $S$ and $Q$ and $p_i$ does not participate in those boxes, thus $c = \{ i
\} \in \setofboxes{Q}$.
\item  [] {\it Case II.} $\lvert c \rvert > 1$. There exists a process $p_j$ with $j\in c$ and
$j\neq i$. Then $j\in A_i$, so that by Lemma \ref{lemAllNormalFormsSameboxes}, $c\in
\setofboxes{Q}$.
\end{itemize}
We have prove that $\setofboxes{S}\subseteq \setofboxes{Q}$. To show that the other inclusion
holds, the argument is symmetric. Therefore $\setofboxes{S} =
\setofboxes{Q}$.
\end{proof}

In order to prove the impossibility of consensus in the WRO iterated model (Theorem 
\ref{thmnowroprot}), we need various results that describe the structure
of protocols in the WRO iterated model. These results tell us that for any given WRO protocol, the
degree of indistinguishability of paths of connected states 
is high, even with the added power of 
safe-consensus objects, because we can make the safe-consensus values returned by the shared 
objects the same in all these reachable states, so that they cannot help the processes distinguish between these states. 
This is the main reason (of) why consensus cannot be implemented with safe-consensus in the WRO iterated model.

The following construction
will be useful for
the rest of this section. If $A_1,\ldots,$ $A_q\subset\overline{n}$ $(q\geqslant 1)$ is a collection of
disjoin subsets of $\overline{n}$, define the round schedule 
\[\xschedulewro{A_1,\ldots,A_q,Y}\] 
for $\mathcal{A}$ as:
\begin{equation*}
 \writeop{A_1},\readop{A_1},\writeop{A_2},\readop{A_2},\ldots,\writeop{A_q},\readop{A_q},\writeop{Y}
,\readop{Y},\scop{\overline{n}}
\end{equation*}
where 
$Y=\overline{n} - (\bigcup_{i=1}^q A_i)$. Sometimes, 
we omit the set $Y$ and just write $\xschedulewro{A_1,\ldots,A_q}$. The key point 
of the round schedule $\xschedulewro{A_1,\ldots,A_q,Y}$ is that in a given 
one-round execution of a WRO protocol, it will allow us to choose the safe-consensus values of the safe-consensus shared objects invoked by at least two 
processes (using the Safe-Validity property of the safe-consensus task).

Our next Lemma is the first step to prove Theorem \ref{thmnowroprot}. Roughly 
speaking, it tells us that if a WRO protocol has entered any reachable state 
$S$, then we can find two one-round successor states of $S$, such that we 
can connect these states with a small path of very high indistinguishability degree. These two successor states of $S$ are obtained when the processes 
execute the WRO protocol in such a way that only one process is delayed in the 
update and snapshot operations, but all processes execute the safe-consensus 
shared objects concurrently in every state of the resulting path, 
thus we can keep the safe-consensus values of every 
safe-consensus shared object the same in every state.

\begin{lemma}
\label{lembasicBregularpathWRO}
 Let $n\geqslant 2$, $\mathcal{A}$ a WRO protocol 
and
$i,j\in \overline{n}$. If $S$ is a reachable state of $\mathcal{A}$ in round 
$r\geqslant 0$ and the states $Q_i=S\cdot \xschedulewro{\overline{n} - i}$ and 
$Q_j=S\cdot \xschedulewro{\overline{n} - j}$ satisfy the conditions
\begin{itemize}
	\item[A.] $\setofboxes{Q_i} = \setofboxes{Q_j}$;
	\item[B.] $(\forall b\in \setofboxes{Q_i})
	(\scvalof{b}{Q_i}=\scvalof{b}{Q_j})$.
\end{itemize}
Then $Q_i$ and $Q_j$ are connected in round $r+1$ of $\mathcal{A}$ with a \regularpath
path $\mathfrak{q}$ with $\deg \mathfrak{q} \geqslant n - 1$.
\end{lemma}

\begin{proof}
 If $i=j$, the result is immediate. Suppose that $i\neq j$ and $\overline{n} - i = \{
l_1,\ldots,l_{n - 1} \}$ with $l_{n-1}=j$. We show that 
$Q_i$
and $Q_j$ can be connected with a \regularpath sequence with
indistinguishability degree $n - 1$ by building three subsequences 
$\mathfrak{q}_1,\mathfrak{q}_2$
and $\mathfrak{q}_3$. 

We proceed to build the sequence $\mathfrak{q}_1$. We claim that there exists a state 
$Q$ of the form $Q=S\cdot \xschedulewro{l_1, \overline{n} - \{ i, l_1\},i}$ such that 
\begin{itemize}
	\item $\setofboxes{Q_i} = \setofboxes{Q}$;
	\item $\scvalof{b}{Q_i}=\scvalof{b}{Q}$  for all $b\in \setofboxes{Q}$.
\end{itemize}
The first property holds because in any state of the form 
$S\cdot \xschedulewro{l_1, \overline{n} - \{ i, l_1\},i}$, 
only process $l_1$ can distinguish between $Q_i$ and 
$Q$, thus by Lemma \ref{lemWROsamesafeconfiguration} we obtain 
that $\setofboxes{Q_i} $ $= \setofboxes{Q}$. Now, in both $Q_i$ and $Q$, 
all processes execute the safe-consensus shared objects concurrently and we can use 
the Safe-Validity property of the safe-consensus task to find a state of $\mathcal{A}$ in round 
$r+1$ such that $Q=S\cdot \xschedulewro{l_1, \overline{n} - \{ i, l_1\},i}$ in which 
$\scvalof{b}{Q_i}=\scvalof{b}{Q}$ for each 
$b\in \setofboxes{Q}=\setofboxes{Q_i}$. This proves the second property. 
Therefore, we have that 
\[ Q_i=S\cdot \xschedulewro{\overline{n} - i,i} \linkstates{\overline{n} - l_1} 
S\cdot \xschedulewro{l_1,\overline{n} - \{ i, l_1\},i} = Q	
\]
and this sequence is a 
\regularpath path. Repeating the same argument, we can connect $Q$ with a state of 
the form $P=S\cdot \xschedulewro{l_1,l_2,\overline{n} - \{ i, l_1,l_2\},i}$ and obtain 
the three state \regularpath path $Q_i \linkstates{\overline{n} - l_1} Q \linkstates{\overline{n} - l_2} P$. 
Continuing in this way, we can show that the path 
$\mathfrak{q}_1$ is given by
\begin{equation*}\label{eqlembasicBregularpathWRO1}
 S\cdot \xschedulewro{\overline{n} - i,i} \linkstates{\overline{n} - l_1} S\cdot \xschedulewro{l_1,
\overline{n} - \{ i, l_1\},i} \linkstates{\overline{n} - l_2} \cdots \linkstates{\overline{n} -
l_{n - 2}} S\cdot \xschedulewro{l_1,\ldots, l_{n - 1}, i},
\end{equation*}
which is clearly a \regularpath path by repeated use of the previous argument. In
the same way, we construct the sequence $\mathfrak{q}_2$ as
\begin{equation*}\label{eqlembasicBregularpathWRO2}
 S\cdot \xschedulewro{l_1,\ldots, l_{n - 1}, i} \linkstates{\overline{n} - l_{n - 1}} S\cdot
\xschedulewro{l_1,\ldots, \{ l_{n - 1}, i\}} \linkstates{\overline{n} - i} 
S\cdot \xschedulewro{l_1,\ldots, i, l_{n - 1}}
\end{equation*}
and finally, in a very similar way, we build $\mathfrak{q}_3$ as follows,
\begin{multline*}\label{eqlembasicBregularpathWRO3}
S\cdot \xschedulewro{l_1,\ldots, i, l_{n - 1}} \linkstates{\overline{n} - l_{n - 2}} S\cdot
\xschedulewro{l_1,\ldots, \{ l_{n - 2}, i\}, l_{n - 1}} \\ 
\linkstates{\overline{n} - l_{n - 3}} \cdots \linkstates{\overline{n} - l_{2}} \\
S\cdot \xschedulewro{l_1,\overline{n} - \{l_1, l_{n - 1} \}, l_{n - 1}} \linkstates{\overline{n} -
l_{1}} S\cdot \xschedulewro{\overline{n} - l_{n - 1},l_{n - 1}}.
\end{multline*}
Both $\mathfrak{q}_2$ and $\mathfrak{q}_3$ are
\regularpath paths and for $k=1,2,3$, $\deg \mathfrak{q}_k = n - 1$. Notice that the processes
execute the safe-consensus objects in the same way in all the states of the previous
sequences, and
by 
construction, 
all the states have the same invocation specification. 
As
$j=l_{n -
1}$, $S\cdot \xschedulewro{\overline{n} - l_{n - 1},l_{n - 1}} = S\cdot \xschedulewro{\overline{n} -
j,j}=Q_j$, thus we can join $\mathfrak{q}_1,\mathfrak{q}_2$ and $\mathfrak{q}_3$ to obtain a
\regularpath path $\mathfrak{q} \colon Q_i \linkstates{} \cdots \linkstates{} Q_j$ such that $\deg
\mathfrak{q} = n - 1$. This finishes the proof.
\end{proof}

Lemma \ref{lemconstates} is the next step towards proving Theorem \ref{thmnowroprot}. 
Basically, it says that for any WRO protocol, if we already have a path of connected states 
with high indistinguishability degree for some round, then we can use this
path to build a new sequence of connected states in the next round with a high degree
of indistinguishability. Again, using the round schedules where the processes execute
the safe-consensus shared objects concurrently, we can render useless the 
safe-consensus values, so that the number of processes that cannot distinguish 
between a pair of states of the path is as high as possible.

\begin{lemma}\label{lemconstates}
Let $\mathcal{A}$ be a WRO protocol and $S,P$ be two reachable states of some
round $r\geqslant 0$ such that $S$ and $P$ are connected with a sequence 
$\mathfrak{s}$ such
that $\deg \mathfrak{s} \geqslant n - 1$;
suppose also that the number of participating processes
is $n\geqslant 2$. Then there exist successor states $S^{\prime},P^{\prime}$ of $S$ and $P$
respectively, in round $r+1$ of $\mathcal{A}$, such that $S^{\prime}$ and $P^{\prime}$ are connected with
a \regularpath path $\mathfrak{q}$ with $\deg \mathfrak{q} \geqslant n - 1$.
\end{lemma}

\begin{proof}
 Suppose that $S$ and $P$ are connected with a
sequence $\mathfrak{s} \colon S=S_1 \linkstates{} \cdots \linkstates{} S_m=P$ such that for all 
$j$ $(1\leqslant j \leqslant m-1)$,
$S_j\linkstates{X} S_{j+1}$, where $\lvert X \rvert \geqslant n - 1$. Let $p_1,\ldots,p_n$ be the
set of participating processes. We now construct a sequence of connected states 
\[\mathfrak{q} \colon Q_1 \linkstates{} \cdots \linkstates{} Q_s, \]
in such a way that each $Q_i$ is a successor state to some state $S_j$ and $\deg \mathfrak{q} \geqslant n
- 1$. We use induction on $m$; for the basis, consider the adjacent states $S_1,S_2$ and suppose
that all processes with ids in the set $X_1$ $(\lvert X_1 \rvert \geqslant n - 1)$ cannot
distinguish between $S_1$ and $S_2$. By the Safe-Validity property of the safe-consensus task and 
Lemma \ref{lemWROsamesafeconfiguration}, we can find states $Q_1,Q_2$ such that
\begin{itemize}
	\item $Q_1=S_1\cdot \xschedulewro{X_1}$ and $Q_2=S_2\cdot \xschedulewro{X_1}$;
	\item $Q_1$ and $Q_2$ are indistinguishable for all processes with ids in $X_{1}$;
	\item $\setofboxes{Q_1}=\setofboxes{Q_2}$.
\end{itemize}
 Combining these facts, we can see that 
 the small sequence
$Q_1\linkstates{X_1} Q_2$ is \regularpath and has indistinguishability degree at least $n - 1$. 
Assume now that we have build the sequence 
\begin{equation*}
\mathfrak{q}^\prime \colon Q_1 \linkstates{} \cdots \linkstates{} Q_{s^\prime},
\end{equation*}
of connected successor states for $S_1\linkstates{} \cdots \linkstates{} S_q$ $(1\leqslant q < m)$
such that $\mathfrak{q}^\prime$ is \regularpath, $\deg \mathfrak{q}^\prime \geqslant n - 1$ and
$Q_{s^\prime}=S_q \cdot \xschedulewro{X}$ with $\lvert X \rvert \geqslant n -1$. Let
$X_q$ with $\lvert X_q \rvert \geqslant n - 1$ be a set of processes' ids that cannot
distinguish between $S_q$ and
$S_{q+1}$. To connect $Q_{s^\prime}$ with a successor state for $S_{q+1}$, we first use Lemma
\ref{lembasicBregularpathWRO} to connect $Q_{s^\prime}$ and $Q=S_q \cdot \xschedulewro{X_q}$ by
means of a \regularpath sequence $\mathfrak{p}$ such that $\deg \mathfrak{p} \geqslant n - 1$.
Second, notice that we have the small \regularpath path $\mathfrak{s}\colon Q \linkstates{X_q}
S_{q+1} \cdot \xschedulewro{X_q}=Q_{s^\prime +1}$. In the end, we use
$\mathfrak{q}^\prime,\mathfrak{p}$ and $\mathfrak{s}$ to get a \regularpath sequence
\[ \mathfrak{q} \colon Q_1 \linkstates{} \cdots \linkstates{} Q_{s^\prime +1}, \]
which fulfils the inequality $\deg \mathfrak{q} \geqslant n - 1$. By induction, the result
follows.
\end{proof}

Using Lemmas \ref{lembasicBregularpathWRO} and \ref{lemconstates} and previous results 
regarding bivalency arguments, we can finally prove the impossibility of consensus in the WRO iterated model.

\begin{theorem}\label{thmnowroprot}
There is no protocol for consensus in the WRO iterated model using safe-consensus objects.
\end{theorem}

\begin{proof}
Suppose $\mathcal{A}$ is a WRO protocol
that solves consensus with safe-consensus objects and without loss, assume that 0,1 are two
valid input values. Let $O,U$ be the initial states in which all processes have as initial values
only 0 and 1 respectively. By Lemma \ref{leministates}, $O$ and $U$ are connected and it is easy to
see that the sequence joining $O$ and $U$, build in the proof on Lemma \ref{leministates} has degree
of indistinguishability $n -1$. So that applying an easy induction, we can use Lemma 
\ref{lemconstates} to show that there exist connected successor states 
$O_r,U_r$ of $O$ and $U$ respectively in each round $r\geqslant 1$ of the protocol $\mathcal{A}$. 
Clearly, $O$ is a 0-valent state and by Lemma \ref{lemvvstates}, $U$ is 0-valent.
But this is a contradiction, because $U$ is a 1-valent state. Therefore no such protocol
$\mathcal{A}$ can exists. 
\end{proof}

At this moment we remark that the key fact that allows us to prove this 
impossibility of consensus in the WRO model, is that in all our proofs, 
we are able to ``neutralize'' the safe-consensus objects by making 
all processes 
invoke their respective shared objects concurrently, this is the entry point 
for the Safe-Validity property of consensus to allow us to find the states of 
the executing protocol with the same safe-consensus values, so that the 
processes can only distinguish different states by their local view of the 
shared memory. In contrast to the protocols in the WOR model, it is not 
possible to apply this technique and neutralize the safe-consensus objects,
because the shared objects operations are located between the update and 
snapshot operations, so that when the processes take the snapshot, they 
have already invoked the safe-consensus shared objects.

\subsection{The equivalence of the WRO and OWR iterated models}\label{secEquivWROOWR}

In this section, we prove Theorem \ref{thmWRO2OWR}, which tells us that the WRO and the OWR iterated models have the same computational power. We show that given  any protocol  in the WRO iterated  model,  there is  an OWR protocol that 
can simulate its behaviour and the converse is  also true. To check the meaning of some definition or
previous result, the reader may consult Sections \ref{secdcBasicDef}, \ref{secWRONF} and 
\ref{secOWR}. For the sake of simplicity, we will refer to a protocol in the OWR iterated model, 
simply as an OWR protocol.

\subsubsection*{Transforming a WRO protocol into an OWR protocol}

The algorithm of Figure \ref{figWRO2OWRprot} is a generic OWR protocol to simulate protocols in the 
WRO iterated model. Suppose that $\mathcal{A}$ is a WRO protocol that solves a task $\Delta$ and 
assume
that $h_\mathcal{A}$ is the deterministic function to select safe-consensus objects in 
$\mathcal{A}$ and $\delta_\mathcal{A}$ is the decision map used in the protocol $\mathcal{A}$. To 
obtain an OWR protocol $\mathcal{B}$ that simulates the behaviour of $\mathcal{A}$, we use the generic 
protocol of Figure \ref{figWRO2OWRprot}, replacing the functions $h_{\_}$ and $\delta_{\_}$ with 
$h_\mathcal{A}$ and $\delta_\mathcal{A}$ respectively. If the processes execute $\mathcal{B}$ with 
valid input values from $\Delta$, then in the first round, the participating processes discard the 
output value of the safe-consensus object that they invoke, because the test at line \eqref{OI5} is 
successful and after that, each process goes on to execute the write-snapshot operations, with the 
same values that they would use to perform the same operations in round one of $\mathcal{A}$. Later, 
at lines \eqref{OI12}-\eqref{OI14}, they do some local computing to try to decide an output value. 
It is clear from the code of the generic protocol that none of the executing processes will write a
non-$\nullval$ value to the local variables $dec$, thus no process makes a decision in the first
round and they finish the current round, only with two simulated operations from $\mathcal{A}$:
write and snapshot. In the second round, the processes invoke one or more safe-consensus objects at 
line 
\eqref{OI4}, accordingly to the return values of the function $h_\mathcal{A}$, which
is invoked with the values obtained from the snapshot operation of the previous round (from the
round one write-read simulation of $\mathcal{A}$) and notice that instead of using the current value
of the round counter $r$, the processes call $h_\mathcal{A}$ with the parameter $r - 1$. Since $r = 
2$, the processes invoke $h_\mathcal{A}$ as if they were still executing the first round of
the simulated protocol $\mathcal{A}$. Finally, after using the 
safe-consensus objects,
processes do local computations in lines \eqref{OI7}-\eqref{OI9}, using the decision map
$\delta_\mathcal{A}$ to simulate the decision phase of $\mathcal{A}$, if for some process $p_j$, it
is time to take a decision in $\mathcal{A}$, it stores the output value in the local variable
$dec_j^\prime$, which is used at the end of the round to write the output value in $dec_j$ and then
$p_j$ has decided. If the map $\delta_\mathcal{A}$ returns $\nullval$, then $p_j$ goes on to do the
next write-read operations (simulating the beginning of round two of $\mathcal{A}$) and proceeds to
round three of $\mathcal{B}$. The described behaviour is going to repeat in all subsequent rounds.

\begin{figure}[htb]
\centering{ 
\fbox{
\begin{minipage}[t]{150mm}
\footnotesize
\small
\renewcommand{\baselinestretch}{2.5}
\resetline
\begin{tabbing}
aaaaa\=aaaaa\=aaaaaa\=aaaa\=aaa\=aaa\kill 
\textbf{Algorithm} $\mathsf{WRO\text{-}Simulation}$ \\
\line{OI1} \> \textbf{init} $r\leftarrow 0$; $sm,tmp \leftarrow input$; 
$dec,dec^\prime \leftarrow \nullval$; $val \leftarrow \nullval$;
\\~\\
\line{OI2} \> \textbf{loop forever} \\

\line{OI3} \> \>   $r\leftarrow r+1$;\\

\line{OI3P} \> \>   $sm\leftarrow tmp$;\\

\line{OI4} \> \> $val$ $\leftarrow 
             {\mathsf S}\left[ h_{\_}(r - 1, id, sm, val) \right].\execobject(v)$;\\

\line{OI5} \> \> {\bf if} $(r = 1)$  {\bf then} \\

\line{OI6} \> \> \> $val \leftarrow \nullval$; $\qquad\qquad\qquad$ /* 1st round: ignore
returned value */ \\

\line{OI7} \> \> {\bf else if} $(dec^\prime = \nullval )$  {\bf then} \\

\line{OI8} \> \> \>  $dec^\prime \leftarrow \delta_{\_}(sm, val)$; $\quad$ /* Decision of
simulated protocol */ \\

\line{OI9} \> \> {\bf end if}  \\

\line{OI10} \> \> 
     $SM\left[ r \right].\writesnapshot(sm, val)$; \\

\line{OI11} \> \> $tmp$ $\leftarrow SM\left[ r \right].\readsnapshot()$;\\

\line{OI12} \> \> {\bf if} $(dec=\nullval)$  {\bf then} \\

\line{OI13} \> \> \>  $dec \leftarrow dec^\prime$;  $\qquad\quad$ /* Decide when simulated
protocol decides */ \\

\line{OI14} \> \> {\bf end if}  \\

\line{OI15} \> \textbf{end loop}
\end{tabbing}
\normalsize
\end{minipage}
}
\caption{General WRO-simulation protocol in the OWR iterated model}
\label{figWRO2OWRprot}
}
\end{figure}

\subsubsection*{Correspondence between executions of the protocols}

 Let $\pi$ be a round schedule. Denote by $\pi\left[ 
\eventfont{W},\eventfont{R} \right]$ the sequence of write
and read events of $\pi$, these events appear in $\pi\left[ \eventfont{W},\eventfont{R} 
\right]$ just in
the same order they are specified in $\pi$. The symbol $\pi\left[ \eventfont{S} \right]$ is defined 
for
the event \eventfont{S} in a similar way. For example, given the round schedule
\[ \pi=\writeop{1, 3}, \readop{1,3}, \writeop{2}, \readop{2}, \scop{1,2,3}, \]
then $\pi\left[ \eventfont{W},\eventfont{R} \right]=\writeop{1, 3}, \readop{1,3}, \writeop{2}, 
\readop{2}$ and
$\pi\left[ \eventfont{S} \right]=\scop{1,2,3}$. For another example, suppose that the round schedule 
$\pi^\prime$ is given by
\[ \pi^\prime=\writeop{1,2,3}, \readop{1,2,3},\scop{1}, \scop{3}, \scop{2}. \]
In this case, we have that $\pi^\prime\left[ \eventfont{W},\eventfont{R} \right]=\writeop{1,2,3}, \readop{1,2,3}$ and 
$\pi^\prime\left[ \eventfont{S} \right]=\scop{1},\scop{3},$ $\scop{2}$.

Let $\mathcal{A}$ be a WRO protocol as given in Figure \ref{figWROprot},
$h_\mathcal{A}$ the deterministic function to select safe-consensus
objects and $\delta_\mathcal{A}$ the decision map used in the protocol $\mathcal{A}$. An OWR
protocol $\mathcal{B}$ that simulates the behaviour of $\mathcal{A}$ is obtained by using the
generic protocol $\mathcal{G}$ of Figure \ref{figWRO2OWRprot}. 
To construct $\mathcal{B}$, we only
replace the functions $h_{\_}$ and 
$\delta_{\_}$ of $\mathcal{G}$ with 
$h_\mathcal{A}$ and
$\delta_\mathcal{A}$ respectively. 

In order to show that $\mathcal{B}$ simulates $\mathcal{A}$, we first notice that there is a 
correspondence between executions of $\mathcal{A}$ and $\mathcal{B}$. For, if $\alpha$ is an 
execution of $\mathcal{A}$ such that
\[ \alpha=S_0,\pi_1,S_1,\ldots\pi_k,S_k,\pi_{k+1},\ldots \]
 Then there exists an execution $\alpha_\mathcal{B}$ of $\mathcal{B}$
\[
\alpha_\mathcal{B}=S_0,\pi^\prime_1,S^\prime_1\ldots,\pi^\prime_k,S^\prime_k,\pi^\prime_{k+1},\ldots
, \]
such that 
\[ \pi_l^\prime = \begin{cases}
    \scop{X_1}\ldots,\scop{X_s},\pi_1\left[ \eventfont{W},\eventfont{R} \right]\quad\text{if } 
l=1, \\
    \pi_{l-1}\left[ \eventfont{S} \right],\pi_l\left[ \eventfont{W},\eventfont{R} \right]\quad 
\text{otherwise,}
   \end{cases}
 \]
where $\bigcup_{j=1}^s X_j=\overline{n}$ is an arbitrary partition of $\overline{n}$. Conversely, 
for
any execution $\beta=R_0,\eta_1,R_1,\ldots$ $\eta_k,R_k,$ $\eta_{k+1},\ldots$ of the protocol
$\mathcal{B}$, we can find an execution
\[\beta_\mathcal{A}=R_0,\eta^\prime_1,R^\prime_1,\ldots\eta^\prime_k,R^\prime_k,\eta_{k+1}, \ldots\]
of $\mathcal{A}$ such that 
\[ \eta_l^\prime=\eta_l\left[ \eventfont{W},\eventfont{R} \right],\eta_{l+1}\left[ \eventfont{S} 
\right] \quad \text{for all } l\geqslant 1. \]

\begin{lemma}\label{leminvarianceStatesRoundRRm1}
 Suppose that $\beta$ is an execution of the protocol $\mathcal{B}$ and process $p_i$ is executing
the protocol (accordingly to $\beta$) in round $r > 0$ at line \ref{OI10}. Then the local variables
$sm_i,dec^\prime_i$ and $val_i$ of $\mathcal{B}$ have the same values of the respectively local
variables $sm_i,dec_i$ and $val_i$ of the protocol $\mathcal{A}$, when $p_i$ finishes executing
round $r - 1$ of $\mathcal{A}$, in the way specified by the execution $\beta_\mathcal{A}$.
\end{lemma}

\begin{proof}
 The proof is based on induction on the round number $r$, the base case and the inductive case are
obtained by an easy analysis of the programs given in Figures \ref{figWROprot} and
\ref{figWRO2OWRprot}.
\end{proof}

The converse of the previous lemma is also true. 

\begin{lemma}\label{leminvarianceStatesRoundRRm1AB}
 Suppose that $\alpha$ is an execution of the protocol $\mathcal{A}$ and process $p_i$ is executing
the protocol (accordingly to $\alpha$) in round $r > 0$ at line \ref{FI4}. Then the local variables
$sm_i,dec_i$ and $val_i$ of $\mathcal{A}$ have the same values of the respective local
variables $sm_i,dec^\prime_i$ and $val_i$ of the protocol $\mathcal{B}$, when $p_i$ finishes 
executing
round $r + 1$ of $\mathcal{B}$, in the way specified by the execution $\alpha_\mathcal{B}$.
\end{lemma}

The final result that we need to prove that
$\mathcal{B}$ simulates $\mathcal{A}$ for task solvability is an immediate consequence of Lemma
\ref{leminvarianceStatesRoundRRm1}.

\begin{lemma}
Suppose that $\Delta$ is a decision task solved by the protocol $\mathcal{A}$ and let $p_i\in \Pi$.
Then $p_i$ decides an output value $v$ in round $r > 0$ of an execution $\alpha$ of $\mathcal{A}$ if
and only if $p_i$ decides the output value $v$ in round $r+1$ of $\mathcal{B}$, in the execution
$\alpha_\mathcal{B}$.
\end{lemma}

This shows that the protocol $\mathcal{B}$ simulates the behaviour of $\mathcal{A}$ for task
solvability and therefore proves the first part of Theorem \ref{thmWRO2OWR}. 

\subsubsection*{Transforming an OWR protocol into a WRO protocol}

Now we show that 
 any OWR protocol can be simulated with a WRO protocol. We follow the same technique that we used 
in the previous section, we use the generic protocol of 
Figure \ref{figOWR2WROprot}. The intuitive argument of how it works is very much the same as in the 
WRO case, thus we omit the details.

\begin{figure}[htb]
\centering{ 
\fbox{
\begin{minipage}[t]{150mm}
\footnotesize
\small
\renewcommand{\baselinestretch}{2.5}
\resetline
\begin{tabbing}
aaaaa\=aaaaa\=aaaaaa\=aaaa\=aaa\=aaa\kill 
\textbf{Algorithm} $\mathsf{OWR\text{-}Simulation}$ \\
\line{UI1} \> \textbf{init} $r\leftarrow 0$; $sm \leftarrow input$; 
$dec,dec^\prime \leftarrow \nullval$; $val \leftarrow \nullval$;
\\~\\
\line{UI2} \> \textbf{loop forever} \\

\line{UI3} \> \>   $r\leftarrow r+1$;\\


\line{UI10} \> \> 
     $SM\left[ r \right].\writesnapshot(sm, val)$; \\

\line{UI11} \> \> $sm$ $\leftarrow SM\left[ r \right].\readsnapshot()$;\\

\line{UI5} \> \> {\bf if} $(r = 1)$  {\bf then} \\

\line{UI6} \> \> \> $sm \leftarrow input$; $\qquad$ /* 1st round: ignore
write-snapshot */ \\

\line{UI7} \> \> {\bf else if} $(dec^\prime = \nullval )$  {\bf then} \\

\line{UI8} \> \> \>  $dec^\prime \leftarrow \delta_{\_}(sm, val)$; 
\\

\line{UI9} \> \> {\bf end if}  \\

\line{UI4} \> \> $val$ $\leftarrow 
             {\mathsf S}\left[ h_{\_}(r, id, sm, val)\right].\execobject(v)$;\\

\line{UI12} \> \> {\bf if} $(dec=\nullval)$  {\bf then} \\

\line{UI13} \> \> \>  $dec \leftarrow dec^\prime$; 
\\

\line{UI14} \> \> {\bf end if}  \\

\line{UI15} \> \textbf{end loop}
\end{tabbing}
\normalsize
\end{minipage}
}
\caption{General OWR-simulation protocol in the WRO iterated model}
\label{figOWR2WROprot}
}
\end{figure}

\begin{theorem}\label{thmWRO2OWR}
 Let $\mathcal{A}$ be a protocol in the WRO iterated model which solves a task $\Delta$. Then
$\mathcal{A}$ can be simulated with a protocol $\mathcal{B}$ in the OWR iterated model. Conversely, 
for every
protocol $\mathcal{B}^\prime$ in the OWR iterated model that solves a task $\Delta^\prime$, there is 
a
protocol $\mathcal{A}^\prime$ in the WRO iterated model, which simulates $\mathcal{B}^\prime$.
\end{theorem}

The equivalence of the WRO iterated model with the OWR iterated model can be combined with Theorem
\ref{thmnowroprot} to obtain the following

\begin{corollary}\label{cornoowrconsensusprot}
 There is no protocol in the OWR iterated model for consensus using safe-consensus objects.
\end{corollary}

The results of this section about the WRO and OWR iterated models combined with the consensus
protocol using safe-consensus objects of \cite{yehudagafni}, allow us to conclude that the 
standard model enriched with safe-consensus shared objects is more powerful that the WRO and 
OWR iterated models. But there are still some unanswered questions regarding these two iterated 
models, see our conclusions at Section \ref{secConclusions}.

\section{Solving consensus in the WOR iterated model}\label{secResultsWOR}

In this section, We investigate the solvability of consensus in the WOR iterated model. We 
show that there exists a protocol in the WOR iterated model for consensus with $\binom{n}{2}$ 
safe-consensus objects. Although the formal specification of this protocol is 
a bit complicated (see Section \ref{secTheConsensusProtocol}), the intuitive idea behind it is 
quite simple and this idea can be depicted in a graphical way. This is a nice property of our 
consensus protocol and it is a consequence of working in an (extended) iterated model of 
computation. As a byproduct of the design of the consensus protocol, we discovered the new $g$-2coalitions-consensus task, which is a new kind of ``consensus task''. More details can be found in Section 
\ref{secCCTask}.
In Section \ref{asymtotictighbounds}, we present the main result of this paper, the lower bound on the number of 
safe-consensus objects needed by any protocol in the WOR iterated model which implements consensus. The lower bound proof is a result of combining a careful study of the connectivity 
of protocols in the WOR iterated model and the effects of invoking safe-consensus shared objects 
on the indistinguishability degree of paths connecting reachable states of the protocols. These 
effects are described in part combinatorially by a new application of 
Johnson graphs in distributed computing.
 We use terminology and results from Sections 
\ref{secdcBasicDef} and \ref{secAllModels}. For simplicity, we refer to any protocol in the WOR 
iterated model as a WOR protocol.

\subsection{The coalitions-consensus task}
\label{secCCTask}

In order to present the complete specification 
of the consensus protocol, we define the 
2coalition-consensus task, which will be useful 
to give a simpler description of the consensus 
protocol of Figure \ref{figWORprotconsensus2} 
and prove its correctness in Theorem 
\ref{thmWORconsensusbuenas}, but before that, we show 
that the 2coalitions-consensus task can be solved in the
WOR iterated model in one round with one snapshot object
and one safe-consensus object.

Intuitively, in the \emph{$g$-2coalitions-consensus} task, $g$ processes try to decide a 
consensus value from two 
proposed input values, there are exactly $g - 2$ processes which know the two input values, 
while each of the two remaining processes know only one input value and it is this last fact that makes the 2coalitions-consensus task dificult to solve.
We now give the formal definition of this new task.
First, 
we need to define the set of all valid input values for this new task. 

Let $I$ be a non-empty set of names such that $\nullval \notin I$ and
$\mathcal{N} = I \cup \{ \nullval \}$. If $x=\langle v_1, v_2 \rangle\in
\mathcal{N} \times \mathcal{N}$, $x.left$ denotes the value $v_1$ and $x.right$ is the value 
$v_2$. If $g\in \naturals$, then 
an ordered $g$-tuple $C= (x_1,\ldots,x_g)\in (\mathcal{N} \times \mathcal{N})^g$ is a
\emph{$g$-coalitions tuple} if and only if
\begin{itemize}
 \item [(a)] For all $x_i$ in $C$, $x_i.left \neq \nullval$ or $x_i.right \neq \nullval$.
\item [(b)] If $x_i,x_j$ are members of $C$ such that $x_i.left \neq \nullval$ and $x_j.left \neq
\nullval$, then $x_i.left = x_j.left$. A similar rule must hold if $x_i.right \neq \nullval$ and
$x_j.right \neq \nullval$.
\end{itemize}
Let $l(C)=\{ i\in\overline{g} \mid x_i.left \neq \nullval \}$ and $r(C)=\{ j\in \overline{g} \mid
x_j.right \neq \nullval \}$. Then we also require that 
\begin{itemize}
\item [(c)] $\lvert l(C) - r(C) \rvert = 1$ and $r(C) - l(C) = \{ g \}$.
\end{itemize}
Notice that property (a) implies that $l(C)\cup r(C)=\overline{g}$. The set of all $g$-coalitions 
tuples is denoted by $\mathcal{C}_g$. Some examples can help us understand better the 
concept of a coalition tuple. Let $\mathcal{N}=\{ 0,1, \nullval \}$. The following tuples are 
examples of elements of the set $\mathcal{C}_4$,
\begin{itemize}
 \item [] $C_1=(\langle 1, \nullval \rangle, \langle 1, 0 \rangle, \langle 1, 0 \rangle, \langle 
\nullval, 0 \rangle)$;
\item  [] $C_2=(\langle 1, 0 \rangle, \langle 1, 0 \rangle, \langle 1, \nullval \rangle,
\langle \nullval, 0 \rangle)$.
\end{itemize}

Clearly, $C_1,C_2\in \mathcal{C}_4$ and also we have that $l(C_1)=\{ 1,2,3 \}$ and $r(C_1)=\{ 2,3,4 
\}$. Notice that $l(C_1) \cup r(C_1)=\overline{4}$ and also $l(C_1) \cap r(C_1) = \{ 2,3 \}$. 
For $C_2$, it is true that $l(C_2) - r(C_2) = \{ 1,2,3 \} - \{ 
1,2,4 \}=\{ 3 \}$, $r(C_2) - l(C_2) = \{ 1,2,4 \} - \{ 1,2,3 \}=\{ 4 \}$ and $l(C_2) \cap r(C_2) 
= \{ 1,2 \}$. The following tuples are not elements of any coalitions tuples set:
\begin{itemize}
 \item [] $C_3=(\langle \nullval, \nullval \rangle, \langle 1, 0 \rangle, \langle 1, 0 
\rangle, 
\langle 1, 0 \rangle)$;
\item  [] $C_4=(\langle 1, \nullval \rangle, \langle 1, 0 \rangle, \langle \nullval, 1 \rangle, 
\langle \nullval, 0 \rangle)$.
\end{itemize}

For $C_3$, properties (a) and (c) are not satisfied and for 
 $C_4$, condition (b) is not fulfilled. Thus we conclude that $C_3,C_4 \notin \mathcal{C}_4$. We 
can now define a new distributed task.

\paragraph{The $g$-2coalitions-consensus task} 
We have $g$ processes
$p_1,\ldots,p_{g}$ and each process $p_i$ starts with a private input value of the form $x_i \in
\mathcal{N} \times \mathcal{N}$ such that $C=(x_1,\ldots,x_g)\in \mathcal{C}_g$,
and $C$ is called a \emph{global input}. 
In any execution, the following properties must be true:
\begin{itemize}
\item Termination: Each process must eventually output some value.
\item Agreement: All processes output the same value.
\item 2coalitions-Validity: If some process outputs $c$, then there must exists a process 
$p_j$ with input
$x_j$ such that $x_j = \langle c, v \rangle$ or $x_j = \langle v, c \rangle$ ($c\in I, v\in
\mathcal{N}$).
\end{itemize}

Notice that the sets 
$\mathcal{C}_g$
($g\geqslant 2$) encapsulate the intuitive ideas given at the beginning of this section about the 
coalition-consensus task that we need: With an input 
$C\in \mathcal{C}_g$, 
there are two 
processes that know exactly one of the input values, while the rest of the processes do know the 
two input values. These processes could easily reach a consensus, but they still need the other two 
processes with unique input values to agree on the same value, and this is the difficult part of the 
$g$-2coalitions-consensus 
task.

The  protocol of Figure~\ref{figWORprotk1coalitionconsensus}
implements  $g$-2coalitions-consensus.
Each process $p_i$ receives as input a tuple with values satisfying the properties of the 
2coalitions-consensus task and then in lines \ref{GCO3}-\ref{GCO5}, $p_i$ writes its input tuple in 
shared memory using the snapshot object $SM$; invokes the safe-consensus object with its id as 
input, storing the unique output value $u$ of the shared object in the local variable $val$ and 
finally, $p_i$ takes a snapshot of the memory. Later, what happens in Lines \ref{GCO6}-\ref{GCO10} 
depends on the output value $u$ of the safe-consensus object. If $u=g$, then by the Safe-Validity 
property, either $p_g$ invoked the object or at least two processes invoked the safe-consensus 
object concurrently and as there is only one process with input tuple $\langle v, \nullval 
\rangle$, $p_i$ will find an index $j$ with $sm\left[ j \right].right \neq \nullval$ in line 
\ref{GCO7}, assign this value to $dec$ and in line \ref{GCO11} $p_i$ decides. On the other hand, 
if $u\neq g$, then again by the Safe-Validity condition of the safe-consensus task, either process 
$p_u$ is running and invoked the safe-consensus object or two or more processes invoked concurrently 
the shared object and because all processes with id not equal to $g$ have input tuple $\langle z, y 
\rangle$ with $z\neq \nullval$, it is guaranteed that $p_i$ can find an index $j$ with $sm\left[ 
j \right].left \neq \nullval$ and assign this value to $dec$ to finally execute line \ref{GCO11} 
to decide its output value. 
All processes decide the same value because of the properties of the input tuples of the 
2coalitions-consensus task and the Agreement property of the safe-consensus task.

\begin{figure}[htb]
\centering{ 
\fbox{
\begin{minipage}[t]{150mm}
\footnotesize
\small
\renewcommand{\baselinestretch}{2.5}
\resetline
\begin{tabbing}
aaaaa\=aaa\=aaa\=aaa\=aa\=aa\kill 

\line{GCOPROC} \> \textbf{procedure} $g\text{-}2coalitions\text{-}consensus(v_1,v_2)$ \\
\line{GCO2} \> \textbf{begin} \\

\line{GCO3} \> \> $SM.\writesnapshot(\langle v_1,v_2 \rangle)$; \\

\line{GCO4} \> \> $val\leftarrow safe\text{-}consensus.exec(id)$; \\

\line{GCO5} \> \> $sm\,\leftarrow SM.\readsnapshot()$; \\

\line{GCO6} \> \> {\bf if} $val = g$  {\bf then} \\

\line{GCO7} \> \> \> $dec \leftarrow \text{choose any } sm\left[ j \right].right \neq \nullval$;
\\

\line{GCO8} \> \> {\bf else}  \\

\line{GCO9} \> \> \> $dec \leftarrow \text{choose any } sm\left[ j \right].left \neq \nullval$;
\\

\line{GCO10} \> \> {\bf end if}  \\

\line{GCO11} \> \> \textbf{decide} $dec$; \\

\line{GCO12} \> \textbf{end}
\end{tabbing}
\normalsize
\end{minipage}
}
\caption{A $g$-2coalitions-consensus protocol with one safe-consensus object.
}
\label{figWORprotk1coalitionconsensus}
}
\end{figure}

We now prove that this $g$-2coalitions-consensus 
protocol 
is correct.

\begin{lemma}\label{lemk1coalitionsconsWOR}
The protocol of Figure \ref{figWORprotk1coalitionconsensus} 
solves the
$g$-2coalitions-consensus task using one snapshot object and one safe-consensus object.
\end{lemma}

\begin{proof}
Let $p_i\in
\{p_1,\ldots,p_g\}$; after $p_i$ writes the tuple $\langle v_1,$ $v_2 \rangle$ to the snapshot 
object $SM$, it invokes the safe-consensus object and takes a snapshot of the shared memory, $p_i$
enters into the {\bf if}/{\bf else} block at lines \ref{GCO6}-\ref{GCO10}. Suppose that the test in
line \ref{GCO6} is successful, this says that the safe-consensus object returned the value $g$
to process $p_i$. By the Safe-Validity condition of the safe-consensus task, either process
$p_g$ invoked the shared object or at least two processes $p_j,p_k$ invoked the safe-consensus 
object concurrently (and it could happen that $j,k\neq g$). In any case, the participating 
processes wrote their input tuples to the snapshot object $SM$ before accessing the safe-consensus 
object, so that $p_i$ can see these values in its local variable $sm_i$, and remember that the 
coalitions tuple $C$ consisting of all the input values of processes $p_1,\ldots,p_g$ 
satisfies the properties
\begin{equation*}
\lvert l(C) - r(C) \rvert = 1 \text{ and } r(C) - l(C) = \{ g \}.
\end{equation*}
Thus, the equation $\left| l(C) - r(C) \right| = 1$ tells us that only one process has input tuple
$\langle x, \nullval \rangle$ $(x\neq \nullval)$, then when $p_i$ executes line \ref{GCO7}, it will
find a local register in $sm_i$ such that $sm_i\left[ j \right].right \neq \nullval$ and this value
is assigned to $dec_i$. Finally, $p_i$ executes line \ref{GCO11}, where it decides the value stored
in $dec_i$ and this variable contains a valid input value proposed by some process.

If the test at line \ref{GCO6} fails, the argument to show that $dec_i$ contains a valid proposed
input value is very similar to the previous one. This proves that the protocol satisfies the
2coalitions-Validity condition of the $g$-2coalitions-consensus task.

The Agreement condition is satisfied because by the Agreement condition of the safe-consensus
task, all participating processes receive the same output value from the shared object and
therefore all processes have the same value in the local variables $val_i$ $(1\leqslant i
\leqslant g)$, which means that all participating processes execute line \ref{GCO7}
or all of them execute line \ref{GCO9}. Then, for every process $p_r$, $dec_r$
contains the value $sm_r\left[ j_r \right].right$ or the value $sm_r\left[ j^\prime_r \right].left$
where $j_r$ ($j_r^\prime$) depends on $p_r$. But because the input pairs of the processes constitute
a coalitions tuple $C$, they satisfy property (b) of the definition of coalitions tuple,
which implies that $sm_r\left[ j_r \right].right = sm_r\left[ j_q \right].right$ whenever
$sm_r\left[ j_r \right].right$ and $sm_r\left[ j_q \right].right$ are non-$\nullval$ (a similar
statement holds for the left sides). We conclude that all processes assign to the variables
$dec_i$ $(1\leqslant i \leqslant g)$ the same value and thus all of them decide the same output
value, that is, the Agreement property of the $g$-2coalitions-consensus tasks is fulfilled. The
Termination property is clearly satisfied.
\end{proof}

\subsection{The WOR protocol for consensus}\label{secTheConsensusProtocol}

We first give an intuitive description of the WOR 
protocol for consensus, using 2coalitions-consensus 
shared objects and then we introduce the formal 
specification and prove its correctness.

\subsubsection*{Intuitive description}

A simple way to describe the protocol that solves consensus is by seeing it as a 
protocol in which the processes use a set of $\binom{n}{2}$ shared objects 
which implement 
our new $g$-2coalitions-consensus task. 
In this way, 
the protocol in Figure \ref{figWORprotconsensus2} can be 
described graphically as shown in Figure \ref{figgraphicconsensusprotocol}, for the case of $n=4$. 
In each round of the protocol, some processes invoke a 2coalitions-consensus object, represented by 
the symbol $\eventfont{2CC}_i$. In round one, $p_1$ and $p_2$ invoke the object $\eventfont{2CC}_1$ 
with input values $\langle v_1, \nullval \rangle$ and $\langle \nullval, v_2 \rangle$ respectively, 
(where $v_i$ is the initial input value of process $p_i$) and the consensus output $u_1$ of 
$\eventfont{2CC}_1$ is stored by $p_1$ and $p_2$ in some local variables. In round two, $p_2$ and 
$p_3$ invoke the $\eventfont{2CC}_2$ object with inputs $\langle v_2, \nullval \rangle$ and 
$\langle \nullval, v_3 \rangle$ respectively and they keep the output value $u_2$ in local 
variables. Round three is executed by $p_3$ and $p_4$ in a similar way, to obtain the consensus 
value $u_3$ from the 2coalition-consensus object $\eventfont{2CC}_3$. At the beginning of round 
four, $p_1,p_2$ and $p_3$ gather the values $u_1,u_2$ obtained from the objects $\eventfont{2CC}_1$ 
and $\eventfont{2CC}_2$ to invoke the $\eventfont{2CC}_4$ 2coalition-consensus object with the 
input values $\langle u_1, \nullval \rangle, \langle u_1, u_2 \rangle$ and $\langle \nullval, u_2 
\rangle$ respectively (Notice that $p_2$ uses a tuple with both values $u_1$ and $u_2$) and they 
obtain a consensus value $u_4$. Similar actions are taken by the processes $p_2,p_3$ and $p_4$ in 
round five with the shared object $\eventfont{2CC}_5$ and the values $u_2,u_3$ to compute an 
unique value $u_5$. Finally, in round six, all processes invoke the last shared object 
$\eventfont{2CC}_6$, with the respective input tuples
\[ \langle u_4, \nullval \rangle,\langle u_4, u_5 \rangle,\langle u_4, u_5 \rangle, 
\langle \nullval, u_5 \rangle, \]
and the shared object returns to all processes an unique output value $u$, which is the decided 
output value of all processes, thus this is the final consensus of the processes.

\begin{figure}[htb]
\centering{
\includegraphics[scale=0.8]{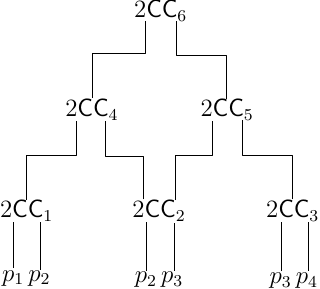}

\caption{The  structure of the 4-consensus protocol using 
2coalitions-consensus 
tasks.}
\label{figgraphicconsensusprotocol}
}
\end{figure}

\subsubsection*{Formal specification}

The formal specification of the iterated consensus protocol with safe-consensus objects is given in Figure 
\ref{figWORprotconsensus2}. This is a protocol that implements consensus using only $\binom{n}{2}$ 
2coalitions-consensus tasks. If we suppose that the protocol is correct, then we can use the 
$g$-2coalitions-consensus protocol presented in Section \ref{secCCTask} (Figure 
\ref{figWORprotk1coalitionconsensus}) to replace the call to the 2coalitions-consensus objects in 
Figure \ref{figWORprotconsensus2} to obtain a full iterated protocol that solves the consensus 
task using $\binom{n}{2}$ safe-consensus objects. 

\begin{figure}[htb]
\centering{ 
\fbox{
\begin{minipage}[t]{150mm}
\footnotesize
\small
\renewcommand{\baselinestretch}{2.5}
\resetline
\begin{tabbing}
aaaaa\=aaa\=aaa\=aaa\=aa\=aa\=aa\=aa\kill 
\line{CAW1} \> \textbf{init} $step,firstid,lastid \leftarrow 1$; \\
\> \> $C,D,dec,newagreement \leftarrow \nullval$;\\
\> \> $\thickspace\: agreements\left[\inthashmap{id}{id} \right] \leftarrow input$;

\\~\\
\line{CAW2} \> \textbf{begin} \\

\line{CAW3} \> \> \textbf{for} $r \leftarrow 1$ \textbf{to} $\binom{n}{2}$ \\

\line{CAW4E} \> \> \> $lastid \leftarrow firstid + step;$\\

\line{CAW4} \> \> \> \textbf{if} $firstid\leqslant id \leqslant lastid$ \textbf{then} \\


\line{CAW5} \> \> \> \> $C \leftarrow agreements\left[ \inthashmap{firstid}{lastid - 1} \right]$; \\

\line{CAW5E1} \> \> \> \> $D \leftarrow agreements\left[ \inthashmap{firstid + 1}{lastid} \right]$; 
\\

\line{CAW5E2} \> \> \> \> $newagreement \leftarrow 
2coalitions\text{-}consensus\left[r \right](C,D)$; \\

\line{CAW5E3} \> \> \> \> $agreements\left[ \inthashmap{firstid}{lastid} \right] = newagreement$; \\

\line{CAW13} \> \> \> {\bf end if}  \\

\line{CAW13E1} \> \> \> {\bf if} $lastid < n$ \textbf{then} \\

\line{CAW13E2} \> \> \> \> $firstid = firstid + 1$; \\

\line{CAW13E3} \> \> \> {\bf else if} $firstid > 1$ \textbf{then} \\

\line{CAW13E4} \> \> \> \> $firstid =  1$; \\

\line{CAW13E5} \> \> \> \> $step = step + 1$; \\

\line{CAW13E6} \> \> \> {\bf end if}  \\

\line{CAW14} \> \> \textbf{end for} \\

\line{CAW15E} \> \> $dec \leftarrow agreements\left[ \inthashmap{1}{n} \right]$; \\

\line{CAW15} \> \> \textbf{decide} $dec$; \\

\line{CAW16} \> \textbf{end}
\end{tabbing}
\normalsize
\end{minipage}
}
\caption{An iterated consensus protocol using $g$-2coalitions-consensus objects.
}
\label{figWORprotconsensus2}
}
\end{figure}

We now give a short description of how the protocol works. There are precisely $\binom{n}{2}$ 
rounds executed by the processes and in each round, some subset of processes try to agree in a new 
consensus value among two given input values. The local variables $step,firstid$ and $lastid$ are 
used by 
the processes to store information that tell them which is the current set of processes that must 
try to reach a new 
agreement in the current round, using a 2coalitions-consensus object (the symbol ''$id$`` contains 
the id of the process which is executing the code). The local array 
$agreements$ contains enough local registers used by the processes to store the agreements made in 
each round of the protocol and two distinct processes can have different agreements in 
$agreements$. Each consensus value $v$ stored in a register of $agreements$ is associated with 
two integers $i_1,i_r \in \overline{n}$ ($r\geqslant 1$), which represent a set of processes 
$p_{i_1},\ldots,p_{i_r}$ (with $i_1 < \cdots < i_r$) that have invoked a 2coalitions-consensus 
object to agree on the value $v$, thus we can say that $v$ is an agreement made by the coalition of 
processes represented by the pair $(i_1,i_r)$. To be able to store $v$ in the array $agreements$, 
the processes use a deterministic function $\inthashmapname \colon \naturals \times \naturals \to 
\naturals$, which maps bijectively $\naturals \times \naturals$ onto $\naturals$. This map can be 
easily constructed, for example, here is a simple definition for $\inthashmapname$
\begin{equation}
\inthashmap{i}{j} = \binom{i + j + 1}{2} + j. 
\end{equation}

Using all the elements described above, the processes can use the protocol of Figure 
\ref{figWORprotconsensus2} to implement consensus using $\binom{n}{2}$ 2coalitions-consensus 
objects, building in the process the structure depicted in Figure 
\ref{figgraphicconsensusprotocol}. From the first round 
and up to round $n - 1$, all the processes use their 
proposed input values to make new agreements in pairs, in 
round 
one, $p_1,p_2$ invoke a 2coalitions-consensus shared 
object to agree on a common value, based on 
their input values; in round 2, $p_2$ and $p_3$ do the 
same with another 2coalitions-consensus 
object and their own input values; later, the turn comes 
to $p_3$ and $p_4$ to do the same with 
another shared object and their input values and so on 
until the round number $n - 1$, where 
$p_{n-1}$ and $p_n$ agree on a common value in the way 
that we have already described. All the 
agreements obtained in these $n - 1$ rounds are stored on 
the local registers

\begin{multline}
agreements_1\left[ \inthashmap{1}{2} \right], agreements_2\left[ \inthashmap{1}{2} 
\right], agreements_2\left[ \inthashmap{2}{3} \right], \\
agreements_3\left[ \inthashmap{2}{3} 
\right] , agreements_3\left[ \inthashmap{3}{4} \right],agreements_4\left[ \inthashmap{3}{4} 
\right], 
\\ \ldots,  \\
agreements_{n - 1}\left[ \inthashmap{n - 1}{n} \right], agreements_{n}\left[ \inthashmap{n - 
1}{n} \right].
\end{multline}

In this way, the processes build the first part of the 
structure shown in Figure 
\ref{figgraphicconsensusprotocol}. At the end of round $n - 1$, in 
lines \ref{CAW13E4}-\ref{CAW13E5}, each 
process updates the values of the local variables $step$ and $firstid$ and proceeds to round $n$. 
What 
happens in the next $n - 2$ 
rounds, is very similar to the case of the previous $n - 1$ rounds, but instead of making 
agreements in pairs, the processes reach new agreements in groups of three processes (see 
Figure \ref{figgraphicconsensusprotocol}) invoking new 2coalitions-consensus shared objects and the 
consensus values obtained in the first $n - 1$ rounds and when each process reaches round $n - 1 + 
n - 2= 2n - 3$, it updates the values of its local variables $step$ and $firstid$ and then the 
processes 
proceed to round $2n - 2$. 

In general, when the processes are about to begin executing round 
$(\sum^m_{j=1} n - j) + 1$ ($m < n$), they will try to make $n - (m + 1)$ new agreements in groups 
of size $m+2$, with the aid of the 2coalitions-consensus objects and the agreements they obtained 
from the previous $n - m$ rounds and store the new consensus values in their local arrays 
$agreements$, using the $\inthashmapname$ function. When a process finishes round $(\sum^{m+1}_{j=1} 
n - j)$, the values of $step$ and $firstid$ are updated to repeat the described behaviour for the 
next $n - (m + 2)$ 
rounds, until the $\binom{n}{2}$ round, where the last agreement is made and this value is the 
output value of all the processes.

Now we are ready to give the full proof of Theorem \ref{thmWORconsensusbuenas}. We need a series of 
Lemmas.

\begin{lemma}\label{thmWORconsensus2}
The protocol in Figure \ref{figWORprotconsensus2} solves the consensus task
using $\binom{n}{2}$ $g$-2coa\-li\-tions-con\-sen\-sus objects.
\end{lemma}

\begin{proof}
The protocol clearly satisfies the Termination
condition of the consensus task (the only loop is finite). We prove that it fulfils the Agreement
and Validity conditions; to do this, we need some definitions and some intermediate results. 

With respect to the protocol of Figure \ref{figWORprotconsensus2}, let $F=\{ i_1, \ldots, i_r 
\}\subseteq \overline{n}$ be elements such 
that $i_1 < \cdots < i_r$. The set $F$ will be called a \emph{coalition} and will be denoted by 
$F=\left[i_1,\ldots,i_r \right]$. Let $p_i$ a process such that $i\in F$. We say that $p_i$ 
\emph{belongs to the coalition} $F$ if and only if 
 the following condition is fulfilled:
\begin{itemize}
\item [(1)] 
$agreements_{i}\left[ \inthashmap{i_1}{i_r} \right] = v$, where $v\neq \nullval$ is a valid input 
value proposed by some participating process.
\end{itemize}
We ask the following \emph{agreement} property to be satisfied by $F$:
\begin{itemize}
 \item [(2)] if $p_i,p_j$ are processes which belong to $F$ and $v,v^\prime$ are the values which 
satisfy (1) for $p_i$ and $p_j$ respectively, then $v = v^\prime$.
\end{itemize}
The value $v$ of the first condition is called the \emph{name} of the coalition. The second 
condition says that all processes that belong to the same coalition must agree on the coalition's
name. For $n>m\geqslant 0$, let $\gamma(n,m)$ be defined by
\[ \gamma(n,m)=\begin{cases}
                0 & \text{ if } m=0, \\
                \gamma(n, m - 1) + n - m & \text{ if } m > 0.
               \end{cases}\]
Notice that $\gamma(n,m)=\sum_{i=1}^{m} n -i$ for $m>0$ and $\gamma(n, n - 1)= \binom{n}{2}$.

\begin{lemma}\label{sublem1WORconsensus2}
 Let $1\leqslant m \leqslant n - 1$ and $1\leqslant c \leqslant n - m$. If $p_i$ is executing the
protocol at the beginning of round number $r=\gamma(n, m - 1) + c$ (before line \ref{CAW13E1}) then
$step_i = m$ and $firstid_i= c$.
\end{lemma}

\begin{proof}
 We prove this by induction on $m$. An easy analysis of the code in Figure
\ref{figWORprotconsensus2} shows that the base case holds (when $m=1$, in the first $n -1$ rounds).
Assume that for $m< n - 1$, the lemma is true. We first prove the following claim: When $p_i$ starts
executing round $\gamma(n,m) + 1$, $step_i = m+1$ and $firstid_i=1$. By the induction hypothesis, 
when
process $p_i$ executed the protocol at round $r^\prime = \gamma(n,m) = \gamma(n, m - 1) + c=
\gamma(n, m - 1) + (n - m)$ before line \ref{CAW13E1}, the local variables $step_i$ and $firstid_i$ 
had
the values of $m$ and $n-m$ respectively. When $p_i$ reached line \ref{CAW13E1}, it executed the
test of the {\bf if} statement, but before that, $p_i$ executed line \ref{CAW4E} of the protocol,
gather that, $lastid_i = firstid_i + step_i = (n - m) + m = n$; thus the test in line \ref{CAW13E1} 
failed
and $p_i$ executed lines \ref{CAW13E4}-\ref{CAW13E5} of the {\bf else} statement ($firstid_i> 1$ 
because
$m < n - 1$) and then $step_i$ was incremented by one and $firstid_i$ was set to 1 at the end of 
round
$r^\prime$. Therefore, when $p_i$ starts executing round $\gamma(n,m) + 1$, $step_i = m+1$ and
$firstid_i=1$.

Now suppose that $p_i$ executes the protocol at the beginning of round number $r = \gamma(n,m) + c$
where $1\leqslant c \leqslant n - (m + 1)$. If $c = 1$ then by the preceding argument, $step_i= m+
1$ and $firstid_i = 1 = c$. Using this as a basis for an inductive argument, we can prove that for
$c\in \{ 1, \ldots, n - (m + 1)\}$, $c=firstid_i$ and $step_i = m + 1$.
\end{proof}

\begin{lemma}\label{sublem2WORconsensus2}
 Let $0\leqslant m \leqslant n - 2$. Suppose that process $p_i$ is about to begin executing the
protocol at round $r=\gamma(n,m) + c$ $(1 \leqslant c \leqslant n - (m + 1))$. If $c\leqslant i
\leqslant c+m+1$ and $p_i$ belongs to the coalition $P= \left[ c,\ldots,c + m \right]$ or it
belongs to the coalition $Q= \left[ c+1,\ldots,c + m+ 1 \right]$, then at the end of round
$r$, $p_i$ belongs to the coalition $\left[ c,\ldots,c + m+ 1 \right]$.
\end{lemma}

\begin{proof}
 Let $p_i$ be any process that begins executing round $r$. By Lemma \ref{sublem1WORconsensus2}, we
know that $step_i = m + 1$ and $firstid_i = c$, which implies that $lastid_i=c + m + 1$. If $i\in \{ 
c,
\ldots, c+ m+1 \}$, the {\bf if}'s test at line \ref{CAW4} is successful and 
after that what happens in lines \ref{CAW5},\ref{CAW5E1} 
depends on $i$. If
$c\leqslant i \leqslant c + m$, then $p_i$ is in the coalition $P$ and if
$c+ 1\leqslant i \leqslant c + m+ 1$, $p_i$ is in the coalition $Q$, gather
that, a valid input value is assigned to at least one of the variables $C_i$ or $D_i$, so $p_i$
invokes in line \ref{CAW5E2} the $(m+2)$-2coalitions-consensus object with the input $x_i =
\langle C_i,D_i\rangle$. We now pause for a moment to argue that the tuple
$J=(x_c,\ldots,x_{c+m+1})$ built with the inputs of processes $p_c,\ldots,p_{c+m+1}$ is a valid
global input for the $(m+2)$-2coalitions-consensus task. Indeed, from the hypothesis, it is
easy to see that $J$ satisfies the requirements (a)-(c) of the definition of coalition tuple and
notice that 
$r(J) - l(J) = \{ c + m+1 \}$ and $l(J) - r(J) = \{ c \}$.

Now back to the execution of process $p_i$. After $p_i$ invokes the
$(m+2)$-2coalitions-consensus task, the return value of the shared object, say $v$, is
assigned to the local variable $newagree\-ment_i$. Finally, $p_i$ stores the contents of 
$newagreement_i$ in the local array $agreements_i$ at position $\inthashmap{firstid_i}{lastid_i}$, 
where 
$firstid_i=c$ and $lastid_i=c+m+1$. 
Because of the Agreement condition of the 2coalitions-consensus task, every process with id in the 
set $\{ c,\ldots, c+ m+ 1\}$ obtained the same value $v$ as return value from the shared object and 
by the 2coalitions-Validity, this is a valid input value proposed by some process. Therefore 
process 
$p_i$ belongs to the coalition $\left[ c,\ldots,c + m+ 1 \right]$ at the end of round $r$. 

On the other hand, if $i\notin \{ c, \ldots, c+ m+1 \}$, $p_i$ does not execute the body of the
{\bf if} statement (lines \ref{CAW5} to \ref{CAW5E3}) and it goes on to execute the {\bf if}/{\bf
else} block at lines \ref{CAW13E1}-\ref{CAW13E6} and then round $r$ ends for $p_i$, thus it does not
try to make a new coalition with other processes.
\end{proof}

\begin{lemma}\label{sublem3WORconsensus2}
Let $m\in \{0,\ldots,n - 2\}$. Suppose that process $p_j$ has exe\-cu\-ted the protocol for
$\gamma(n, m)$ rounds and that there exists $i\in \{ 1, \ldots, n - (m + 1) \}$ such that $p_j$
belongs to the coalition $\left[ i,\ldots, i + m\right]$ or it belongs to the coalition $\left[
i + 1, \ldots, i + m +1\right]$. Then after $p_j$ has executed the protocol for $n - (m+1)$ more
rounds, $p_j$ belongs to the coalition $\left[ i,\ldots, i + m + 1\right]$.
\end{lemma}

\begin{proof}
 We apply Lemma \ref{sublem2WORconsensus2} in each round $\gamma(n, m) + c$, where $c\in \{ 1,
\ldots,n - (m+1)\}$.
\end{proof}

Now we can complete the proof of Lemma \ref{thmWORconsensus2}. Let $p_j$ be any process that
executes the protocol. Just at the beginning of the first round (line \ref{CAW2}), process $p_j$
belongs to the coalition $\left[ j \right]$ (because of the assignments made to the local array
$agreements_j$ in line \ref{CAW1}, so that if $p_j$ executes the protocol for $\gamma(n,1)=n-1$
rounds, we can conclude using Lemma \ref{sublem3WORconsensus2}, that process $p_j$ belongs to some
of the coalitions $\left[ i, i+1 \right]$ $(1\leqslant i \leqslant n - 1)$. Starting from this fact
and using induction on $m$, we can prove that for all $m=1,\ldots,n - 1$; at the end of round
$\gamma(n,m)$, $p_j$ belongs to some of the coalitions $\left[ i,\ldots, i + m \right]$
$(1\leqslant i \leqslant n - m)$. In the last round (when $m=n -1$), after executing the main {\bf
for} block, process $p_j$ belongs to the coalition $T=\left[ 1,\ldots,n\right]$, thus when
$p_j$ executes line \ref{CAW15E}, it will assign to the local variable $dec_i$ a valid proposed
input value and this is the value decided by $p_j$ at line \ref{CAW15}. All processes decide the
same value because all of them are in the coalition $T$. Therefore the protocol satisfies the
Agreement and Validity conditions of the consensus
task.
\end{proof}

The final result of this section is obtained by combining
Lemmas \ref{lemk1coalitionsconsWOR} and 
\ref{thmWORconsensus2}.

\begin{theorem}\label{thmWORconsensusbuenas}
There exists a WOR protocol that solves the consensus task for $n$ processes using
$\binom{n}{2}$ safe-consensus objects.
\end{theorem}

\subsection{The lower bound}\label{asymtotictighbounds}

The main result of this paper is a matching lower bound on the number of 
safe-consensus objects needed to solve consensus using  safe-consensus. 
Our lower bound proof is based partly on standard bivalency arguments \cite{fischerImpossCon}, but in 
order to be able to apply them, a careful combinatorial work is necessary.

\subsubsection*{The connectivity of iterated protocols with 
safe-consensus.}\label{subsecConnWOR1NTSC}

Roughly speaking, a typical  consensus impossibility proof  
shows that a protocol $\mathcal{A}$ cannot solve consensus because there exists one execution of 
$\mathcal{A}$ in which processes decide a consensus value $v$ and a second execution of 
$\mathcal{A}$ where the consensus output of the processes is $v^\prime$, with $v\neq v^\prime$, such 
that the global states of these executions can be 
connected with paths of connected states. The existence of such paths will imply that in some 
execution of $\mathcal{A}$, some processes decide distinct output values 
\cite{fischerImpossCon,impossConSHM,HS99,getco10topologicaltheory}, violating the Agreement 
requirement of consensus. Any protocol that solves consensus, must be able to destroy 
these paths of connected states.

Remember from Section \ref{secSharedObjsCombinatorialSets} that the set $\Gamma_\mathcal{A}(n,m)$ encodes all
the valid subsets of processes of size $m$ that can invoke safe-consensus objects in $\mathcal{A}$. The cardinality 
of $\Gamma_\mathcal{A}(n,m)$ is denoted by $\nu_\mathcal{A}(n ,m)$ and 
\[\nu_\mathcal{A}(n)=\sum_{m=2}^n \nu_\mathcal{A}(n ,m).\]
For the case of our lower bound proof, the main property that will prevent $\mathcal{A}$ from 
solving consensus (i.e., from destroying paths of connected states) is that for some $m_0\in \{ 2,\ldots,n \}$, 
it is true that 
\begin{equation}\label{eqCantDestroyPaths}
	\nu_\mathcal{A}(n,m_0) \leqslant n - m_0,
\end{equation}
i.e., at most $n - m_0$ subsets of processes of size 
$m_0$ can invoke safe-consensus shared objects in the protocol $\mathcal{A}$. 
We are ready to present our main result, it is the following 

\begin{theorem}\label{lemminimumKboxesFULLWOR}
If $\mathcal{A}$ is a WOR protocol for $n$-consensus using safe-consensus objects, then for 
every $m\in \{ 2,\ldots, n \}$, $\nu_\mathcal{A}(n,m) > n - m$.
\end{theorem}

Theorem \ref{lemminimumKboxesFULLWOR} describes the minimum number of different  process groups, 
each of size $m$, that must invoke safe-consensus shared objects, in order to be able to solve the consensus 
task, for each $m = 2,\ldots, n$. The lower bound on the total number of safe-consensus objects 
necessary to implement consensus, is an easy consequence of Theorem \ref{lemminimumKboxesFULLWOR} 
and the definition of $\nu_\mathcal{A}(n)$. 

To understand better the implications of Theorem \ref{lemminimumKboxesFULLWOR} on the 
minimum number of safe-consensus shared objects needed to implement $n$-consensus, here is an 
example: Let $n = 5$ and suppose $\mathcal{A}$ is a WOR protocol that solves $5$-consensus.
By Theorem \ref{lemminimumKboxesFULLWOR}, for each $m= 2,3,4,5$, $\mathcal{A}$ satisfies the 
following inequalities:
\begin{itemize}
  \item[] $\nu_\mathcal{A}(5,2) > 5 - 2 = 3$,
  \item[] $\nu_\mathcal{A}(5,3) > 5 - 3 = 2$,
  \item[] $\nu_\mathcal{A}(5,4) > 5 - 4 = 1$,
  \item[] $\nu_\mathcal{A}(5,5) > 5 - 5 = 0$.
\end{itemize}

In words, There must be at least 4 groups of 2 processes that invoke sa\-fe-con\-sen\-sus ob\-jects, 
3 groups of 3 processes that invoke safe-consensus objects, 2 groups of 4 processes and finally, 
a group of 5 processes (e.g., all the processes) that must invoke a safe-consensus object. Thus the
minimum total number of safe-consensus shared objects that the processes must invoke in $\mathcal{A}$ 
to implement consensus is
\begin{eqnarray*}
  \nu_{\mathcal{A}}(5) & = & \nu_\mathcal{A}(5,2) + \nu_\mathcal{A}(5,3) + \nu_\mathcal{A}(5,4) + \nu_\mathcal{A}(5,5) \\
    & \geqslant & 4 + 3 + 2 + 1 \\
    & = & 10 \\
    & = & \binom{5}{2}.
\end{eqnarray*}
Generalizing this example for arbitrary $n$, we obtain the inequality $\nu_{\mathcal{A}}(n) \geqslant \binom{n}{2}$.
Notice that we are not counting trivial boxes (safe-consensus objects invoked by one process) in the lower bound.

Theorem \ref{lemminimumKboxesFULLWOR} will be proven by contradiction, 
that is, we will assume that for some $m_0$, Equation 
\eqref{eqCantDestroyPaths} holds. There are two structural results about 
WOR protocols that will be needed: Lemmas 
\ref{lemkboxesleqnmkconnFULLWOR} and \ref{thm2boxesleqnm2}. Basically, these 
two results tell us how to find paths of connected reachable states in 
every round of an executing WOR protocol, whenever Equation 
\eqref{eqCantDestroyPaths} is satisfied. The proof of each lemma relies on 
the development of various results about the connectivity of states of WOR 
protocols with safe-consensus shared objects and additionally, 
these results give us 
information about the 
combinatorial interactions that exists between the sets of processes which cannot distinguish 
between the states of a path and the boxes that represent safe-consensus objects invoked by the 
processes. This extra information is provided by two results about the 
connectivity of subgraphs of Johnson graphs, Lemma \ref{lemPartitionn2} 
and Theorem 
\ref{NoConnUmNoConnzetaUm}. Since the theory and prior results needed to prove 
these two results are purely combinatorial, we defer their formal proofs to 
Appendix \ref{appendixProofJohnsonGraphs}.


\subsubsection*{Further model terminology}\label{subsecfurthermodeldefs}

We will be using definitions and terminology given in 
Sections \ref{secdcBasicDef} and \ref{secAllModels}. 
Before giving the formal proof of our results, we need some technical definitions and lemmas.
First, 
we define a set of round schedules that will be very useful for the rest of the paper. 
Given $q\geqslant 1$ disjoint sets 
$A_1,\ldots,A_q\subset\overline{n}$, 
define the round schedule $\xschedule{A_1,\ldots,A_q}$ for $\mathcal{A}$ as:
\begin{equation}
\writeop{A_1},\scop{A_1},\readop{A_1},\ldots,
\writeop{A_q},\scop{A_q},\readop{A_q},\writeop{Y},\scop{Y},\readop{Y},
\end{equation}
where $Y=\overline{n} - (\bigcup_{i=1}^q A_i)$. 
For any state $S$ and $u\geqslant 0$, define
\[S\cdot \xschedulesup{u}{A_1,\ldots,A_q} = \begin{cases} S & \text{if } u = 0,\\
  (S\cdot \xschedulesup{u-1}{A_1,\ldots,A_q}) \cdot \xschedule{A_1,\ldots,A_q} &
\text{otherwise.} \end{cases}\]
I.e. $S\cdot \xschedulesup{u}{A_1,\ldots,A_q}$ is the state that we obtain after we run the
protocol $\mathcal{A}$ (starting from $S$) $u$ rounds with the round schedule
$\xschedule{A_1,\ldots,A_q}$ each iteration.

Let $\mathcal{A}$ be a WOR protocol 
for $n$-processes, $X\subseteq
\overline{n}$, $R$ a state of some round $r > 0$ and $b\in \setofboxes{R}$. We say that $R$ is a
\emph{ladder state for $X$ with step $b$}, if $R= S\cdot \xschedule{C_1,\ldots,C_u, B}$ $(u
\geqslant 0)$, where $S$ is a reachable state in $\mathcal{A}$ and
\begin{itemize}
\item $X=(\bigcup_j^u C_j) \cup B$;
 \item for $j=1,\ldots,u$, $0 \leqslant \lvert C_j \rvert \leqslant 2$;
\item $(\bigcup_j^u C_j) \cap b=\varnothing$;
\item $B = b \cap X$.
\end{itemize}
The following definition is given for convenience. For each box $b_i$ and $W \subseteq 
\overline{n}$, let the set $W_{(i)}$ be defined as 
\begin{equation}\label{eqDefWcapBoxBi}
 W_{(i)} = W \cap b_i.
\end{equation}

\begin{lemma}\label{lemWORsamesafeconfigurationAnySuccessor}
 Let $\mathcal{A}$ be a WOR protocol for $n\geqslant 2$ processes with sa\-fe-con\-sen\-sus
objects and $S$ a reachable state in $\mathcal{A}$. Then for any two round schedules $\pi_1,\pi_2$,
$\setofboxes{S\cdot \pi_1} = \setofboxes{S\cdot \pi_2}$.
\end{lemma}

\begin{proof}
In Figure \ref{figWORprot}, notice that process $p_i$ executes the deterministic function $h$ with
the local state it has from the previous round (except from the round counter $r$), so that if $r$ 
is the round number of $S$ and $p_i$ starts executing the protocol $\mathcal{A}$ in round $r+1$ with 
the local state it had in $S$, the input to the 
function $h$ is given by the tuple 
$(r+1, i, sm_S, val_S)$, where $sm_S,val_S$ depend only on $S$. Thus in both successor states 
$S\cdot \pi_1$ and $S\cdot \pi_2$, $p_i$ feeds the same input to $h$, gather that
\[ \setofboxes{S\cdot \pi_1} = \setofboxes{S\cdot \pi_2} \]
and the lemma is true.
\end{proof}

\subsection{The local structure of WOR protocols with safe-consensus
objects}\label{secStructureProtocolsWORSafeConsensus}

In this section, we present results that will be needed to prove Lemma 
\ref{lemkboxesleqnmkconnFULLWOR}. The goal here is to show Theorem \ref{lemlocalconnectXYGeneric}, which give us a way to connect two states of the form
$S\cdot \xschedule{X}$ and $S\cdot \xschedule{Y}$ for $X,Y\subseteq \overline{n}$ and 
any reachable state $S$ of a WOR protocol. Put it another way, with Theorem \ref{lemlocalconnectXYGeneric} we can connect for a given WOR protocol, a state that 
is obtained if the processes with ids in $X$ first update the memory, invoke 
safe-consensus objects and finally they snapshot the shared memory and after that, processes with ids in $\overline{n} - X$ perform the same actions, with a state that is obtained when processes with ids in $Y$ and ids in $\overline{n} - Y$ perform the same 
action in the same way. All this is done with a careful handling of the 
safe-consensus values returned by the shared objects to the processes, in order to keep
the degree of indistinguishability of the path as high as possible. The ladder
states defined previously play an important role in the construction of the path given 
by Theorem \ref{lemlocalconnectXYGeneric}.

\begin{lemma}
\label{lemStairWayProcessPhase1}
 Let $n\geqslant 2$, $\varnothing \neq X\subseteq \overline{n}$, $\mathcal{A}$ a WOR protocol
with sa\-fe-con\-sen\-sus objects, $S$ a state that is reachable in $\mathcal{A}$ in some round
$r\geqslant 0$ and $b\in \setofboxes{S_X}$, where $S_X=S\cdot \xschedule{X}$. Then
there exists a state $L$, such that $L$ is a ladder state for $X$ with step $b$, such that $S_X$
and $L$ are connected in round $r+1$ with a path of states $\mathfrak{p}$ with $\deg \mathfrak{p}
\geqslant n -2$.
\end{lemma}

\begin{proof}
 Let $\setofboxes{S_X}= \{ b_1,\ldots, b_q \}$ $(q\geqslant 1)$. By Lemma
\ref{lemWORsamesafeconfigurationAnySuccessor}, for any one-round successor state $Q$ of $S$,
$\setofboxes{Q}=\setofboxes{S_X}$ and we can write $\mathbf{Inv}_{S}$ instead of 
$\setofboxes{S_X}$. Without loss of generality, assume that
$b_1=b$. If $q=1$, then $b_1=\overline{n}$ and the result is immediate, because $S_X$ is a ladder
state for $X$ with step $\overline{n}$. Suppose that $q > 1$. Partition $X$ as $X=X_{(1)}
\cup \cdots \cup X_{(q)}$ and build the following sequence of connected states
\begin{multline}\label{eqlemlocalconnectStairwayP1Subsequence1}
  S_X \linkstates{\overline{n} - \laddersubsetX{2}{1}} S \cdot
\xschedule{\laddersubsetX{2}{1}, X - \laddersubsetX{2}{1}} \\
\linkstates{\overline{n} - \laddersubsetX{2}{2}} S\cdot
\xschedule{\laddersubsetX{2}{1},\laddersubsetX{2}{2}, X - (\laddersubsetX{2}{1} \cup
\laddersubsetX{2}{2})} \\
\linkstates{\overline{n} - \laddersubsetX{2}{3}} \cdots \linkstates{\overline{n} -
\laddersubsetX{2}{\alpha_2}} \\ 
S \cdot \xschedule{\laddersubsetX{2}{1},\laddersubsetX{2}{2},\ldots,\laddersubsetX{2}{\alpha_2}, X -
X_{(2)}},
\end{multline}
where $X_{(2)}=\bigcup_{i=1}^{\alpha_2} \laddersubsetX{2}{i}$ is a partition of $X_{(2)}$ such that
$\lvert \laddersubsetX{2}{j} \rvert =1$ for $j=2,\ldots,\alpha_2$. The set $\laddersubsetX{2}{1}$
has cardinality given by
\[ \lvert \laddersubsetX{2}{1} \rvert =\begin{cases}
      2 & \text{if } \lvert X_{(2)} \rvert > 1, \\
      \lvert X_{(2)} \rvert & \text{otherwise,}
     \end{cases}
 \]
 (notice that the $\laddersubsetX{2}{j}$'s and $\alpha_2$ depend on $b_2$ and $X_{(2)}$), 
and we choose the safe-consensus value of every box in the set $\mathbf{Inv}_{S}$ to be the same in
each state of the previous sequence. This can be done because of the way we partition $X_{(2)}$,
the election of the elements of the set $\laddersubsetX{2}{1}$ and the properties of the
safe-consensus task (Safe-Validity). We execute similar steps with box $b_3$, so that we
obtain the path
\begin{multline}\label{eqlemlocalconnectStairwayP1Subsequence2}
S \cdot \xschedule{\laddersubsetX{2}{1},\ldots,\laddersubsetX{2}{\alpha_2}, X - X_{(2)}} \\
\linkstates{\overline{n} - \laddersubsetX{3}{1}} S \cdot
\xschedule{\laddersubsetX{2}{1},\ldots,\laddersubsetX{2}{\alpha_2}, \laddersubsetX{3}{1},X
- (X_{(2)} \cup \laddersubsetX{3}{1})} \\
\linkstates{\overline{n} - \laddersubsetX{3}{2}} S \cdot
\xschedule{\laddersubsetX{2}{1},\ldots,\laddersubsetX{2}{\alpha_2},
\laddersubsetX{3}{1},\laddersubsetX{3}{2},X - (X_{(2)} \cup \laddersubsetX{3}{1}\cup
\laddersubsetX{3}{2})} \\
\linkstates{\overline{n} - \laddersubsetX{3}{3}} \cdots \linkstates{\overline{n} -
\laddersubsetX{3}{\alpha_3}} \\
S \cdot
\xschedule{\laddersubsetX{2}{1},\ldots,\laddersubsetX{2}{\alpha_2},\laddersubsetX{3}{1},\ldots,
\laddersubsetX{3}{\alpha_3},X - (X_{(2)} \cup X_{(3)})},
\end{multline}
where the $\laddersubsetX{3}{i}$'s and $\alpha_3$ depend on $b_3$ and $X_{(3)}$, just in the same
way the $\laddersubsetX{2}{j}$'s and $\alpha_2$ depend on $b_2$ and $X_{(2)}$ and each box has 
safe-consensus value equal to the value it has in the sequence of 
\eqref{eqlemlocalconnectStairwayP1Subsequence1}. We can repeat the very same steps for
$b_4,\ldots,b_q$ to obtain the sequence
\begin{multline}\label{eqlemlocalconnectStairwayP1Subsequence3}
S \cdot \xschedule{\laddersubsetX{2}{1},\ldots,\laddersubsetX{3}{\alpha_3},X - (X_{(2)} \cup
X_{(3)})} \\
  \linkstates{\overline{n} - \laddersubsetX{4}{1}} S\cdot
\xschedule{\laddersubsetX{2}{1},\ldots,\laddersubsetX{3}{\alpha_3},\laddersubsetX{4}{1},X -
(X_{(2)} \cup X_{(3)} \cup \laddersubsetX{4}{1})} \\
\linkstates{\overline{n} - \laddersubsetX{4}{2}} \cdots \linkstates{\overline{n} -
\laddersubsetX{q}{1}} \\ 
S \cdot \xschedule{\laddersubsetX{2}{1},\ldots,\laddersubsetX{q}{1}, X -
(\laddersubsetX{q}{1} \cup \bigcup\nolimits_{i=2}^{q-1} X_{(i)})} \\
\linkstates{\overline{n} - \laddersubsetX{q}{2}} \cdots \linkstates{\overline{n} -
\laddersubsetX{q}{\alpha_q}} S \cdot
\xschedule{\laddersubsetX{2}{1},\ldots,\laddersubsetX{q}{\alpha_q}, X_{(1)}}.
\end{multline}
It is easy to prove that $L=S \cdot
\xschedule{\laddersubsetX{2}{1},\ldots,\laddersubsetX{q}{\alpha_q}, X_{(1)}}$ is a ladder state for
$X$
with step $b_1$ and that each of the sequences of equations
\eqref{eqlemlocalconnectStairwayP1Subsequence1}, \eqref{eqlemlocalconnectStairwayP1Subsequence2} and
\eqref{eqlemlocalconnectStairwayP1Subsequence3} has indistinguishability degree no less that $n
-2$. Combining all these sequences, we obtain a new path 
\[ \mathfrak{p} \colon S_{X} \linkstates{} \cdots \linkstates{} L \] 
with $\deg \mathfrak{p} \geqslant n - 2$. 
\end{proof}

\begin{lemma}
\label{lemStairWayProcessPhase2}
 Let $n\geqslant 2,X,Y \subseteq \overline{n}$, $\mathcal{A}$ a WOR protocol with
safe-consensus objects and $S$ a state that is reachable in $\mathcal{A}$ in some round $r\geqslant
0$. Assume also that $b_j$ is a box representing a safe-consensus object invoked by some processes 
such that $b_j\in \setofboxes{L_1} = \setofboxes{L_2}$, where $L_1 = 
S\cdot\xschedule{C_1,\ldots,C_u,X_{(j)}}$ is a ladder state for $X$ with step $b_j$ and $L_2 = 
S\cdot \xschedule{C_1,\ldots,C_u, Y_{(j)}}$ is a ladder state for $(X - X_{(j)}) \cup Y_{(j)}$ with 
step $b_j$. Finally, suppose that if $b_j=\overline{n}$, $\scvalof{b_j}{L_1}=\scvalof{b_j}{L_2}$. 
Then $L_1$ and $L_2$ are connected in round $r+1$ with a \regularpath path of states $\mathfrak{p}$
that satisfies the following properties:
\begin{itemize}
 \item [i)] If $\scvalof{b_j}{L_1}=\scvalof{b_j}{L_2}$ then $\deg \mathfrak{p} \geqslant n - 2$.
\item  [ii)] If $\scvalof{b_j}{L_1} \neq \scvalof{b_j}{L_2}$ then $\deg \mathfrak{p}  \geqslant n -
\lvert b_j \rvert$ and $\overline{n} - b_j \in \isets{\mathfrak{p}}$.
\end{itemize}
\end{lemma}

\begin{proof}
Partition $\overline{n}$ as the disjoint union $\overline{n}=X^- \cup (X \cap Y) \cup Y^- \cup W$, where 
\begin{itemize}
\item [] $X^-=X - Y$;
\item [] $Y^- = Y - X$;
\item [] $W =\overline{n} - (X^- \cup (X \cap Y) \cup Y^-)$. 
\end{itemize}
What we need to do to go from $L_1$ to $L_2$ is to ``interchange'' $X_{(j)}$ with $Y_{(j)}$. We 
first construct a sequence of connected states $\mathfrak{p}_1$ given by
\begin{multline*} 
S \cdot \xschedule{C_1,\ldots,C_u, X_{(j)},Y^- \cup W} \\
\linkstates{\overline{n} -
\laddersubsetX{j}{1}} S \cdot \xschedule{C_1,\ldots,C_u, \laddersubsetX{j}{1}, X_{(j)} -
\laddersubsetX{j}{1},Y^- \cup W} \\ 
\linkstates{\overline{n} - \laddersubsetX{j}{2}} S \cdot \xschedule{C_1,\ldots,C_u,
\laddersubsetX{j}{1},\laddersubsetX{j}{2}, X_{(j)} - (\laddersubsetX{j}{1}\cup
\laddersubsetX{j}{2}),Y^- \cup W} \\
\linkstates{\overline{n} - \laddersubsetX{j}{3}} \cdots \linkstates{\overline{n} -
\laddersubsetX{j}{\alpha_j-1}} S \cdot \xschedule{C_1,\ldots,C_u,
\laddersubsetX{j}{1},\laddersubsetX{j}{2},\ldots,\laddersubsetX{j}{\alpha_j},Y^- \cup W} \\
\linkstates{\overline{n} - \laddersubsetX{j}{\alpha_j}} S \cdot \xschedule{C_1,\ldots,C_u,
\laddersubsetX{j}{1},\laddersubsetX{j}{2},\ldots,\laddersubsetX{j}{\alpha_j - 1},
\laddersubsetX{j}{\alpha_j} \cup Y^- \cup W} \\
\linkstates{\overline{n} - \laddersubsetX{j}{\alpha_j - 1}} \cdots \linkstates{\overline{n}
- \laddersubsetX{j}{1}} S \cdot \xschedule{C_1,\ldots,C_u,X_{(j)} \cup Y^- \cup W},
\end{multline*}
where the following properties hold:
\begin{itemize}
\item $X_{(j)}=\bigcup_{i=1}^{\alpha_j} \laddersubsetX{j}{i}$ is a partition of $X_{(j)}$ such that
$\lvert \laddersubsetX{j}{l} \rvert =1$ for $l=2,\ldots,\alpha_j$;  
\item $\lvert \laddersubsetX{j}{1} \rvert = 2$ if $\lvert X_{(j)} \rvert > 1$ and $\lvert
\laddersubsetX{j}{1} \rvert = \lvert X_{(j)} \rvert$ otherwise; 
\item $\mathfrak{p}_1$ is a \regularpath sequence (Lemma
\ref{lemWORsamesafeconfigurationAnySuccessor}) with $\deg \mathfrak{p}_1 \geqslant n - 2$;
\item the safe-consensus value of every box $b_i$ is the same in each state of $\mathfrak{p}_1$.
This can be achieved by a proper selection of elements of the set $\laddersubsetX{j}{1}$ and the
Safe-Validity property of the safe-consensus task.
\end{itemize}
Now, since the following equalities hold: 
\begin{eqnarray*}
  X_{(j)} \cup Y^- & = & (X \cap b_j) \cup Y^- \\
    & = & ((X^- \cup (X \cap Y)) \cap b_j) \cup Y^- \\
    & = & X^-_{(j)} \cup (X \cap Y \cap b_j) \cup Y^- \\
    & = & X^-_{(j)} \cup (X \cap Y \cap b_j) \cup (Y^- - Y^-_{(j)}) \cup Y^-_{(j)} \\
    & = & (X \cap Y \cap b_j) \cup Y^-_{(j)} \cup X^-_{(j)} \cup (Y^- - Y^-_{(j)}) \\
    & = & Y_{(j)} \cup X^-_{(j)} \cup (Y^- - Y^-_{(j)}),
\end{eqnarray*}
we can write the state $S \cdot \xschedule{C_1,\ldots,C_u,X_{(j)} \cup Y^- \cup W}$ as 
\[S \cdot
\xschedule{C_1,\ldots,C_u,Y_{(j)} \cup X^-_{(j)} \cup (Y^- - Y^-_{(j)}) \cup W}.\]  
We need to build
a second path $\mathfrak{p}_2$ as follows:
\begin{multline}\label{eqlemlocalconnectStairwayP2SubsequenceBETTER2}
 S \cdot \xschedule{C_1,\ldots,C_u,Y_{(j)} \cup X^-_{(j)} \cup (Y^- - Y^-_{(j)})\cup W} \\
\linkstates{\overline{n} - U} S \cdot
\xschedule{C_1,\ldots,C_u,\laddersubsetY{j}{1},(Y_{(j)} - \laddersubsetY{j}{1}) \cup X^-_{(j)} \cup
(Y^- - Y^-_{(j)}) \cup W} \\ 
\linkstates{\overline{n} - \laddersubsetY{j}{2}} \cdots \linkstates{\overline{n} -
\laddersubsetY{j}{\epsilon_j - 1}} \\ 
 S \cdot \xschedule{C_1,\ldots,C_u,\laddersubsetY{j}{1},\ldots,\laddersubsetY{j}{\epsilon_j} \cup
X^-_{(j)} \cup (Y^- - Y^-_{(j)}) \cup W}
\\
\linkstates{\overline{n} - \laddersubsetY{j}{\epsilon_j}} S \cdot
\xschedule{C_1,\ldots,C_u,\laddersubsetY{j}{1},\ldots,\laddersubsetY{j}{\epsilon_j},X^-_{(j)} \cup
(Y^- - Y^-_{(j)}) \cup W} \\
\linkstates{\overline{n} - \laddersubsetY{j}{\epsilon_j - 1}} \\
S \cdot \xschedule{C_1,\ldots,C_u,\laddersubsetY{j}{1},\ldots,\laddersubsetY{j}{\epsilon_j
- 1} \cup \laddersubsetY{j} {\epsilon_j},X^-_{(j)} \cup (Y^- - Y^-_{(j)}) \cup W} \\
\linkstates{\overline{n} - \laddersubsetY{j}{\epsilon_j - 2}} \cdots \linkstates{\overline{n} -
\laddersubsetY{j}{1}} S \cdot \xschedule{C_1,\ldots,C_u,Y_{(j)},X^-_{(j)} \cup (Y^- - Y^-_{(j)})
\cup W}.
\end{multline}
Let $L_2$ be the last state of the previous sequence. The next assertions are true for the path
$\mathfrak{p}_2$: 
\begin{itemize}
\item The sets $\laddersubsetY{j}{i}$ and $\epsilon_j$ are defined for $Y_{(j)}$ and $b_j$ in the
same way as the $\laddersubsetX{j}{i}$ and $\alpha_j$ are defined for $X_{(j)}$ and $b_j$;
\item The sequence $\mathfrak{p}_2$ is \regularpath (Lemma 
\ref{lemWORsamesafeconfigurationAnySuccessor});
\item The safe-consensus value of every box $c\neq b_j$ is the same in every element of
$\states{\mathfrak{p}_2}$;
\item $\scvalof{b_j}{Q}=\scvalof{b_j}{P}$ for all $Q,P\in \states{\mathfrak{p}_2} - \{ R \}$, where $R$ is such that
\[R=S \cdot \xschedule{C_1,\ldots,C_u,Y_{(j)} \cup X^-_{(j)} \cup (Y^- - Y^-_{(j)})\cup W};\]
 \item the set $U$ is defined by
\begin{equation}\label{eqlemlocalconnectStairwayP2SubsequenceBETTER2_Prorperties_U}
U=\begin{cases}
   b_j & \text{if } \scvalof{b_j}{L_1} \neq \scvalof{b_j}{L_2} \\
\laddersubsetY{j}{1} & \text{otherwise;}
  \end{cases}
\end{equation}
\end{itemize}
Notice that by the last assertion, we can deduce that $\deg \mathfrak{p}_2 \geqslant n - \lvert b_j
\rvert \text{ and } \overline{n} - b_j \in \isets{\mathfrak{p}_2}$ if $\scvalof{b_j}{L_1} $ $\neq
\scvalof{b_j}{L_2}$ and $\deg \mathfrak{p}_2 \geqslant n - 2$ when $\scvalof{b_j}{L_1} =
\scvalof{b_j}{L_2}$. Thus we can use $\mathfrak{p}_1$ and $\mathfrak{p}_2$ to obtain a
\regularpath sequence $\mathfrak{p}$ which fulfils properties i)-ii) of the lemma and that
concludes the proof.
\end{proof}

\begin{lemma}
\label{lemCompleteStairWayProcessPhase}
 Let $n\geqslant 2$, $X,Y \subseteq \overline{n}$,
$\mathcal{A}$ a WOR protocol with safe-consensus objects, $S$ a state that is reachable in
$\mathcal{A}$ in some round $r\geqslant 0$ and $b_j$ a box representing 
a safe-consensus object invoked by some processes in round $r+1$ such that $b_j\in 
\setofboxes{Q_1} = \setofboxes{Q_2}$, where $Q_1 = S\cdot \xschedule{X}$ and $Q_2 = S\cdot 
\xschedule{Y_{(j)} \cup (X - X_{(j)})}$. Assume also that if $b_j=\overline{n}$, $\scvalof{b_j}{Q_1} 
= \scvalof{b_j}{Q_2}$. Then the states $Q_1$ and
$Q_2$ are connected in round $r+1$ with a \regularpath path of states $\mathfrak{p}$
that satisfies the following properties:
\begin{itemize}
 \item [a)] If $\scvalof{b_j}{Q_1} = \scvalof{b_j}{Q_2}$ then $\deg \mathfrak{p} \geqslant n - 2$.
\item  [b)] If $\scvalof{b_j}{Q_1} \neq \scvalof{b_j}{Q_2}$ then $\deg \mathfrak{p}  \geqslant n -
\lvert b_j \rvert$ and $\overline{n} - b_j \in \isets{\mathfrak{p}}$.
\end{itemize}
\end{lemma}

\begin{proof}
 By Lemma \ref{lemStairWayProcessPhase1}, the state $Q_1$ can be connected with a
state of the form $Q_{X_{(j)}}=S\cdot \xschedule{B_{1},\ldots,B_{s},X_{(j)}}$ with a \regularpath
path $\mathfrak{q}_1$ such that $\deg \mathfrak{q}_1 \geqslant n - 2$. Using Lemma
\ref{lemStairWayProcessPhase2}, $Q_{X_{(j)}}$ can be connected with $Q_{Y_{(j)}}=S\cdot
\xschedule{B_{1},\ldots,B_{s},Y_{(j)}}$ by means of a \regularpath sequence $\mathfrak{q}_2\colon
Q_{X_{(j)}}\linkstates{} \cdots \linkstates{} Q_{Y_{(j)}}$ such that $\mathfrak{q}_2$ satisfies
properties i) and ii) of that Lemma. And we can apply Lemma \ref{lemStairWayProcessPhase1} to 
connect $Q_{Y_{(j)}}$ with $Q_2$ with the
\regularpath path $\mathfrak{q}_3$ which has indistinguishability degree no less that $n - 2$. 
Therefore the sequence of connected states built by first gluing together the sequences
$\mathfrak{q}_1$ and $\mathfrak{q}_2$, followed by $\mathfrak{q}_3$, is a \regularpath sequence
$\mathfrak{q}$ such that the requirements  a)-b) are satisfied.
\end{proof}

With the three previous lemmas at hand, we are almost ready to state and prove Theorem 
\ref{lemlocalconnectXYGeneric}, only one more definition is necessary. 
Let $Q_1,Q_2$ be two reachable 
states in round $r$ of an iterated protocol $\mathcal{A}$ for $n$ processes with safe-consensus 
objects. 
The set
$\mathfrak{D}_\mathcal{A}^r(Q_1,Q_2)$ is defined as
\begin{equation*}
 \mathfrak{D}_\mathcal{A}^r(Q_1,Q_2) = \{ b\in \Gamma_\mathcal{A}(n) \mid b \in \setofboxes{Q_1}
\cap \setofboxes{Q_2} 
\text{ and } 
\scvalof{b}{Q_1} \neq \scvalof{b}{Q_2} \}.
\end{equation*}

In words, $\mathfrak{D}_\mathcal{A}^r(Q_1,Q_2)$ represents the safe-consensus shared 
objects that the same subset of processes invoked in both states $Q_1$ and $Q_2$, 
such that the shared objects returned different output values in each state. 
If there is no confusion about which round number and protocol we refer to, we will
write $\mathfrak{D}_\mathcal{A}^r(Q_1,Q_2)$ as $\mathfrak{D}(Q_1,Q_2)$.

\begin{theorem}\label{lemlocalconnectXYGeneric}
 Let $n\geqslant 2$ and $X,Y \subseteq \overline{n}$,
 $\mathcal{A}$ a WOR protocol with safe-consensus objects, $S$ a reachable state of
$\mathcal{A}$ in some round $r\geqslant 0$ and let $Q_1=S\cdot \xschedule{X}$ and $Q_2=S\cdot
\xschedule{Y}$ be such that $\overline{n} \notin \mathfrak{D}(Q_1,Q_2)$. Then
$Q_1$ and $Q_2$ are connected in round $r+1$ with a \regularpath path of states $\mathfrak{p}$ such
that
\begin{itemize}
\item[] (A) If $\mathfrak{D}(Q_1,Q_2)=\varnothing$, then $\deg \mathfrak{p}
\geqslant n - 2$.
\item[] (B) If the set $\mathfrak{D}(Q_1,Q_2)$ is not empty, then
\begin{enumerate}
\item $\deg \mathfrak{p} \geqslant \min \{ n -\lvert b \rvert \}_{b\in
\mathfrak{D}(Q_1,Q_2)}$;
\item for every $Z\in \isets{\mathfrak{p}}$ with $\lvert Z \rvert < n - 2$, there exists
 $b\in \mathfrak{D}(Q_1,$ $Q_2)$ such that $Z=\overline{n}- b$.
\end{enumerate}
\end{itemize}
\end{theorem}

\begin{proof}
By Lemma \ref{lemWORsamesafeconfigurationAnySuccessor}, we
have the following equation,
\[\mathfrak{D}(P_1,P_2) = \{
b\in \Gamma_\mathcal{A}(n) \mid b \in \mathbf{Inv}_{S}
\text{ and } \scvalof{b}{P_1} \neq \scvalof{b}{P_2}\},\] 
where $P_l$ is a one-round successor state of $S$ and $\mathbf{Inv}_{S}=\setofboxes{P_l}$,
$l=1,2$. Let $\mathbf{Inv}_{S}=\{ b_1, \ldots, b_q \}$. By Lemma 
\ref{lemCompleteStairWayProcessPhase}, we can connect the state $Q_1$ with $R_1=S\cdot
\xschedule{Y_{(1)} \cup (X - X_{(1)})}$ using a \regularpath path $\mathfrak{p}_1\colon Q_1
\linkstates{} \cdots \linkstates{} R_1$ such that
\begin{itemize}
 \item[] $(A_1)$ If $b_1\notin \mathfrak{D}(Q_1,R_1)$, then $\deg \mathfrak{p}_1
\geqslant n - 2$.
\item[] $(B_1)$ If $b_1\in \mathfrak{D}(Q_1,R_1)$, then $\deg \mathfrak{p}_1 
\geqslant n - \lvert b_1 \rvert$ and $\overline{n} - b_1 \in \isets{\mathfrak{p}_1}$.
\end{itemize}
Assume
there is a set $Z\in \isets{\mathfrak{p}_1}$ with size strictly less that $n - 2$.
Then it must be true that $b_1\in \mathfrak{D}(Q_1,R_1)$, because if $b_1\notin
\mathfrak{D}(Q_1,R_1)$, by $(A_1)$, $\lvert Z \rvert \geqslant n - 2$ and this
is impossible. Thus the conclusion of property $(B_1)$ holds for $\mathfrak{p}_1$, which means that
$\overline{n} - b_1 \in \isets{\mathfrak{p}_1}$. Examining the proof\footnote{
  In the proof of Lemma \ref{lemCompleteStairWayProcessPhase}, $\mathfrak{p}_1$ is build using 
  three subpaths $\mathfrak{q}_1,\mathfrak{q}_2$ y $\mathfrak{q}_3$. The paths $\mathfrak{q}_1$ 
  and $\mathfrak{q}_3$ have indistinguishability degree at least $n - 2$ (Lemma \ref{lemStairWayProcessPhase1}).
  The path $\mathfrak{q}_2$ is build using Lemma \ref{lemStairWayProcessPhase2} and we can check that 
  every set  $X \in \states{\mathfrak{q}_2}$ satisfies $| X| \geqslant n - 2$, except at most, the set 
  $\overline{n} - U$ in Equation \eqref{eqlemlocalconnectStairwayP2SubsequenceBETTER2}, since $U$ satisfies 
  Equation \eqref{eqlemlocalconnectStairwayP2SubsequenceBETTER2_Prorperties_U}.
}
 of Lemma
\ref{lemCompleteStairWayProcessPhase} we can convince ourselves that every $W\in
\isets{\mathfrak{p}_1}$ such that $W\neq \overline{n} - b_1$ has cardinality at least $n - 2$, 
therefore, $Z=\overline{n} - b_1$.

Now, considering that\footnote{Remember that $Y_{(1)} - X_{(2)}=(Y \cap b_1) - (X \cap b_2)= Y \cap b_1= Y_{(1)}$, 
since $b_1 \cap b_2 = \varnothing$.} 
$Y_{(2)} \cup ((Y_{(1)} \cup (X - X_{(1)})) - X_{(2)}) = 
Y_{(2)} \cup ((Y_{(1)} - X_{(2)}) \cup ((X - X_{(1)})- X_{(2)})) = 
Y_{(1)} \cup Y_{(2)} \cup (X - (X_{(1)} \cup X_{(2)}))$, 
we can apply Lemma \ref{lemCompleteStairWayProcessPhase} to the states $R_1$ and $R_2=S\cdot
\xschedule{Y_{(1)} \cup Y_{(2)} \cup (X - (X_{(1)} \cup X_{(2)}))}$ to find a \regularpath sequence
$\mathfrak{p}_2 \colon R_1 \linkstates{} \cdots \linkstates{} R_2$ such that $\mathfrak{p}_2$ and
$b_2$ enjoy the same properties that $\mathfrak{p}_1$ and $b_1$ have. We can combine the paths
$\mathfrak{p}_1$ and $\mathfrak{p}_2$ to obtain a \regularpath sequence $\mathfrak{p}_{12}\colon Q_1
\linkstates{} \cdots \linkstates{} R_2$ from $Q_1$ to $R_2$ satisfying the properties
\begin{itemize}
 \item[] $(A_2)$ If $\{ b_1, b_2 \} \cap \mathfrak{D}(Q_1,R_2) = \varnothing$,
then $\deg \mathfrak{p}_{12} \geqslant n - 2$.
\item[] $(B_2)$ If $\{ b_1, b_2 \} \cap \mathfrak{D}(Q_1,R_2)$ is not empty,
then 
\begin{itemize}
\item[] 1. $\deg \mathfrak{p}_{12} \geqslant \min \{ n -\lvert b \rvert \}_{b\in \{ b_1, b_2 \} \cap
\mathfrak{D}(Q_1,R_2)}$;
\item[] 2. for every $Z\in \isets{\mathfrak{p}_{12}}$ with $\lvert Z \rvert < n - 2$, 
there exists $b\in \{ b_1, b_2 \} \cap \mathfrak{D}(Q_1,$ $R_2)$ 
such that $Z=\overline{n}-b$.
\end{itemize}
\end{itemize}
We can repeat this process for all $s\in \{1, \ldots, q \}$. In general, if $R_s=S \cdot
\xschedule{(\bigcup_{i}^s Y_{(i)}) \cup (X - (\bigcup_i^s X_{(i)})) }$, we can construct a
\regularpath path 
\[ \mathfrak{p}_{1s}\colon Q_1 \linkstates{} \cdots \linkstates{} R_s, \]
with the properties
\begin{itemize}
 \item[] $(A_s)$ If $\{ b_1, \ldots, b_s \} \cap \mathfrak{D}(Q_1,R_s) =
\varnothing$, then $\deg \mathfrak{p}_{1s} \geqslant n - 2$.
\item[] $(B_s)$ If $\{ b_1, \ldots, b_s \} \cap \mathfrak{D}(Q_1,R_s)$ is not
empty, then 
\begin{itemize}
\item[] 1. $\deg \mathfrak{p}_{1s} \geqslant \min \{ n -\lvert b \rvert \}_{b\in \{ b_1, \ldots, b_s \}
\cap \mathfrak{D}(Q_1,R_s)}$;
\item[] 2. for every $Z\in \isets{\mathfrak{p}_{1s}}$ with $\lvert Z \rvert < n - 2$, there exists an
unique $b\in \{ b_1,\ldots,$ $b_s \} \cap \mathfrak{D}(Q_1,R_s)$ such that 
$Z=\overline{n}- b$.
\end{itemize}
\end{itemize}
As $R_q=S \cdot \xschedule{(\bigcup_{i}^q Y_{(i)}) \cup (X - (\bigcup_i^q X_{(i)})) }=S\cdot
\xschedule{Y}=Q_2$, the sequence $\mathfrak{p}_{1q}$ is the desired \regularpath path from $Q_1$ to
$Q_2$, fulfilling conditions $(A)$ and $(B)$ (because $\{ b_1, \ldots, b_q \} \cap
\mathfrak{D}(Q_1,R_q) = \mathbf{Inv}_{S} \cap
\mathfrak{D}(Q_1,Q_2) = \mathfrak{D}(Q_1,Q_2)$). The result
follows.
\end{proof}

\subsection{The main structural results}
\label{secStructuralResultsngeq3}

In this section, we prove Lemma \ref{lemkboxesleqnmkconnFULLWOR}. To 
do this, we need one more structural result, Lemma 
\ref{lemSequenceRoundRSequenceRoundRp1geqR} and our main result 
about Johnson graphs: Theorem \ref{NoConnUmNoConnzetaUm}. We first 
introduce some combinatorial definitions before stating this theorem. 
Remember from Section \ref{secSharedObjsCombinatorialSets} that 
for $1\leqslant m \leqslant n$, 
$V_{n,m}=\{ c\subseteq \overline{n} \mid \left| c \right| = m \}$
and let $U\subseteq V_{n,m}$. We 
define the set $\zeta(U)$ as 
\begin{equation*}
\zeta(U)= \{ c\cup d \mid c,d \in U \text{ and }  \lvert c \cap d \rvert = m - 1 \}. 
\end{equation*}
Notice that each $f\in \zeta(U)$ has size $m+1$, 
thus $\zeta(U)\subseteq V_{n,m+1}$. For any
$v=0,\ldots,n-m$, the \emph{iterated $\zeta$-operator} $\zeta^v$ is given by
\begin{equation*}
 \zeta^v(U)=\begin{cases}
         U & \text{if } v=0, \\
	 \zeta(\zeta^{v - 1}(U)) & \text{otherwise}.
        \end{cases}
\end{equation*}
Since $U\subseteq V_{n,m}$, we can check that $\zeta^v(U)\subseteq V_{n,m+v}$. 

The intuitive interpretation of $U,\zeta$ and the iterated 
$\zeta$-operator from the point of view of WOR consensus protocols 
is the 
following: $U$ represents a set of possible intermediate 
agreements between some subsets of processes of size $m$ and $\zeta$ is 
some kind of method the processes use to extend the partial 
agreements specified by $U$ to new agreements between sets of processes 
of size $m+1$, but these new agreements can be made only if the 
required agreements of $m$ processes exist inside $U$. The iterated 
$\zeta$-operator is just a way to extend $\zeta$ to produce larger 
agreements, starting from agreements between sets of processes of size 
$m$. We are ready to state the combinatorial result needed to prove 
Lemma \ref{lemkboxesleqnmkconnFULLWOR}. 

\begin{theorem}\label{NoConnUmNoConnzetaUm}
 Let $U\subset V_{n,m}$ such that $\lvert U \rvert \leqslant n - m$. 
 Then $\zeta^{n - m}(U) = \varnothing$.
\end{theorem}

Notice that 
$\zeta^{n - m}(U)\subseteq V_{n,m+(n-m)}=V_{n,n}=\{ \overline{n} \}$. 
Thus, the intuitive meaning of Theorem \ref{NoConnUmNoConnzetaUm} (a\-gain, from the point of view of WOR 
consensus protocols) is that if there are not enough agreements 
$(\leqslant n - m)$ of subsets of processes of size $m$, it is 
impossible to reach consensus among $n$ processes, represented by the set $\overline{n}$. 
See Appendix \ref{appendixProofJohnsonGraphs} for more details and the full proof of Theorem \ref{NoConnUmNoConnzetaUm}.

With respect to the proofs of Lemmas \ref{lemSequenceRoundRSequenceRoundRp1geqR} and 
\ref{lemkboxesleqnmkconnFULLWOR}, we need one more definition regarding 
paths of connected states of a WOR protocol and the boxes representing safe-consensus 
shared objects. For a given path $\mathfrak{s}$ 
of connected states of an iterated protocol $\mathcal{A}$, consider the set
\[ \beta_\mathcal{A}(\mathfrak{s}) = \{ b \in \Gamma_\mathcal{A}(n) \mid (\exists X,Y\in
\isets{\mathfrak{s}}) (X \neq Y \text{ and } \lvert X \cap b \rvert = 1 \text{ and } \lvert Y \cap b \rvert = 1) \}, \]
and if $m\in \overline{n}$, let $\beta_\mathcal{A}(\mathfrak{s}; m)=\beta_\mathcal{A}(\mathfrak{s})
\cap \Gamma_\mathcal{A}(n,m)$. Roughly speaking, the set 
$\beta_\mathcal{A}(\mathfrak{s})$ captures the safe-consensus shared 
objects the processes can invoke to decrease the degree of 
indistinguishability of paths composed of states which are succesor 
states of the elements of $\states{\mathfrak{s}}$. In other words, if 
the path $\mathfrak{s}$ is composed of reachable states in $\mathcal{A}$ in round $r$, 
then the boxes (safe-consensus objects) in 
$\beta_\mathcal{A}(\mathfrak{s})$ can help the processes diminish the 
degree of indistinguishability of any posible path $\mathfrak{q}$ such 
that each state in $\states{\mathfrak{q}}$ is a successor state of some 
member of $\states{\mathfrak{s}}$. 
 If there is no confusion about which protocol we refer to, we write 
 $\beta(\mathfrak{s})$ and $\beta(\mathfrak{s}; m)$ instead of $\beta_\mathcal{A}(\mathfrak{s})$  
 and $\beta_\mathcal{A}(\mathfrak{s}; m)$ respectively.

\begin{lemma}\label{lemSequenceRoundRSequenceRoundRp1geqR}
 Let $n\geqslant 3$ and $1 \leqslant v \leqslant s \leqslant n - 2$ be fixed. Suppose that
$\mathcal{A}$ is a WOR protocol with safe-consensus objects and that there exists a
sequence
\[ \mathfrak{s} \colon S_0\linkstates{X_1} \cdots \linkstates{X_q} S_q\quad (q\geqslant 1), \]
of connected states of round $r\geqslant 0$ such that 
\begin{itemize}
 \item[] $\Lambda_1$ $\deg \mathfrak{s} \geqslant v$. 
 \item[] $\Lambda_2$ For every $X_i$ with $v \leqslant \lvert X_i \rvert < s$, there exists $b_i\in
\Gamma_\mathcal{A}(n,n - \lvert X_i \rvert)$ such that $X_i=\overline{n} - b_i$
 \item[] $\Lambda_3$ $\overline{n} \notin \beta(\mathfrak{s})$.
\end{itemize}
Then in round $r+1$ there exists a sequence
\[\mathfrak{q} \colon Q_0\linkstates{} \cdots \linkstates{} Q_u, \]
of connected successor states of all the $S_j$, such that the following statements hold:
\begin{itemize}
\item[] $\Psi_1$ If $\beta(\mathfrak{s}; n - v + 1)=\varnothing$, then $\deg
\mathfrak{q} \geqslant v$.
\item[] $\Psi_2$ If $\beta(\mathfrak{s}; n - v + 1) \neq \varnothing$, then $\deg
\mathfrak{q} \geqslant v - 1$.
\item[] $\Psi_3$ For every $Z\in \isets{\mathfrak{q}}$ with $\lvert Z \rvert = v - 1$, there
exist $X,X^\prime\in \isets{\mathfrak{s}}$ and 
$b\in \beta(\mathfrak{s}; n - v
+ 1)$ such that $Z=\overline{n} - b=X \cap X^\prime$  and $b=c_{1}\cup c_{2}$, $c_{k}\in
\Gamma_\mathcal{A}(n,n - v)$.
\item[] $\Psi_4$ For every $Z\in \isets{\mathfrak{q}}$ with $v \leqslant \lvert Z \rvert < s$, there
exists 
$b\in \Gamma_\mathcal{A}(n,n - \lvert Z \rvert)$ such that $Z=\overline{n} - b$.
\item[] $\Psi_5$ $\beta(\mathfrak{q}; n - l + 1)\subseteq
\zeta(\beta(\mathfrak{s}; n - l))$ for $v \leqslant l < s$.
\end{itemize}
\end{lemma}

\begin{proof}
In order to find the path $\mathfrak{q}$ that has the properties $\Psi_1$-$\Psi_5$, first we need to
define a set of states of round $r+1$, called
$\Phi(\mathfrak{s})$, which we will use to construct the path $\mathfrak{q}$. The
key ingredient of $\Phi(\mathfrak{s})$ is the safe-consensus values of the
boxes used in each member of $\Phi(\mathfrak{s})$. Define the set of states
$\Phi(\mathfrak{s})$ as
\[
 \Phi(\mathfrak{s}) = \{ R \mid R=S \cdot \xschedule{X}, (S, X) \in
\states{\mathfrak{s}} \times \isets{\mathfrak{s}} \}.
\]
 Each element of $\Phi(\mathfrak{s})$ is a state in round $r+1$ 
which is obtained when all the processes with ids in some set $X\in \isets{\mathfrak{s}}$ execute
concurrently the operations of $\mathcal{A}$, followed by all processes with ids in $\overline{n} -
X$. The safe-consensus values of each box for the states of
$\Phi(\mathfrak{s})$
in round $r+1$ are defined by using the following rules: Let $b$ be any box such that $b\notin
\beta(\mathfrak{s})$.
\begin{itemize}
\item If $\lvert X \cap b\rvert \neq 1$ for every $X\in \isets{\mathfrak{s}}$, then 
we 
choose any value $j$ such that $j$ is the safe-consensus value of
$b$ in every state $R\in \Phi(\mathfrak{s})$ 
with 
$b\in \setofboxes{R}$.
\item If there exists exactly one set $X\in \isets{\mathfrak{s}}$ such that $X \cap b = \{ x \}$,
then $b$ has $j_x$ as its safe-consensus value in every state
$Q\in \Phi(\mathfrak{s})$ with $b\in \setofboxes{Q}$,
where $j_x$ is the value proposed by process $p_x$ when invoking the safe-consensus object represented by $b$.
\end{itemize}
Now we establish rules to define the safe-consensus values of every element of the set
$\beta(\mathfrak{s})$, using the order on the set $\isets{\mathfrak{s}}$ induced by the
path $\mathfrak{s}$, when we traverse $\mathfrak{s}$ from $S_0$ to $S_q$. 
That is, 
$\isets{\mathfrak{s}}$ is ordered as
\begin{equation}\label{eqOrderSetPhiArp1s}
 X_1,\ldots,X_q.
\end{equation}
 For each $b\in \beta(\mathfrak{s})$, let that $X_i,X_{i+ z_1}, \ldots,X_{i+
z_k}$, ($1 \leqslant i < i+z_1 < i+z_2 <\cdots < i+z_{k} \leqslant q; z_j > 0$ and $k \geqslant
1$) be the (ordered) subset of $\isets{\mathfrak{s}}$ of all $X\in \isets{\mathfrak{s}}$ with the
property $\lvert X \cap b \rvert = 1$. Take the set $X_i$. If $x_i\in X_i \cap b$, then 
we make the value proposed by process $p_{x_i}$ (when invoking the 
shared object represented by $b$) 
the safe-consensus value of $b$ for all the elements of
$\Phi(\mathfrak{s})$ of
the form 
\[P_j=S\cdot \xschedule{X_j}, \qquad 1 \leqslant j \leqslant i + z_1 - 1,\] 
where $b\in \setofboxes{P_j}$.
 Next, in all states $R_u\in \Phi(\mathfrak{s})$ such that $R_u=S\cdot
\xschedule{X_u}$ with $b\in \setofboxes{R_u}$ and $i + z_1 \leqslant u \leqslant i + z_2 - 1$, $b$
has safe-consensus value equal to the input value of the process with 
id 
$x_{i+z_1}\in X_{i+z_1}\cap b$ when invoking the safe-consensus object 
represented by $b$. 
In general, in all states
$T_v\in
\Phi(\mathfrak{s})$ of the form 
\[T_v=S\cdot \xschedule{X_v}, \qquad i + z_l \leqslant v \leqslant i + z_{l+1} - 1\text{ and }1\leqslant l \leqslant k - 1,\] 
where $b\in \setofboxes{T_v}$, 
$b$ has safe-consensus value equal to 
the input value of the process with id $x_{i+z_{l}}\in X_{i+z_l}\cap b$ 
feed to the safe-consensus object represented by $b$.
Finally, in each
state
$L_w=S\cdot \xschedule{X_w} \in \Phi(\mathfrak{s})$ with $b\in \setofboxes{L_w}$
and $i+z_{k} \leqslant w \leqslant q$, $b$ has safe-consensus value equal to the input value that the process with id 
$x_{i+z_{k}}\in X_{i+z_k}\cap b$ inputs to the safe-consensus 
shared object 
represented by $b$.

Using the  Safe-Validity property of the safe-consensus task, it is an easy 
routine task to show that we can find states of $\mathcal{A}$ in round $r+1$ which have the form of
the states given in $\Phi(\mathfrak{s})$ and with the desired safe-consensus
values for every box.

We are ready to build the sequence $\mathfrak{q}$ of the Lemma\footnote{Notice that the order
given to $\Phi(\mathfrak{s})$ in equation \ref{eqOrderSetPhiArp1s} is the precise
order in which the states of this set will appear in the path $\mathfrak{q}$.}. We use induction on
$q\geqslant 1$. For the base case when $q=1$, we have that $\mathfrak{s} \colon S_0\linkstates{X_1}
S_1$, with $\lvert X_1 \rvert \geqslant v$. Here we easily build the sequence $S_0 \cdot
\xschedule{X_1} \linkstates{X_1} S_1 \cdot \xschedule{X_1}$ which clearly satisfies properties
$\Psi_1$-$\Psi_5$. Assume that for $1\leqslant w < q$ we have a path 
\[ \mathfrak{q}^\prime \colon Q_0\linkstates{} \cdots \linkstates{} Q_l \quad (l\geqslant 1), \]
Satisfying conditions $\Psi_1$-$\Psi_5$ and such that $Q_0=S_0\cdot \xschedule{X_1},Q_l=S_w \cdot
\xschedule{X_w}$. We now connect the state $Q_l$ with $Q=S_{w+1} \cdot \xschedule{X_{w+1}}$. By
Theorem \ref{lemlocalconnectXYGeneric}, there is a \regularpath sequence 
\[ \mathfrak{v}\colon Q_l \linkstates{} \cdots \linkstates{} Q^\prime \] 
such that $Q^\prime=S_w \cdot \xschedule{X_{w+1}}$ and $\mathfrak{v},Q_l$ and $Q^\prime$ satisfy
conditions (A)-(B) of that theorem.
Clearly  $\mathfrak{D}(Q_l,$ $Q^\prime)\subseteq \beta(\mathfrak{s})$,
because the only boxes used by $\mathcal{A}$ in the states $Q_l$ and $Q^\prime$ with different
safe-consensus value, are the boxes which intersect two different sets $X,X^\prime \in
\isets{\mathfrak{s}}$ in one element. Joining the sequence $\mathfrak{q}^\prime$ with
$\mathfrak{v}$, followed by the small path $\mathfrak{t} \colon Q^\prime \linkstates{X_{w+1}} Q$ 
(such that $\deg \mathfrak{t} \geqslant v$), we obtain a new sequence $\mathfrak{q} \colon Q_0 \linkstates{} \cdots
\linkstates{}
Q$. We now need to show that $\mathfrak{q}$ satisfies properties $\Psi_1$-$\Psi_5$.
\begin{itemize}
 \item[] $\Psi_1$ Notice that every box in $\beta(\mathfrak{s})$ has size no bigger that $n - v +
1$. Suppose that $\beta(\mathfrak{s}; n - v + 1)=\varnothing$. In
particular, this implies that $\mathfrak{D}(Q_l,Q^\prime) \cap
\Gamma_\mathcal{A}(n, n - v + 1)=\varnothing$. If it happens that
$\mathfrak{D}(Q_l,Q^\prime)$ is void, then as condition (A) of Theorem 
\ref{lemlocalconnectXYGeneric} holds for $\mathfrak{v}$, $\deg \mathfrak{v} \geqslant v$. But if
$\mathfrak{D}(Q_l,Q^\prime)\neq \varnothing$ we have that $\lvert b \rvert
\leqslant n - v$ for all $b\in \mathfrak{D}(Q_l,Q^\prime)$, thus by part 1. of
property (B) of Theorem \ref{lemlocalconnectXYGeneric}, 
\[ \deg \mathfrak{v} \geqslant \min \{ n -\lvert b \rvert \}_{b\in
\mathfrak{D}(Q_l,Q^\prime)} \geqslant v. \] 
By the induction, hypothesis $\deg \mathfrak{q}^\prime \geqslant v$. It is also true that 
$\deg \mathfrak{t} \geqslant v$, therefore $\deg \mathfrak{q} \geqslant v$.
\item[] $\Psi_2$ If $\beta(\mathfrak{s}; n - v + 1)\neq \varnothing$, then 
either $\deg \mathfrak{q}^\prime \geqslant v - 1$ (by the induction hypothesis) or $\deg
\mathfrak{v} \geqslant v - 1$. This last assertion is true, because by (B) of Theorem 
\ref{lemlocalconnectXYGeneric}, we have that $\deg \mathfrak{v} \geqslant \min \{ n -\lvert b 
\rvert \}_{b\in \mathfrak{D}(Q_l,Q^\prime)} \geqslant v - 1$. Thus, 
$\deg \mathfrak{q} \geqslant v - 1$.
\item[] $\Psi_3$ We remark first that $\isets{\mathfrak{q}}= \isets{\mathfrak{q}^\prime} \cup
\isets{\mathfrak{v}} \cup \isets{\mathfrak{t}}$. If we are given $Z\in \isets{\mathfrak{q}}$ such
that $\lvert Z \rvert = v -
1$, then $Z\in \isets{\mathfrak{q}^\prime}$ or $Z\in \isets{\mathfrak{v}}$. When $Z\in
\isets{\mathfrak{q}^\prime}$, we use our induction hypothesis to show that $Z=Y \cap
Y^\prime=\overline{n} - d$ for some $Y,Y^\prime\in \isets{\mathfrak{s}}$ and $d\in
\beta(\mathfrak{s}; n - v + 1)$ such that $d=d^\prime \cup d^{\prime\prime}$,
$d^{\prime},d^{\prime\prime} \in \Gamma_\mathcal{A}(n, n - v)$. On the other hand, if $Z\in
\isets{\mathfrak{v}}$, it must be true that $\mathfrak{D}(Q_l,Q^\prime)\neq
\varnothing$ (if not, then $\deg \mathfrak{v} \geqslant n - 2$ by property (A) of Theorem 
\ref{lemlocalconnectXYGeneric} for $\mathfrak{v}$ and this contradicts the size of $Z$). As $v - 1 <
n - 2$, we can apply part 2 of condition (B) of Theorem \ref{lemlocalconnectXYGeneric} to deduce
that $Z=\overline{n} - b$ for a unique $b\in \mathfrak{D}(Q_l,Q^\prime)$. Because
$\lvert Z \rvert = v - 1$, it is true that $\lvert b \rvert = n - v + 1$. Now, $b\in
\mathfrak{D}(Q_l,Q^\prime)\subseteq \beta(\mathfrak{s})$, thus 
there exist $X,X^\prime \in \isets{\mathfrak{s}}$ such that 
\begin{equation}\label{lemSequenceRoundRSequenceRoundRp1geqR_eqPsi3}
  \lvert X \cap b \rvert = 1 \text{ and } \lvert X^\prime \cap b \rvert = 1
\end{equation}
Combining Equation \eqref{lemSequenceRoundRSequenceRoundRp1geqR_eqPsi3} 
with the fact that the size of $b$ is $n - v + 1$, we can check that $\lvert X \rvert =  \lvert X^\prime
\rvert = v$ and $\overline{n} - b = X \cap X^\prime$, that is,
\[ Z=\overline{n} - b=X \cap X^\prime, \]
and finally, we apply $\Lambda_2$ of the sequence $\mathfrak{s}$ to $X,X^\prime$ to find two unique
boxes $c,c^\prime \in \Gamma_\mathcal{A}(n,n - v)$ with $X=\overline{n} - c \text{ and }
X^\prime=\overline{n} - c^\prime$, so that 
\[ \overline{n} - b=X \cap X^\prime = (\overline{n} - c) \cap (\overline{n} - c^\prime) =
\overline{n} - (c \cup c^\prime), \]
then $b=c \cup c^\prime$ and condition $\Psi_3$ is satisfied by $\mathfrak{q}$.
\item[] $\Psi_4$ The sequence $\mathfrak{q}$ fulfils this property because of the induction
hypothesis on $\mathfrak{q}^\prime$, condition $\Lambda_2$ for $\mathfrak{s}$ and (B) of Theorem 
\ref{lemlocalconnectXYGeneric} for $\mathfrak{v}$.
\item[] $\Psi_5$ Let $l\in \{ v,\ldots,s - 1\}$. This property is satisfied by $\mathfrak{q}$ if
$\beta(\mathfrak{q}; n - l + 1)=\varnothing$. Suppose that $c\in
\beta(\mathfrak{q}; n - l + 1)$, there exist $X,Y\in \isets{\mathfrak{q}}$ such that
$\lvert X \cap c \rvert = \lvert Y \cap c \rvert = 1$. Then $\lvert X \rvert = \lvert Y \rvert = l$,
so that we use property $\Psi_3$ (or $\Psi_4$) on $X,Y$ to find two boxes $b_X,b_Y\in
\beta(\mathfrak{s}; n - l)$ such that $X = \overline{n} - b_X \text{ and } Y = 
\overline{n} -
b_Y$. Also, $X \cap Y = \overline{n} - c$, thus $\overline{n} - c = X \cap Y = (\overline{n} - b_X)
\cap (\overline{n} - b_Y) = \overline{n} - (b_X \cup b_Y)$. Therefore $c = b_X \cup b_Y$ and this
says that $c\in \zeta(\beta(\mathfrak{s}; n - l))$. Condition $\Psi_5$ is fulfilled by
$\mathfrak{q}$.
\end{itemize}
We have shown by induction that we can build the path $\mathfrak{q}$ from the sequence
$\mathfrak{s}$, no matter how many states $\mathfrak{s}$ has. This proves the Lemma.
\end{proof}

The following corollary is a weaker version of Lemma \ref{lemSequenceRoundRSequenceRoundRp1geqR}, 
and it can be proven as a consequence of that result.

\begin{corollary}\label{lemSequenceRoundRSequenceRoundRp1geqRSimple}
 Let $n\geqslant 3$ and $1 \leqslant s \leqslant n -2$ be fixed. Suppose that $\mathcal{A}$ is a WOR 
protocol with safe-consensus objects and there exists a sequence 
\[ \mathfrak{s} \colon S_0\linkstates{X_1} \cdots \linkstates{X_r} S_q\quad (q\geqslant 1), \]
of connected states of round $r\geqslant 0$ such that $\deg \mathfrak{s} \geqslant s$ and
$\overline{n}\notin \beta(\mathfrak{s})$. Then in round $r+1$ there exists a sequence
\[\mathfrak{q} \colon Q_0\linkstates{Z_1} \cdots \linkstates{Z_u} Q_u, \]
such that the following statements hold:
\begin{itemize}
\item[] $\Psi_1$ If $\beta(\mathfrak{s}; n - s + 1)=\varnothing$, then $\deg
\mathfrak{q} \geqslant s$.
\item[] $\Psi_2$ If $\beta(\mathfrak{s}; n - s + 1)$ is not empty, then $\deg
\mathfrak{q} \geqslant s - 1$. 
\item[] $\Psi_3$ For every $Z\in \isets{\mathfrak{q}}$ with $\lvert Z \rvert = s - 1$, there
exist $X,X^\prime\in \isets{\mathfrak{s}}$ and a unique $b\in \beta(\mathfrak{s}; n -
s
+ 1)$ such that $Z=\overline{n} - b=X \cap X^\prime$.
\item[] $\Psi_5$ $\beta(\mathfrak{q}; n - s + 2)\subseteq
\zeta(\beta(\mathfrak{s}; n - s + 1))$.
\end{itemize}
\end{corollary}

We have gathered all the required elements to prove one of the key ingredient of the full 
proof of Theorem \ref{lemminimumKboxesFULLWOR}. Lemma 
\ref{lemkboxesleqnmkconnFULLWOR} expresses formally the fact that 
if in a WOR consensus protocol, the processes cannot make enough 
agreements of size $m$ ($\nu_\mathcal{A}(n, m)\leqslant n-m$) then 
they won't be able to decrease the degree of indistinguishability of 
paths of connected states of reachable states in every round 
$r\geqslant 1$ of the protocol.

\begin{lemma}\label{lemkboxesleqnmkconnFULLWOR}
  Assume that $3 \leqslant m \leqslant n$. Suppose that $\mathcal{A}$ is a WOR 
protocol with safe-consensus objects such that $\nu_\mathcal{A}(n, m)\leqslant n-m$.
If $S_r,Q_r$ are two reachable states in $\mathcal{A}$ for some round $r\geqslant 0$, connected 
with
a sequence $\mathfrak{q} \colon S_r\linkstates{} \cdots \linkstates{} Q_r$ of connected states 
such that $\deg \mathfrak{q} \geqslant n - m+1$. 
Then for all $u\geqslant 0$, there exist successor states
$S_{r+u},Q_{r+u}$ of $S_r$ and $Q_r$ respectively, in round $r+u$ of $\mathcal{A}$, such that
$S_{r+u}$ and $Q_{r+u}$ are connected.
\end{lemma}
\begin{proof}
Let $\mathcal{A}$ be a protocol with
the given hypothesis, set $\mathfrak{q}_0=\mathfrak{q}$ and $m=n - s + 1$. Because $3\leqslant
m \leqslant n$, $s \in \{1,\ldots, n - 2 \}$. Using the sequence $\mathfrak{q}_0$
and applying Corollary \ref{lemSequenceRoundRSequenceRoundRp1geqRSimple}, we build a path
$\mathfrak{q}_1\colon S_{r+1}\linkstates{} \cdots \linkstates{} Q_{r+1}$, connecting the successor
states $S_{r+1},Q_{r+1}$ of $S_r$ and $Q_r$ respectively, such that 
\begin{itemize}
\item[] $\Psi_{1,1}$ If $\beta(\mathfrak{q}_0; n - s + 1)=\varnothing$, then $\deg
\mathfrak{q}_1 > s - 1$.
\item[] $\Psi_{1,2}$ If $\beta(\mathfrak{q}_0; n - s + 1)$ is not empty, then $\deg
\mathfrak{q}_1 \geqslant s - 1$.
\item[] $\Psi_{1,3}$ For every $Z\in \isets{\mathfrak{q}_1}$ with $\lvert Z \rvert = s - 1$, there
exist $X,X^\prime\in \isets{\mathfrak{q}_0}$ and a unique $b\in \beta(\mathfrak{q}_0; n
- s + 1)$ such that $Z=\overline{n} - b=X \cap X^\prime$.
\item[] $\Psi_{1,5}$ $\beta(\mathfrak{q}_1; n - s + 2)\subseteq
\zeta(\beta(\mathfrak{q}_0; n - s + 1))$.
\end{itemize}
Starting from $\mathfrak{q}_1$ and using induction together with Lemma
\ref{lemSequenceRoundRSequenceRoundRp1geqR}, we can prove that for each $u\in \{1,\ldots, s - 1 \}$,
there exist successor states $S_{r+u},Q_{r+u}$ of the states $S_{r}$ and $Q_{r}$ respectively, and a
sequence
\[ \mathfrak{q}_u \colon S_{r+u} \linkstates{} \cdots \linkstates{} Q_{r+u},\]
satisfying the properties:
\begin{itemize}
\item [] $\Psi_{u,1}$ If $\beta(\mathfrak{q}_{u - 1}; n - s + u)=\varnothing$, then $\deg
\mathfrak{q}_u > s - u$.
\item [] $\Psi_{u,2}$ If $\beta(\mathfrak{q}_{u - 1}; n - s + u)$ is not empty, then
$\deg \mathfrak{q}_u \geqslant s - u$.
\item [] $\Psi_{u,3}$ For each $Z\in \isets{\mathfrak{q}_u}$ with $\lvert Z \rvert = s - 1$, there
exist $X,X^\prime\in \isets{\mathfrak{q}_{u - 1}}$ and a unique $b\in
\beta(\mathfrak{q}_{u - 1}; $ $n - s + 1)$ such that $Z=\overline{n} - b=X \cap X^\prime$.
\item [] $\Psi_{u,4}$ For each $Z\in \isets{\mathfrak{q}_u}$ with $s - u \leqslant \lvert Z \rvert <
s - 1$, there exist $X,X^\prime \in \isets{\mathfrak{q}_{u - 1}}$ and a unique $b\in
\beta(\mathfrak{q}_{u - 1};n - \lvert Z \rvert)$ such that $Z=\overline{n} - b=X \cap
X^\prime$ and $b=c_{1}\cup c_{2}$, $c_{k}\in \beta(\mathfrak{q}_{u - 1};n - \lvert Z
\rvert - 1)$. 
\item [] $\Psi_{u,5}$ $\beta(\mathfrak{q}_u; n - l + 1)\subseteq
\zeta(\beta(\mathfrak{q}_{u - 1}; n - l))$ for $s - u \leqslant l < s$.
\end{itemize}
When $u=s - 1$, we obtain a sequence $\mathfrak{q}_{s - 1} \colon S_{r+ s - 1} \linkstates{} \cdots
\linkstates{} Q_{r+s - 1}$ that connects the states $S_{r+ s - 1}$ and $Q_{r+ s - 1}$, such that
\begin{itemize}
\item [] $\Psi_{s-1,1}$ If $\beta(\mathfrak{q}_{s - 2}; n - 1)=\varnothing$, then $\deg
\mathfrak{q}_{s - 1} > 1$.
\item [] $\Psi_{s-1,2}$ If $\beta(\mathfrak{q}_{s - 2}; n - 1)$ is not empty, then $\deg
\mathfrak{q}_{s - 1} \geqslant 1$.
\item [] $\Psi_{s-1,3}$ For each $Z\in \isets{\mathfrak{q}_{s - 1}}$ with $\lvert Z \rvert = s -
1$, there exist $X,X^\prime\in \isets{\mathfrak{q}_{s - 2}}$ and a unique $b\in
\beta(\mathfrak{q}_{s - 2}; n - s + 1)$ such that $Z=\overline{n} - b=X \cap X^\prime$.
\item [] $\Psi_{s-1,4}$ For any $Z\in \isets{\mathfrak{q}_{s - 1}}$ with $1 \leqslant \lvert Z
\rvert < s - 1$, there exist $X,X^\prime\in \isets{\mathfrak{q}_{s - 2}}$ and a unique $b\in
\beta(\mathfrak{q}_{s - 2};n - \lvert Z \rvert)$ such that $Z=\overline{n} - b=X \cap
X^\prime$ and $b=c_{1}\cup c_{2}$, $c_{k}\in \beta(\mathfrak{q}_{s - 2};n - \lvert Z
\rvert - 1)$. 
\item [] $\Psi_{s-1,5}$ $\beta(\mathfrak{q}_{s - 1}; n - l + 1)\subseteq
\zeta(\beta(\mathfrak{q}_{s - 2}; n - l))$ for $1\leqslant l < s$.
\end{itemize}
Our final goal is to show that for all $v \geqslant s - 1$, we can connect in each round
$r+v$ of $\mathcal{A}$, successor states of $S_r$ and $Q_r$. We first claim that for any
$w=0,\ldots,s - 1$ and $z\geqslant 0$,
\[ \zeta^z(\beta(\mathfrak{q}_{w}; n - s + 1)) \subseteq \zeta^z(\Gamma_\mathcal{A}(n, n
- s + 1)), \]
(this is true because $\zeta$ preserves $\subseteq$) and combining this fact with the properties
$\Psi_{u,5}$ and induction, we can show that for $1 \leqslant l < s$
\begin{equation}\label{eqbetasetsstationaries}
 \beta(\mathfrak{q}_{s - 1}; n - l + 1) \subseteq \zeta^{s -
l}(\beta(\mathfrak{q}_{l - 1}; n - s + 1)) \subseteq \zeta^{s -
l}(\Gamma_\mathcal{A}(n, n - s + 1)).
\end{equation}
And because $\lvert \Gamma_\mathcal{A}(n, n - s + 1) \rvert = \nu_\mathcal{A}(n, n - s + 1)
\leqslant s - 1$ and $m=n - s + 1$, we use Theorem \ref{NoConnUmNoConnzetaUm} to check that
$\zeta^{s - 1}(\Gamma_\mathcal{A}(n, n - s + 1))=\varnothing$, implying that
\[\beta(\mathfrak{q}_{s - 1}; n) = \varnothing.\]
With all this data, Lemma \ref{lemSequenceRoundRSequenceRoundRp1geqR} and an easy inductive argument 
(starting at the base case $v=s -  1$), we can find for all $v\geqslant s - 1$,
states $S_{r+v},Q_{r+v}$ of round $r+v$ of $\mathcal{A}$, which are successor states of $S_r$ and
$Q_r$ respectively and connected with a sequence $\mathfrak{q}_v$ such that
\begin{itemize}
\item[] $\Psi^\prime_{v,1}$ $\deg \mathfrak{q}_{v} \geqslant 1$.
\item[] $\Psi^\prime_{v,2}$ For every $Z\in \isets{\mathfrak{q}_{v}}$ with $\lvert Z \rvert = s -
1$, there
exist $X,X^\prime\in \isets{\mathfrak{q}_{v - 1}}$ and a unique $b\in
\beta(\mathfrak{q}_{v - 1}; n - s + 1)$ such that $Z=\overline{n} - b=X \cap X^\prime$.
\item[] $\Psi^\prime_{v,3}$ For every $Z\in \isets{\mathfrak{q}_{v}}$ with $1 \leqslant \lvert Z
\rvert < s - 1$, there exist $X,X^\prime\in \isets{\mathfrak{q}_{v - 1}}$ and a unique $b\in
\beta(\mathfrak{q}_{v - 1};n - \lvert Z \rvert)$ such that $Z=\overline{n} - b=X \cap
X^\prime$ and $b=c_{1}\cup c_{2}$, $c_{k}\in \beta(\mathfrak{q}_{v - 1};n - \lvert Z
\rvert - 1)$. 
\item[] $\Psi^\prime_{v,4}$ $\beta(\mathfrak{q}_{v}; n - l + 1) \subseteq \zeta^{s
- l}(\Gamma_\mathcal{A}(n, n - s + 1))$ for $1 \leqslant l < s$.
\end{itemize}
It is precisely these four properties of the path $\mathfrak{q}_{v}$ which allow us to find the
path $\mathfrak{q}_{v+1}$, applying Lemma \ref{lemSequenceRoundRSequenceRoundRp1geqR}, such that
$\mathfrak{q}_{v+1}$ enjoys the same properties. Therefore, starting from round $r$, we can connect
in all rounds of $\mathcal{A}$, successor states of $S_r$ and $Q_r$. The Lemma is
proven.
\end{proof}

\subsection{The case $\nu_\mathcal{A}(n,2) \leqslant n - 2$}
\label{secStructuralResultsneq2}

In Section \ref{secStructuralResultsngeq3}, we proved Lemma 
\ref{lemkboxesleqnmkconnFULLWOR}, a key result to prove Theorem 
\ref{lemminimumKboxesFULLWOR}, but that lemma does not cover the case 
$\nu_\mathcal{A}(n,2) \leqslant n - 2$. In this section we prove 
 Lemma \ref{thm2boxesleqnm2}, which covers the case 
 $\nu_\mathcal{A}(n,2) \leqslant n - 2$ and it is 
 the last ingredient that we need to prove Theorem \ref{lemminimumKboxesFULLWOR}. This can be done easily using 
 Lemma \ref{lemBasicPathn2}, a structural result concerning paths of 
 connected states and the following combinatorial result.

\begin{lemma}\label{lemPartitionn2}
 Let $U\subset V_{n,2}$ such that $\lvert U \rvert \leqslant n - 2$. Then there exists a partition
$\overline{n}=A \cup B$ such that 
\[ (\forall b\in U) (b\subseteq A \text{ or } b \subseteq B). \]
\end{lemma}

The proof of Lemma \ref{lemPartitionn2} can be found in Section \ref{secAppendixFinalProofs} of the Appendix. 
Lemma \ref{lemPartitionn2} describes the combinatorics of Lemma 
\ref{thm2boxesleqnm2} for the case $m=2$, just in the same way Theorem 
\ref{NoConnUmNoConnzetaUm} does the same thing for Lemma 
\ref{lemkboxesleqnmkconnFULLWOR}. Intuitively, when 
$\lvert U \rvert \leqslant n - 2$, there are not enough agreements 
between pairs of processes and this allow us to partition the set of 
processes $\Pi$ into two sets 
\[ \Pi = A \cup B, \]
and this partition can be used to build sequences of connected states 
in every round of an executing WOR protocol (Lemmas 
\ref{lemBasicPathn2} and \ref{thm2boxesleqnm2}). We now proceed to prove
the required results.

\begin{lemma}\label{lemBasicPathn2}
 Let $\mathcal{A}$ be a WOR protocol with safe-consensus objects. Suppose there
exists a partition $\overline{n}=A \cup B$ and a sequence
\[\mathfrak{p} \colon S_0\linkstates{}\cdots\linkstates{} S_l\qquad (l\geqslant 0) \]
of connected states in round $r\geqslant 0$ of $\mathcal{A}$, with the following properties
\begin{itemize}
 \item[] I) $\isets{\mathfrak{p}}=\{ A,B \}$; 
\item[] II) $(\forall b\in \Gamma_\mathcal{A}(n, 2)) (b\subseteq A \text{ or } b \subseteq B)$.
\end{itemize}
Then in round $r+1$ of $\mathcal{A}$ there exists a path
\[\mathfrak{q} \colon Q_0\linkstates{}\cdots\linkstates{} Q_s\qquad (s\geqslant 1)\]
of connected states and the following properties hold:
\begin{itemize}
 \item [a)] Each state $Q_k$ is of the form $Q_k=S_j \cdot \xschedule{X}$, where $X=A \text{ or } 
X=B$;
\item [b)] $\isets{\mathfrak{q}}=\{ A,B \}$.
\end{itemize}
\end{lemma}

\begin{proof}
  The techniques needed to prove this result are similar to those used in the proof of Lemma 
  \ref{lemSequenceRoundRSequenceRoundRp1geqR}.
  We first define the safe-consensus value of every box $b\in \Gamma_\mathcal{A}(n, 2)$, using property
  II) and the safe-Validity property of safe-consensus. After doing that, we apply 
  induction on $l$ to build the path $\mathfrak{q}$ (of succesor states of elements from $\isets{\mathfrak{p}}$)
  satisfying a)-b). We omit the details.
\end{proof}

\begin{lemma}\label{thm2boxesleqnm2}
Let $n\geqslant 2$. If $\mathcal{A}$ is a WOR protocol for $n$ processes using
safe-consensus objects with $\nu_\mathcal{A}(n,2)$ $\leqslant n - 2$ and $S$ is a reachable 
state
in $\mathcal{A}$ for some round $r\geqslant 0$, then there exists a partition of the set 
$\overline{n}=A\cup B$ such that for all $u\geqslant 0$, the states $S \cdot \xschedulesup{u}{A}$
and $S\cdot \xschedulesup{u}{B}$ are connected.
\end{lemma}

\begin{proof}
This proof is analogous to the proof of Lemma \ref{lemkboxesleqnmkconnFULLWOR}. We use Lemma 
\ref{lemPartitionn2} to find the partition of $\overline{n}$ and then we apply 
inductively 
Lemma \ref{lemBasicPathn2}. We omit the details.
\end{proof}

\subsection{The proof of Theorem \ref{lemminimumKboxesFULLWOR}}

Here we give the proof of Theorem \ref{lemminimumKboxesFULLWOR} for any $n \geqslant 2$, thus 
completing all the necessary proofs of the paper.

\paragraph{Proof of Theorem \ref{lemminimumKboxesFULLWOR}} Assume that there exists a protocol 
$\mathcal{A}$ for consensus such that there is some $m$ with
$2\leqslant m \leqslant n$ with $\nu_\mathcal{A}(n,m) \leqslant n - m$. Let $O,U$ be the
initial states in which all processes have as input values 0s and 1s respectively. We now find
successor states of $O$ and $U$ in each round $r\geqslant 0$, which are connected. There are two cases:
\begin{itemize}
 \item [] {\it Case $m=2$.} By Lemma \ref{thm2boxesleqnm2}, there exists a partition of
$\overline{n}=A\cup B$ such that for any state $S$ and any $r\geqslant 0$, $S \cdot
\xschedulesup{r}{A}$ and $S\cdot \xschedulesup{r}{B}$ are connected. Let $OU$ be the initial state
in which all processes with ids in $A$ have as input value 0s and all processes with ids in $B$ 
have as input values 1s. Then for all $r\geqslant 0$ we have that
\[ O\cdot \xschedulesup{r}{A}\linkstates{A} OU\cdot \xschedulesup{r}{A}\quad \text{and} \quad
OU\cdot \xschedulesup{r}{B}\linkstates{B} U\cdot \xschedulesup{r}{B} \]
and by Lemma \ref{thm2boxesleqnm2}, the states $OU\cdot \xschedulesup{r}{A}$ and $OU\cdot 
\xschedulesup{r}{B}$ are connected. 
Thus, for any $r$-round partial execution of $\mathcal{A}$, we can connect the states $O^r=O \cdot
\xschedulesup{r}{A}$ and $U^r=U\cdot \xschedulesup{r}{B}$.
 \item [] {\it Case $3\leqslant m \leqslant n$.} By Lemma \ref{leministates}, we know that any two 
initial states for 
consensus are connected, 
so that we can connect $O$ and $U$ with a sequence 
$\mathfrak{q}$ of initial states of $\mathcal{A}$ and it is not hard to check that $\deg 
\mathfrak{q} \geqslant n - 1 \geqslant n - m + 1$. By Lemma \ref{lemkboxesleqnmkconnFULLWOR}, for 
each round $r\geqslant 0$ of $\mathcal{A}$, there exist successor states $O^r,U^r$ of $O$ and $U$ 
respectively, such that $O^r$ and $U^r$ are connected.
\end{itemize}
In this way, we have connected successor states of $O$ and $U$ in each round of the protocol
$\mathcal{A}$. Now, $O$ is a $0$-valent, initial state, which is connected to the initial state 
$U$, so that we can apply Lemma \ref{lemvvstates} to conclude that $U$ is $0$-valent. But this 
contradicts the fact that $U$ is a $1$-valent state, so we have reached a contradiction. 
Therefore $\nu_\mathcal{A}(n,m) > n - m$.
\hspace*{\fill}$\square$\par\vspace{3mm}

\section{Conclusion}\label{secConclusions}

In this paper, we have introduced three extensions to the basic iterated model of distributed 
computing \cite{borowsky}, using the safe-consensus task proposed by Yehuda, Gafni and Lieber in 
\cite{yehudagafni} and studied some of their properties and the 
solvability of the consensus task in each model. These extensions are 
very natural, in the sense that, for the set of new shared 
objects used by the processes (besides the snapshot objects),
the protocols follow the conventions of 
the standard iterated model for the snapshot objects and the 
new shared objects: Each set of objects is arranged as an array 
and each object is used only once by each processes and each 
process can access at most one shared object. This 
gives the protocols of the extended models a well behaved 
structure and allows for an easier analysis of their executions, 
because of their inductive nature, 
even with the use of the new shared objects. We believe that 
iterated models extended with shared objects will play an 
important role in the development of the theory of models
of distributed computing systems with shared objects, just in 
the same way the basic iterated model has played a fundamental 
role in the development of the theory of standard shared memory 
systems.

In the first iterated model that we investigated, the WRO iterated model, the processes first 
write to memory, then they snapshot it and after that, they invoke safe-consensus 
objects. We proved that in this model, the consensus task cannot be implemented.
The impossibility 
proof uses simpler connectivity arguments that those used in the lower bound proof of the WOR 
iterated model. For the second iterated model, the OWR iterated model, processes first 
invoke safe-consensus objects, then they write to memory and then they snapshot the contents of the 
shared memory. We proved that this model is equivalent to the WRO iterated model for task 
solvability using simulations, thus we obtained as a corollary that consensus cannot be implemented in the OWR 
iterated model.

Finally, in the third model, the WOR iterated model, processes write to memory, 
invoke safe-consensus objects and they snapshot the shared memory. We first constructed a 
WOR protocol which can solve $n$-consensus with $\binom{n}{2}$ safe-consensus objects. To make this 
protocol more readable and simplify its analysis, we introduced a new group consensus task: The 
$g$-2coalitions-consensus 
task, which captures the difficulties of solving $g$ processes consensus, even when 
almost all the processes know the proposed input values, except for two
processes that only know one input value. This new task may be of 
independent interest in future research. We also proved that our WOR consensus protocol is sharp, by giving a $\binom{n}{2}$ lower 
bound on the number of 
safe-consensus objects necessary to implement $n$-consensus in the WOR iterated model with 
safe-consensus. This lower bound is the main result of this paper. To obtain this lower bound, 
we combined structural results about WOR protocols and combinatorial results applied to all the different 
subsets of processes that 
can invoke safe-consensus objects.
At the very end, it is bivalency \cite{fischerImpossCon}, but in order to be able to 
use bivalency, an intricate connectivity argument must be developed, in which we relate the way that 
the processes use the safe-consensus shared objects to solve consensus with subgraphs of the Johnson graphs. 
More specifically, our results suggest that for a given WOR 
protocol which solves $n$-consensus, the ability of the processes to make incremental partial agreements
 using safe-consensus objects to reach consensus, is encoded in the number of vertices and the connectivity of 
 specific subgraphs of $J_{n,m}$ for 
$2\leqslant m \leqslant n$. It is somehow surprising that Johnson graphs played 
such an important role in the proof of Theorem \ref{lemminimumKboxesFULLWOR}, 
these graphs and their properties are known to have applications in coding theory, specifically, Johnson graphs are closely related to \emph{Johnson schemes} \cite{DelLev98}. 
It is the first time that they appear in the scene of 
distributed systems. Our results might suggest some kind of connection 
between distributed systems and coding theory.

Also, in all our structural results of iterated protocols with safe-consensus, 
we see that 
connectivity of graphs appear again in the scene for the consensus task, this suggests that topology will play 
an important role to understand better the behaviour of protocols that use shared objects more 
powerful that read/write shared memory registers. The proofs developed to obtain the lower bound 
say that the connectivity of the topological structures given by the shared objects used by the 
processes, can affect the connectivity of the topology of the protocol complex \cite{HerlihyKR2013}.

With respect to the relationship between the iterated models introduced in this paper and 
the standard model of distributed computing 
\cite{Herlihy91,Herlihy90} extended with safe-consensus objects used in \cite{yehudagafni} to 
implement consensus using safe-consensus objects, we can conclude the following. Our results for 
the WRO and the OWR iterated models say that these two models are not equivalent to the standard 
model with safe-consensus (for task solvability), as consensus can be implement in the latter model 
\cite{yehudagafni}, but consensus cannot be implemented in the WRO and OWR iterated models. Still, an 
open problem related to these two iterated models remains unsolved: The characterization of their exact 
computational power. Another closely related problem is: Are these two models more powerful 
than the standard shared memory model of distributed computing ?

On the other hand, the relationship between the standard model extended with safe-consensus and the 
WOR iterated model remains unknown. Are these two models equivalent? We conjecture that the answer 
is yes.

An interesting question for future work is the relation between safe consensus and abort objects.
Hadzilacos and Toueg \cite{HVT_AbortableObjects} introduced objects that behave like consensus objects except 
that any operation may abort without taking effect if it is concurrent with another operation.

\bibliographystyle{unsrt}       

\bibliography{bib/biblio}           

\appendix

\section{Appendix: Results on subgraphs of Johnson graphs}\label{appendixProofJohnsonGraphs}

In this appendix, we prove Theorem \ref{NoConnUmNoConnzetaUm} and 
Lemma \ref{lemPartitionn2}, two results about subgraphs of Johnson 
graphs used to build the full proof of Theorem 
\ref{lemminimumKboxesFULLWOR}. The main result here is 
Theorem \ref{NoConnUmNoConnzetaUm}, which provides a new and interesting 
connection between Johnson graphs and iterated consensus protocols.
We will be using some combinatorial facts of graph theory and finite sets. 

\subsection{Subgraphs of Johnson graphs}\label{subsecSubgraphsJohnsonGraphs}

Our graph terminology is standard, see for example \cite{bollobas1998modern}, all the graphs that
we use are simple graphs. For any set $A$ and $E\subseteq 2^{A}$, we denote the union of
all the elements of every set in $E$ as $\bigcup E$. 
For $1\leqslant m \leqslant n$, the
\emph{Johnson graph} $J_{n,m}$ has as
vertex set all subsets of $\overline{n}$ of cardinality $m$, two vertices $b_1,b_2$ are adjacent if
and only if $\lvert b_1 \cap b_2 \rvert = m - 1$. Let $V_{n,m}=V(J_{n,m})$ 
and $U\subseteq V_{n,m}$, 
define the set $\zeta(U)$ as 
\begin{equation}\label{eqDefZetaU}
\zeta(U)= \{ c\cup d \mid c,d \in U \text{ and }  \lvert c \cap d \rvert = m - 1 \}. 
\end{equation}
Notice that each $f\in \zeta(U)$ has size $m+1$ because, if $f= c\cup d$ for $c,d \in U$, then
$\lvert c \rvert = \lvert d \rvert = m$ and it is known that $\lvert c \cap d \rvert = m - 1
\Leftrightarrow \lvert c \cup d \rvert = m + 1$. Thus $\zeta(U)\subseteq V_{n,m+1}$. For any
$v=0,\ldots,n-m$, the \emph{iterated $\zeta$-operator} $\zeta^v$ is given by
\begin{equation}\label{eqDefIteratedZetaU}
 \zeta^v(U)=\begin{cases}
         U & \text{if } v=0, \\
	 \zeta(\zeta^{v - 1}(U)) & \text{otherwise}.
        \end{cases}
\end{equation}
As $U\subseteq V_{n,m}$, we can check that $\zeta^v(U)\subseteq V_{n,m+v}$. A simple, but useful
property of the $\zeta$-operator is that
\begin{equation}\label{eqInclusionIteratedZetaU}
 \bigcup \zeta^v(U)\subseteq \bigcup U.
\end{equation}

\subsection{Preliminaries}\label{secAppendixConnectivityJnm}

We are ready to prove all the combinatorial results that we need.

\begin{lemma}\label{lemjohnsonsubgraph}
Let $U \subseteq V_{n,m}$ and $G=G\left[ U \right]$. The following properties are satisfied.
\begin{itemize}
 \item [(i)] If $G$ is connected, then $\bigl| \bigcup U \bigr| \leqslant m - 1 + \lvert U
\rvert$.
\item [(ii)] If $U_1,\ldots,U_r$ are connected components of $G$, then $\zeta(\bigcup_{i=1}^r
U_i)=\bigcup_{i=1}^r \zeta(U_i)$.
\end{itemize}
\end{lemma}

\begin{proof}
(i) can be proven easily using induction on $\lvert U \rvert$ and (ii) is an easy consequence of
the definitions, thus we omit the proof.
\end{proof}

\begin{lemma}\label{lemjohnsoninducedgraph2}
 Let $U\subseteq V_{n,m}$ with $\lvert U \rvert \leqslant n - m$. If $G\left[U\right]$ is connected
then $\bigcup U \neq \overline{n}$.
\end{lemma}

\begin{proof}
 Suppose that with the given hypothesis $\bigcup U = \overline{n}$. As $G\left[ U \right]$
is connected, we can use part (i) of Lemma \ref{lemjohnsonsubgraph} to obtain the inequality
$n\leqslant m + \lvert U \rvert -1$. But this implies that $\lvert U \rvert \geqslant n - m + 1$ and
this is a contradiction. Therefore $\bigcup U \neq \overline{n}$.
\end{proof}

%

%

%

\begin{lemma}\label{GconnectedGprimeConnectedandMore}
 Let $U\subset V_{n,m}$ 
with $\lvert U \rvert > 1$ 
and $G=G\left[ U \right]$ a connected subgraph of $J_{n,m}$.
Then the graph $G^\prime=G\left[\zeta(U) \right]$ is connected and $\bigcup \zeta(U) = \bigcup U$.
\end{lemma}

\begin{proof}
We show that $G^\prime$ is connected. If $G$ contains only two vertices, then $G^\prime$ contains
only one vertex and the result is immediate. So assume that $\lvert U \rvert > 2$ and take $b,c\in
\zeta(U)$. We need to show that there is a path from $b$ to $c$. We know that $b=b_1 \cup b_2$ and
$c=c_1 \cup c_2$ with $b_i,c_i\in U$ for $i=1,2$. As $G$ is connected, there is a path
\[ v_1=b_2,v_2,\ldots,v_q=c_2, \]
and we use it to build the following path in $G^\prime$,
\[ b=b_1\cup v_1,v_1 \cup v_2,\ldots, v_{q-1} \cup v_q,c=v_q \cup c_1, \]
thus $b$ and $c$ are connected in $G^\prime$, so that it is a connected graph. 

Now we prove that $\bigcup \zeta(U) = \bigcup U$. By Equation \eqref{eqInclusionIteratedZetaU},
$\bigcup \zeta(U) \subseteq \bigcup U$, so it only remains to prove the other inclusion. Let $x\in
\bigcup U$, there is a vertex $b\in U$ such that $x\in b$ and as $\lvert U \rvert > 1$ and $G$ is
connected, there exists another vertex $c\in U$ such that $b$ and $c$ are adjacent in $G$. Then
$b\cup c\in \zeta(U)$, thus
\[ x\in b \subset b\cup c \subset \bigcup \zeta(U), \]
gather that, $\bigcup U \subseteq \bigcup \zeta(U)$ and the equality holds. This concludes the
proof.
\end{proof}

\begin{lemma}\label{CharacterizationConnCompGZeta}
 Let $U\subseteq V_{n,m}$, $G=G\left[ U \right], G^\prime=G\left[ \zeta(U) \right]$ and
$\mathcal{U}$ be the set of connected components of the graph $G$. Then the following conditions
hold:
\begin{enumerate}
 \item For any connected component $V$ of $G^\prime$, there exists a set $\mathcal{O}\subseteq
\mathcal{U}$ such that $V=\zeta (\bigcup \mathcal{O})=\bigcup_{Z\in \mathcal{O}} \zeta(Z)$.
\item If $V,V^\prime$ are two different connected components of $G^\prime$ and
$\mathcal{O},\mathcal{O}^\prime\subseteq \mathcal{U}$ are the sets which fulfil property 1 for $V$
and $V^\prime$ respectively, then $\mathcal{O} \cap \mathcal{O}^\prime = \varnothing$.
\end{enumerate}
\end{lemma}

\begin{proof}
Part 2 is clearly true, so we only need to prove part 1. Define the graph
$\mathcal{H}=(V(\mathcal{H}), E(\mathcal{H}))$ as follows:
\begin{itemize}
 \item [] $V(\mathcal{H})=\mathcal{U}$;
\item [] Two vertices $Z, Z^\prime$ form an edge in $E(\mathcal{H})$ if and only if there exist
$b,c\in Z$ and $d,e \in Z^\prime$ such that
\begin{itemize}
\item $\lvert b \cap c \rvert = m - 1 \text{ and } \lvert d \cap e \rvert = m - 1$
\item $\bigl| (b \cup c) \cap (d \cup e) \bigr| \geqslant m$.
\end{itemize}
\end{itemize}
Roughly speaking, $\mathcal{H}$ describes for two connected components $Z,Z^\prime$ of $G$, 
whe\-ther
$\zeta(Z),\zeta(Z^\prime)$ lie in the same connected component of $G^\prime$ or not.
Let $V$ be a connected component of $G^\prime$, if $b \in V$, then $b=b_1
\cup b_2$, where $b_1,b_2$ are in some connected component $Z_b\subseteq U$ of $G$. In
$\mathcal{H}$, there exists a component $\mathcal{O}$ such that $Z_b\in \mathcal{O}$. By Lemma
\ref{GconnectedGprimeConnectedandMore}, $G\left[ \zeta(Z_b) \right]$ is connected and $b\in \zeta
(Z_b) \cap V$, so it is clear that $\zeta(Z_b) \subseteq V$. Suppose that $Z^\prime\in \mathcal{O}$
with $Z^\prime \neq Z_b$ and these components form an edge in $\mathcal{H}$. This means that
$G\left[ \zeta(Z_b) \right]$ and $G\left[ \zeta(Z^\prime) \right]$ are joined by at least one edge
in $G^\prime$, thus every vertex of $\zeta(Z^\prime)$ is connected with every vertex of
$\zeta(Z_b)$, and as $\zeta(Z_b)\subseteq V$, $\zeta(Z_b) \cup \zeta(Z^\prime)\subseteq V$. We can
continue this process with every element of the set $\mathcal{O} - \{ Z_b, Z^\prime \}$ to show that
\[ \zeta \bigl(\bigcup \mathcal{O} \bigr)=\bigcup_{Z\in \mathcal{O}} \zeta(Z) \subseteq V. \]
(The first equality comes from part (ii) of Lemma \ref{lemjohnsonsubgraph}). We prove the other
inclusion, if $V = \{ b \}$, we are done. Otherwise, let $c\in V - \{ b \}$, if $c\in \zeta(Z_b)$,
then $c\in \zeta \bigl( \bigcup \mathcal{O}\bigr)$. In case that $c\notin \zeta(Z_b)$, there must
exists some connected component $Z_c$ of $G$ such that $c\in \zeta(Z_c)$ and $Z_c\neq Z_b$. In
$G^\prime$, we can find a path $b=v_1,\ldots,v_q=c$ where $v_i\in V$ for $i=1,\ldots,q$. The set
$\mathcal{Q} \subset \mathcal{U}$ defined as
\[ \mathcal{Q} = \{ X\in \mathcal{U} \mid (\exists j) (1 \leqslant j \leqslant q \text{ and }  v_j 
\in
\zeta(X)) \}, \]
can be seen to have the property that every pair of vertices $X,X^\prime \in \mathcal{Q}$ are
connected by a path in $\mathcal{H}$. Since $Z_b,Z_c\in \mathcal{Q}$, they are connected in
$\mathcal{H}$ and as $Z_b \in \mathcal{O}$, then $Z_c\in \mathcal{O}$, gather that, $c\in \zeta
(\bigcup \mathcal{O})$. Therefore $V\subseteq \zeta(\bigcup \mathcal{O})$ and the equality
$V=\bigcup_{Z\in \mathcal{O}} \zeta(Z)$ holds. This proves part 1 and finishes the proof.
\end{proof}

\begin{lemma}\label{UnionConnCompInequalityGeneralized}
 Let $U\subseteq V_{n,m}$ and $\mathcal{U}$ be the set of connected components of $G\left[ U
\right]$. Then for all $s\in \{ 0,\ldots, n - m \}$ and $G^s=G\left[ \zeta^s(U) \right]$, the
following conditions are satisfied.
\begin{itemize}
 \item [{\bf H1}] For every connected component $V$ of the graph $G^s$, there exists a set
$\mathcal{O} \subseteq \mathcal{U}$ such that 
\begin{equation}\label{eqUnionConnCompInequalityGeneralized}
\bigl| \bigcup V \bigr| \leqslant m - 1 + \sum_{Z\in \mathcal{O}} \bigl| Z \bigr|.
\end{equation}
\item [{\bf H2}] If $V,V^\prime$ are two different connected components of $G^s$ and
$\mathcal{O},\mathcal{O}^\prime \subseteq \mathcal{U}$ are the sets which make true the inequality
given in \eqref{eqUnionConnCompInequalityGeneralized} for $V$ and $V^\prime$ respectively, then
$\mathcal{O} \cap \mathcal{O}^\prime = \varnothing$.
\end{itemize}
\end{lemma}

\begin{proof}
We prove the lemma by using induction on $s$. For the base case $s=0$, $G^0=G\left[ U \right]$, 
we use Lemma \ref{lemjohnsonsubgraph} and we are done. Suppose that for $0 \leqslant s < n - m$,
{\bf H1} and {\bf H2} are true. We prove the case $s+1$, let $W$ be a connected component of the
graph $G^{s+1}$. By part 1 of Lemma \ref{CharacterizationConnCompGZeta}, we know that there is a
unique set $\mathcal{Q}$ of connected components of $G^{s}$ such that
\[ W=\bigcup_{Q\in \mathcal{Q}} \zeta(Q) \]
and because $\zeta(Q) \neq \varnothing$, $\lvert Q \rvert > 1$, so that by Lemma
\ref{GconnectedGprimeConnectedandMore}, $\bigcup \zeta(Q) = \bigcup Q$ for all $Q\in \mathcal{Q}$,
thus it is true that
\begin{eqnarray*}
 \bigcup W & = & \bigcup \bigl(\bigcup_{Q\in \mathcal{Q}} \zeta\bigl(Q\bigr)\bigr)  \\
	   & = & \bigcup_{Q\in \mathcal{Q}} \bigl(\bigcup \zeta\bigl(Q \bigr)\bigr)  \\
	   & = & \bigcup_{Q\in \mathcal{Q}} \bigl(\bigcup Q \bigr).
\end{eqnarray*}
By the induction hypothesis, 
$\lvert \bigcup Q \rvert \leqslant m - 1 + \sum_{Z\in \mathcal{O}_Q} \lvert Z \rvert$, where
$\mathcal{O}_Q\subseteq \mathcal{U}$ for all $Q\in \mathcal{Q}$ and when  $Q \neq R$, $\mathcal{O}_Q
\cap \mathcal{O}_R=\varnothing$. With a simple induction on $\lvert \mathcal{Q} \rvert$, it is
rather forward to show that 
\[ \bigl| \bigcup W \bigr| \leqslant m - 1 + \sum_{Z\in \mathcal{O}} \bigl| Z \bigr|, \]
where $\mathcal{O}=\bigcup_{Q\in \mathcal{Q}} \mathcal{O}_Q$, gather that, property {\bf H1} is
fulfilled. We now show that condition {\bf H2} holds. If $W^\prime$ is another connected component
of $G^{s+1}$, then applying the above procedure to $W^\prime$ yields $W^\prime=\bigcup_{S \in
\mathcal{S}} \zeta(S)$, $\bigcup W^\prime = \bigcup_{S\in \mathcal{S}} \bigl(\bigcup S \bigr)$ and
$\bigl| \bigcup W^\prime \bigr| \leqslant m - 1 + \sum_{X \in \mathcal{O}^\prime}
\bigl| X \bigr|$, with $\mathcal{O}^\prime=\bigcup_{S\in \mathcal{S}} \mathcal{O}_S$. By 2 of Lemma
\ref{CharacterizationConnCompGZeta}, $\mathcal{Q} \cap \mathcal{S} = \varnothing$, so that if $Q\in
\mathcal{Q}$ and $S\in \mathcal{S}$, then $Q\neq S$ and the induction hypothesis tells us that
$\mathcal{O}_Q \cap \mathcal{O}_S=\varnothing$, thus
\begin{eqnarray*}
 \mathcal{O} \cap \mathcal{O}^\prime & = & \bigl(\bigcup_{Q\in \mathcal{Q}} \mathcal{O}_Q \bigr)
\cap \bigl(\bigcup_{S\in \mathcal{S}} \mathcal{O}_S \bigr) \\
	   & = & \bigcup_{(Q,S) \in \mathcal{Q} \times \mathcal{S}} \bigl( \mathcal{O}_Q
\cap \mathcal{O}_S \bigr)  \\
	   & = & \varnothing,
\end{eqnarray*}
and property {\bf H2} is satisfied, so that by induction we obtain the result.
\end{proof}

\subsection{Proofs of Theorem \ref{NoConnUmNoConnzetaUm} and Lemma \ref{lemPartitionn2}}\label{secAppendixFinalProofs}

Using all the results that we have developed in this appendix, 
we can now prove Theorem \ref{NoConnUmNoConnzetaUm}, 
(the ma\-in result of the combinatorial part of the proof of Theorem 
\ref{lemminimumKboxesFULLWOR}) and Lemma \ref{lemPartitionn2}, which is used to prove Lemma \ref{thm2boxesleqnm2}.\\

\noindent {\it Proof of Theorem \ref{NoConnUmNoConnzetaUm}}.
For a contradiction, suppose that $\zeta^{n - m}(U) \neq \varnothing$, then $G\left[ \zeta^{n -
m}(U) \right]$ is the graph $J_{n,n}$ and contains the unique connected component $C=\{ \overline{n}
\}$, thus by Lemma \ref{UnionConnCompInequalityGeneralized}, for some set $\mathcal{O}$ of
connected components of $G\left[ U \right]$,
\[ n = \bigl| \bigcup C \bigr| \leqslant m - 1 + \sum_{Z\in \mathcal{O}} \bigl| Z \bigr| \leqslant
m - 1 + \bigl| U \bigr|, \]
and we conclude that $\lvert U \rvert \geqslant n - m + 1$, a contradiction. So that $\zeta^{n -
m}(U) \neq \varnothing$ is impossible. Therefore $\zeta^{n - m}(U)$ has no elements.
\hspace*{\fill}$\square$\par\vspace{3mm}

The proof of Lemma \ref{lemPartitionn2} uses properties of 
subgraphs of $J_{n,2}$ and it is far easier to obain than the proof of
Theorem \ref{NoConnUmNoConnzetaUm}.\\

\noindent {\it Proof of Lemma \ref{lemPartitionn2}}. If $\bigcup U \neq \overline{n}$, then setting $A=\overline{n} - \{ i \}$ and $B=\{ i \}$ where 
$i\notin \bigcup U$ we are done. Otherwise, by Lemma \ref{lemjohnsoninducedgraph2}, the induced 
subgraph $G=G\left[ U \right]$ of $J_{n,2}$ is disconnected. There exists a partition $V_1,V_2$ of 
$V(G)=U$ with the property that there is no edge of $G$ from any vertex of $V_1$ to any vertex of 
$V_2$. Let 
\[ A=\bigcup V_1 \quad \text{and} \quad B = \bigcup V_2. \]
It is easy to show that $\overline{n}=A\cup B$ is the partition of $\overline{n}$ that we need.
\hspace*{\fill}$\square$\par\vspace{3mm}

\end{document}